\PassOptionsToPackage{unicode}{hyperref}
\PassOptionsToPackage{hyphens}{url}
\PassOptionsToPackage{dvipsnames,svgnames,x11names,table}{xcolor}
\documentclass[
  12pt]{article}

\usepackage{amsmath,amssymb}
\usepackage{iftex}
\ifPDFTeX
  \usepackage[T1]{fontenc}
  \usepackage[utf8]{inputenc}
  \usepackage{textcomp} % provide euro and other symbols
\else % if luatex or xetex
  \usepackage{unicode-math}
  \defaultfontfeatures{Scale=MatchLowercase}
  \defaultfontfeatures[\rmfamily]{Ligatures=TeX,Scale=1}
\fi
\usepackage{lmodern}
\ifPDFTeX\else  
\fi
\IfFileExists{upquote.sty}{\usepackage{upquote}}{}
\IfFileExists{microtype.sty}{% use microtype if available
  \usepackage[]{microtype}
  \UseMicrotypeSet[protrusion]{basicmath} % disable protrusion for tt fonts
}{}
\makeatletter
\@ifundefined{KOMAClassName}{% if non-KOMA class
  \IfFileExists{parskip.sty}{%
    \usepackage{parskip}
  }{% else
    \setlength{\parindent}{0pt}
    \setlength{\parskip}{6pt plus 2pt minus 1pt}}
}{% if KOMA class
  \KOMAoptions{parskip=half}}
\makeatother
\usepackage{xcolor}
\makeatletter
\ifx\paragraph\undefined\else
  \let\oldparagraph\paragraph
  \renewcommand{\paragraph}{
    \@ifstar
      \xxxParagraphStar
      \xxxParagraphNoStar
  }
  \newcommand{\xxxParagraphStar}[1]{\oldparagraph*{#1}\mbox{}}
  \newcommand{\xxxParagraphNoStar}[1]{\oldparagraph{#1}\mbox{}}
\fi
\ifx\subparagraph\undefined\else
  \let\oldsubparagraph\subparagraph
  \renewcommand{\subparagraph}{
    \@ifstar
      \xxxSubParagraphStar
      \xxxSubParagraphNoStar
  }
  \newcommand{\xxxSubParagraphStar}[1]{\oldsubparagraph*{#1}\mbox{}}
  \newcommand{\xxxSubParagraphNoStar}[1]{\oldsubparagraph{#1}\mbox{}}
\fi
\makeatother

\usepackage{longtable,booktabs,array}
\usepackage{calc} % for calculating minipage widths
\usepackage{etoolbox}
\makeatletter
\patchcmd\longtable{\par}{\if@noskipsec\mbox{}\fi\par}{}{}
\makeatother
\IfFileExists{footnotehyper.sty}{\usepackage{footnotehyper}}{\usepackage{footnote}}
\makesavenoteenv{longtable}
\usepackage{graphicx}
\makeatletter
\def\maxwidth{\ifdim\Gin@nat@width>\linewidth\linewidth\else\Gin@nat@width\fi}
\def\maxheight{\ifdim\Gin@nat@height>\textheight\textheight\else\Gin@nat@height\fi}
\makeatother
\setkeys{Gin}{width=\maxwidth,height=\maxheight,keepaspectratio}
\makeatletter
\def\fps@figure{htbp}
\makeatother

\makeatletter
\@ifpackageloaded{caption}{}{\usepackage{caption}}
\AtBeginDocument{%
\ifdefined\contentsname
  \renewcommand*\contentsname{Table of contents}
\else
  \newcommand\contentsname{Table of contents}
\fi
\ifdefined\listfigurename
  \renewcommand*\listfigurename{List of Figures}
\else
  \newcommand\listfigurename{List of Figures}
\fi
\ifdefined\listtablename
  \renewcommand*\listtablename{List of Tables}
\else
  \newcommand\listtablename{List of Tables}
\fi
\ifdefined\figurename
  \renewcommand*\figurename{Figure}
\else
  \newcommand\figurename{Figure}
\fi
\ifdefined\tablename
  \renewcommand*\tablename{Table}
\else
  \newcommand\tablename{Table}
\fi
}
\@ifpackageloaded{float}{}{\usepackage{float}}
\floatstyle{ruled}
\@ifundefined{c@chapter}{\newfloat{codelisting}{h}{lop}}{\newfloat{codelisting}{h}{lop}[chapter]}
\floatname{codelisting}{Listing}

\makeatother
\makeatletter
\@ifpackageloaded{caption}{}{\usepackage{caption}}
\@ifpackageloaded{subcaption}{}{\usepackage{subcaption}}
\makeatother

\ifLuaTeX
  \usepackage{selnolig}  % disable illegal ligatures
\fi
\usepackage{natbib}
\shortcites{ghcnd_data,Menne2012}  
\usepackage{bookmark}

\IfFileExists{xurl.sty}{\usepackage{xurl}}{} % add URL line breaks if available
\hypersetup{
  pdftitle={Spatial scale-aware tail dependence modeling for high-dimensional spatial extremes},
  pdfauthor={Muyang Shi; Likun Zhang; Mark D. Risser; Benjamin A. Shaby},
  pdfkeywords={Asymptotic dependence, Nonstationary, Scale mixture, Spatial extremes},
  colorlinks=true,
  linkcolor={blue},
  filecolor={Maroon},
  citecolor={Blue},
  urlcolor={Blue},
  pdfcreator={LaTeX via pandoc}}

\newcommand{\anon}{1}

\usepackage{lmodern}
\usepackage{amsfonts}
\usepackage{graphicx,psfrag,epsf}
\usepackage{bbm}
\usepackage{mathtools}
\usepackage[shortlabels]{enumitem}
\usepackage{mdframed}
\usepackage{siunitx}
\usepackage[normalem]{ulem}
\usepackage{float}
\usepackage{lineno}
\usepackage{bm}
\usepackage{multirow}
\usepackage{amsthm}
\usepackage{setspace}
\usepackage{changepage}
\usepackage{booktabs}
\usepackage{array}

\usepackage{etoolbox}
\newcommand{\zerodisplayskips}{%
  \setlength{\abovedisplayskip}{0pt}%
  \setlength{\belowdisplayskip}{0pt}%
  \setlength{\abovedisplayshortskip}{0pt}%
  \setlength{\belowdisplayshortskip}{0pt}}
\appto{\normalsize}{\zerodisplayskips}
\appto{\small}{\zerodisplayskips}
\appto{\footnotesize}{\zerodisplayskips}

\newcommand{\bs}{\boldsymbol{s}}
\newcommand{\bw}{\boldsymbol{w}}
\newcommand{\btheta}{\boldsymbol{\theta}}
\newcommand{\given}{\, | \,}
\newcommand{\pr}{\text{Pr}}
\newcommand{\trans}{^\text{T}}
\newcommand{\iid}{\stackrel{\text{iid}}{\sim}}
\newcommand{\stkout}[1]{\ifmmode\text{\sout{\ensuremath{#1}}}\else\sout{#1}\fi}
\newcommand{\indep}{\perp \!\!\! \perp}
\newcommand\numberthis{\addtocounter{equation}{1}\tag{\theequation}}

\newtheorem{theorem}{Theorem}%  meant for continuous numbers
\newtheorem{proposition}[theorem]{Proposition}%
\newtheorem{example}{Example}%
\newtheorem{remark}{Remark}%

\newtheorem{lemma}[theorem]{Lemma}
\newtheorem{illustration}{Illustration}

\definecolor{amber(sae/ece)}{rgb}{1.0, 0.49, 0.0}
\newcommand{\revise}[1]{\textcolor{black}{#1}}

\makeatletter
\renewcommand*\env@matrix[1][\arraystretch]{%
  \edef\arraystretch{#1}%
  \hskip -\arraycolsep
  \let\@ifnextchar\new@ifnextchar
  \array{*\c@MaxMatrixCols c}}
\makeatother

\begin{document}

\def\spacingset#1{\renewcommand{\baselinestretch}%
{#1}\small\normalsize} \spacingset{1}

%%%%%%%%%%%%%%%%%%%%%%%%%%%%%%%%%%%%%%%%%%%%%%%%%%%%%%%%%%%%%%%%%%%%%%%%%%%%%%

\if1\anon
{
  \title{\bf Spatial scale-aware tail dependence modeling for high-dimensional spatial extremes}
  \author{Muyang Shi\thanks{
    The authors gratefully acknowledge the support of NSF grants DMS-2308680, DMS-2001433, and DMS-2308679, which made this research possible. The authors acknowledge the use of ChatGPT version 5 for grammar proofreading and assistance with coding language syntax.}\hspace{.2cm}\\
    Department of Statistics, Colorado State University\\
    Likun Zhang \\
    Department fo Statistics, University of Missouri \\
    Mark D. Risser \\
    Climate and Ecosystem Sciences, Lawrence Berkeley National Lab \\
    and \\
    Benjamin A. Shaby \\
    Department of Statistics, Colorado State University}
  \maketitle
} \fi

\if0\anon
{
  \bigskip
  \bigskip
  \bigskip
  \begin{center}
    {\LARGE\bf Spatial scale-aware tail dependence modeling for high-dimensional spatial extremes}
\end{center}
  \medskip
} \fi

\bigskip
\begin{abstract}
Extreme events over large spatial domains may exhibit highly heterogeneous tail dependence characteristics, yet most existing spatial extremes models yield only one dependence class over the entire spatial domain. \revise{To accurately characterize dependence in extreme events}, we propose a mixture model that achieves flexible dependence properties and allows high-dimensional inference \revise{($\sim600$ spatial locations in our data example)} for extremes of spatial processes. We modify the popular random scale construction that multiplies a Gaussian random field by a single radial variable; we allow the radial variable to vary smoothly across space and add non-stationarity to the Gaussian process. As the level of extremeness increases, this single model exhibits both asymptotic independence at long ranges and either asymptotic dependence or independence at short ranges.  We make joint inference on the dependence model and a marginal model using a copula approach within a Bayesian hierarchical model.
Three different simulation scenarios show close to nominal frequentist coverage rates. Lastly, we apply the model to a dataset of extreme summertime precipitation over the central United States. We find that the joint tail of precipitation exhibits non-stationary dependence structure that cannot be captured by limiting extreme value models or current state-of-the-art sub-asymptotic models.
\end{abstract}

\noindent%
\vfill

\newpage
\spacingset{1.8} % DON'T change the spacing!

\section{Introduction}\label{sec-intro}

Accurately characterizing the spatial extent and intensity of large-scale extreme precipitation events is a critical factor in infrastructure planning, risk mitigation, and adaptation \citep[][]{field2014summary}, particularly in our changing global climate \citep[][]{milly2008stationarity}.
In this paper, we propose a \revise{flexible} tail dependence model that modulates a radial variable to introduce varying levels of dependence strength both locally and across long spatial ranges, and incorporates a nonstationary covariance function to further account for spatial heterogeneity. Our sub-asymptotic framework summarizes the tail dependence structure in a parsimonious manner and allows for efficient inference from large spatial data sets. Therefore, we can appropriately model the varying scale and intensity of extreme precipitation events.

%=============================
% Point-of-paragraph: 
% Setup
%=============================
Let $\{Y(\bs): \bs\in \mathcal{S}\subseteq \mathbb{R}^2\}$ be the stochastic process of interest, and denote its realization at the $j$th location $\bs_j$ by $Y_j=Y(\bs_j)$, $j=1,\ldots, D$. In a spatial extremes context, a common approach is to model the marginal distributions of each $Y_j$ (denoted by $F_j$) separately from their dependence structure via the copula. 
% The copula, which is defined by the distribution function of $(X_1,\ldots, X_D)=(F_1(Y_1), \ldots, F_D(Y_D))$, fully describes the dependence structure of the process $\{Y(\bs)\}$. 
After transforming the data to the uniform scale via the probability integral transform, a carefully chosen spatial dependence model can be used to %accurately
characterize the tail dependence properties of the copula while accounting for spatial non-stationarity. 

%=============================
% Point-of-paragraph: 
% Summary of "traditional" spatial extremes methods
%=============================
Unfortunately, previous spatial extremes models have dependence structures that are too rigid for large domains and thus lead to underestimation or overestimation of the spatial extent and intensity of extreme events \citep{huser2016non}. Statistical modeling for spatial extremes traditionally uses either max-stable processes (when considering block maxima) or generalized Pareto processes (for exceedances of a high threshold) due to their asymptotic justifications \citep{davison_geostatistics_2013, huser_advances_2022}. The copulas of both processes exhibit stability properties, wherein their dependence structures are invariant to the operations of taking maxima over larger blocks or conditioning on a higher threshold. In the bi-variate case, this means that the upper dependence measure of the process between any two locations $\bs_i$ and $\bs_j$
\begin{equation}\label{eqn:chi_measure}
        \chi_{ij}(u) = P(F_j(Y_j)>u\given F_i(Y_i)>u)
\end{equation}
will always have a non-zero limit as the quantile level $u \rightarrow 1$, a property generally referred to as asymptotic dependence (AD). Max-stable and generalized Pareto models also have the property that $\chi_{ij}(u)$ becomes independent of $u$ for increasingly extreme events. However, for sub-asymptotic modeling, both of these properties are often violated empirically: estimates of \eqref{eqn:chi_measure} tend to decrease with $u$ for many observed environmental processes, including high-frequency wave height data \citep{huser2019modeling}, daily fire weather indices \citep{zhang2021hierarchical}, annual maximum temperature \citep{zhang2024}, and extreme seasonal precipitation (see Figure %\ref{fig:emp_chi_moving_windows}).
\ref{fig:application_chi_select}). Furthermore, weakening $\chi_{ij}(u)$ in the empirical estimates as $u\rightarrow 1$ may eventually result in asymptotic independence (AI), i.e., $\lim_{u\rightarrow 1}\chi_{ij} (u)=0$. Of course,  empirically weakening $\chi_{ij}(u)$ as $u$ approaches its observed upper bound does not necessarily correspond to asymptotic independence since estimates of $\chi_{ij}(u)$ for large $u$ have very large uncertainty and may not suggest convergence in the limit. 

\begin{figure}
    \centering
    \includegraphics[width=0.65\textwidth]{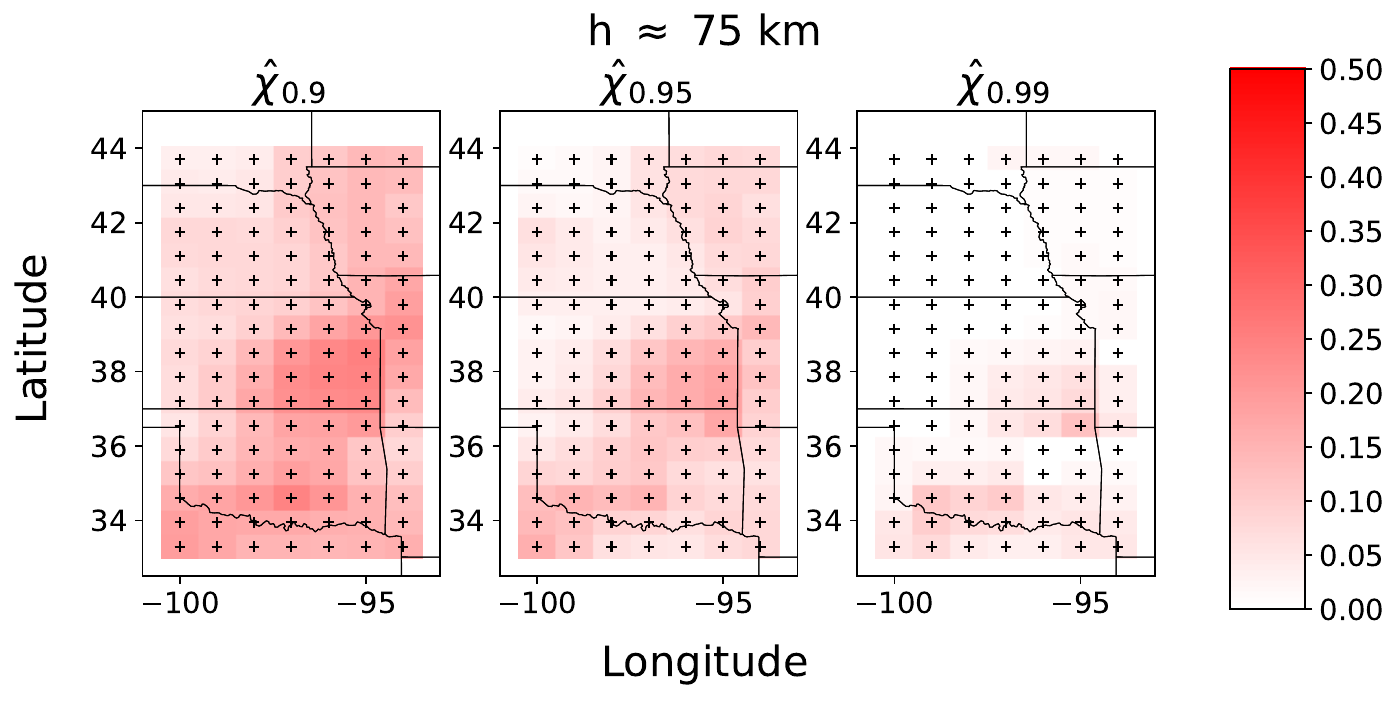}
    \caption{Empirical estimates of $\hat{\chi}_{h}(u)$ from 75 years of summertime maximum daily precipitation of central US within moving windows centered at 119 locations (marked as white `+'), using all pairs with separating distance $\approx 75$km, at three quantile levels $u=0.9, 0.95, 0.99$. The local tail dependence behavior appears to vary across the region. Further exploration is shown in Figure \ref{fig:application_chi} in Section \ref{sec:application}.}
    \label{fig:application_chi_select}
\end{figure}

%=============================
% Point-of-paragraph: 
% Summary of methods for which one can infer AI or AD from the data
%=============================
To avoid having to specify the dependence class before analyzing a specific data set, a series of flexible \revise{random scale mixture models} have been developed to encompass both AI and AD dependence classes \citep[e.g.,][]{huser2017bridging, huser2019modeling}. The general form of the spatial dependence model in these approaches is 
\begin{equation}\label{eqn:scalemix_model}
    X(\bs) = R\cdot g(Z(\bs)),\; R\given \theta_R \sim F_R,
\end{equation}
in which $R$ is a random scaling factor whose distribution function $F_R$ is parameterized by the vector $\theta_R$ containing a parameter $\phi$ which controls its tail heaviness, $g(\cdot)$ is a suitable link function, and $\{Z(\bs)\}$ is a stationary spatial process exhibiting asymptotic independence at any two locations. Intuitively, 
% $R$ is a completely dependent spatial process 
\revise{the scale factor $R$ can be thought of as a completely dependent spatial process in that it takes on a single value regardless of spatial location.  It}
amplifies the co-occurrences of extreme events in $\{Z(\bs)\}$, and the tail behavior of $X(\bs)$ is controlled by the relative rate of tail decay between $R$ and $\{g(Z(\bs))\}$ \citep{engelke2019extremal}. Consequently, varying a single tail index parameter $\phi$ allows  random scale mixture models to be either asymptotically independent or asymptotically dependent, with the transition occurring smoothly in the interior of the parameter space. In spite of this desirable property, inference for random scale mixture models is challenging. To mitigate computational burden while preserving tail dependence properties, \citet{zhang2021hierarchical} proposed the addition of a nugget effect which facilitates likelihood calculations for large data sets. Nevertheless, the random scale mixture models in \eqref{eqn:scalemix_model} may lack sufficient flexibility because, for a given value of $\phi$, they can only exhibit one dependence class, at all spatial lags, for the entire spatial domain. Coupled with the stationarity of the latent process $\{Z(\bs)\}$, this means that the dependence measure for any two locations is also stationary (i.e., $\chi_{ij}(u)\equiv \chi_{\lVert\bs_i - \bs_j\rVert}(u)$). Such assumptions of stationarity in both the dependence class of $X(\cdot)$ and the length scales of $Z(\cdot)$ are convenient but are rarely appropriate for heterogeneous real-world data sets.  The assumption that the tail dependence is of the same asymptotic class at all spatial lags seems unintuitive for large domains, where a process may have very strong tail dependence at short distances but weak tail dependence at long distances.  \revise{We therefore seek tail dependence models that are ``scale aware'' in that they are AI for observations separated by large distances and either AI or AD for observations separated by short distances.}

%=============================
% Point-of-paragraph: 
% Summary of the literature on nonstationary tail dependence
%=============================
To introduce non-stationarity in the tail dependence structure, \citet{castro2020local} let the dependence parameter $\phi$ be a spatial process $\phi(\bs)$ that alters the amplifying effect of $R$ over space. Their specification is in the form of random location mixture, which is equivalent to a random scale mixture through an exponential marginal transformation:
\begin{equation}\label{eqn:castro_model}
    X(\bs) = \phi^{-1}(\bs)R + Z(\bs), \; \bs\in\mathcal{S}.
\end{equation}
They furthermore specify $\{Z(\bs)\}$ to be a zero-mean Gaussian process with a nonstationary Mat\'{e}rn covariance function \citep{paciorek2006spatial}, resulting in a non-stationary dependence measure $\chi_{ij}$ for locations $\bs_i$ and $\bs_j$.  While this model yields much more flexible local dependence properties, $\chi_{ij}$ is always strictly positive unless $\max\{\phi_{i},\phi_j\}=\infty$; i.e., this model still exhibits asymptotic dependence everywhere, even for locations that are very far apart. 
% \begin{equation*}
%     \chi_{ij}=\lim_{u\rightarrow 1}\chi_{ij}(u)  = 2[1-\Psi\{(\phi_i^2+\phi_j^2-2\rho_{ij}\phi_i\phi_j)^{1/2}/2\}]
% \end{equation*}
% (here, $\Psi$ denotes the Gaussian distribution function). While this model yields much more flexible local dependence properties, $\chi_{ij}$ is always strictly positive unless $\max\{\phi_{i},\phi_j\}=\infty$; i.e., this model still exhibits asymptotic dependence everywhere, even for locations that are very far apart. 

To address this limitation, \citet{hazra2021realistic} generalizes the \citet{huser2017bridging} model by instead allowing $R$ to be a spatial process and $\phi$ a spatial constant:
\begin{equation}\label{eqn:huser_model}
    X(\bs) = R(\bs)Z(\bs), \; \bs\in\mathcal{S},
\end{equation}
where $Z(\bs)$ is a zero-mean isotropic Gaussian process, and $R(\bs)$ is a kernel-weighted sum of univariate random effects.  The kernels live in same spatial domain as the data and are compact, resulting in AI of the $X(\bs)$ process for locations sufficiently far apart (see our Theorem \ref{thm:dependence_properties} for our analogous result).  Their random effects distribution contains a tail parameter $\phi \geq 0$ such that $\phi=0$ gives AD for locations close enough to share a common kernel.  The model in \citet{hazra2021realistic} represents a major step forward in the flexibility of spatial extremes methods in that it combines local AD with long-range AI. However, we improve upon their approach in several ways.  First, we use a different random effects distribution, which allows the transition between local AI and AD to occur in the interior of its parameter space.  In contrast, whereas we are able to estimate the local tail dependence class, \citet{hazra2021realistic} only consider $\phi=0$ in their analysis, which fixes the dependence class to be AD at short ranges. For highly heterogeneous weather processes like extreme precipitation, it is often far from clear that local asymptotic dependence is always realistic, as seen in Figure \ref{fig:application_chi_select}.  Second, our random effects specification also permits a richer set of analytical results than is contained in \citet{hazra2021realistic}, including results on sub-asymptotic tail dependence in the local AI case.  Furthermore, we allow the tail parameter $\phi$ to vary in space, similarly to \citet{castro2020local}.  This can be important, as empirical estimates of the dependence class for our precipitation data suggests that it may exhibit both local AI and local AD in different parts of the domain; see again Figure \ref{fig:application_chi_select}.

Finally, the way \citet{hazra2021realistic} handle marginal fitting is computationally convenient but induces a few modeling disadvantages which we avoid.  Specifically, they perform location-specific standardization of the data before model fitting, then assume that the standardized observations follow a location-scale transformation of their dependence model $X(\bs)$.  Such a scheme does not allow for spatial prediction, since the site-specific standardization leaves no way to ``de-standardize'' any model-predicted $X(\bs)$ at un-observed locations. More worryingly, assuming the data follow a location-scale transformation of $X(\bs)$ forces the marginal tail weight of the observations to coincide with the marginal weight behavior of the dependence model $X(\bs)$.  This inextricably couples the marginal and dependence properties, a situation that we would like to avoid. In contrast, we perform full joint inference on the marginal and dependence models using a copula approach.  While this requires a good deal more implementation effort, it allows natural spatial prediction and prevents characteristics of the marginal tail weight from bleeding into the fit of the dependence parameters.

\citet{majumder2024modeling} also use a construction like \eqref{eqn:huser_model}, but $R(\bs)$ is specified as a max-stable process.  This yields short-range AD and long-range AI, but their model presents challenges.  First, analytical results are elusive, so tail dependence properties must be confirmed through simulation.  More importantly, likelihoods are unavailable, so \citet{majumder2024modeling} use a Vecchia-type approximation \citep{vecchia1988estimation} along with neural net emulators to perform approximate Bayesian inference.

\citet{wadsworth2022higher} use a completely different approach to obtain flexible models capable of short-range AD and long-range AI.  They specify their model conditionally on observing a large value at a single reference location.  Since there is usually no natural location to choose as the reference, \citet{wadsworth2022higher} form composite likelihoods by choosing several sites in turn, and summing the log likelihoods obtained from each.  This has the effect of fitting a model that requires conditioning on a single site, but uses the composite likelihood to do a kind of averaging across different choices of conditioning sites when fitting model parameters.  This approach is an important advancement that generates very flexible tail dependence characteristics.  However, because the resultant fits do not correspond to any well-defined joint probability model, interpretation is challenging.

\section{Model}\label{sec:model}
\subsection{Construction}\label{sec:model_construction}

Here, we combine the desirable features of \eqref{eqn:castro_model} and \eqref{eqn:huser_model} by allowing both $R(\bs)$ and $\phi(\bs)$ to vary spatially by using a mixture component representation for $R(\bs)$ and a spatially-varying tail index $\phi(\bs)$ in order to introduce more flexible local and long-range tail dependence behaviors:
\begin{equation}\label{eqn:model}
    X(\bs)=R(\bs)^{\phi(\bs)}g(Z(\bs)),
\end{equation}
where $\{Z(\bs)\}$ is a spatial process with hidden regular variation \citep{ledford1996statistics, heffernan2005hidden}.  The hidden regular variation assumption ensures
\begin{equation}\label{eqn:eta_measure}
        P(F^Z_i(Z_i)>u,F^Z_j(Z_j)>u)=L_Z\{(1-u)^{-1}\}(1-u)^{1/\eta^Z_{ij}},
\end{equation}
where $\eta^{Z}_{ij}\in (1/2,1)$ so that $(Z_i, Z_j)$ exhibits asymptotic independence, $F^Z_{j}$ is the univariate distribution function of $Z_j$, $j=1,\ldots, D$, and $L_Z$ is a slowly varying function, i.e., $L_Z(tz)$ $/$ $L_Z(z)$ $\rightarrow 1$ as $z \rightarrow \infty$ for all fixed $t > 0$. The coefficient $\eta^{Z}_{ij}$ was termed the \emph{coefficient of tail dependence} for $\{Z(\bs)\}$ between $\bs_i$ and $\bs_j$ by \citet{ledford1996statistics}, and it complements \eqref{eqn:chi_measure} in the case of asymptotic independence. 

In \eqref{eqn:model}, both the scaling factor and the tail index vary across space. This makes the model much more flexible but also adds considerable complexity to both the theoretical analysis and the computations.  The theoretical tools used to analyze models where either the scaling factor \citep{hazra2021realistic} or the tail index \citep{castro2020local}, but not both, vary spatially, will not suffice here.
However, if the random factors are weighted averages of Stable random variables \citep{nolan2020univariate}, the analysis is simplified considerably. Let
\begin{equation}\label{eqn:scaling_process}
    R(\bs)=\sum_{k=1}^K w_k(\bs, r_k) S_{k},\;S_{k}\stackrel{\text{indep}}{\sim} \text{Stable}\left(\alpha,1,\gamma_k,
    \delta\right),
\end{equation}
in which $w_k(\bs, r_k)$ is some compactly supported basis function over $\mathbb{R}^2$ centered at the $k$th knot with radius $r_k$, $k=1,\ldots, K$. Also, $\sum_{k=1}^K w_k(\bs;r_k) = 1$ for all $\bs\in\mathcal{S}$. Properties of the univariate and joint distribution of the constructed $\{X(\bs)\}$ and interpretation of its parameters are discussed in Section \ref{sec:univariate dependence model} and \ref{sec:joint dependence model}. The link function $g(\cdot)$ transforms the margins of $\{Z(\bs)\}$ to the margins of a type II Pareto distribution whose distribution function is 
\begin{equation}\label{eqn:type_II}
    F^{\text{Pareto}}(x)=\{1-(1+x-\delta)^{-1}\}\mathbbm{1}(x\geq \delta).
\end{equation}
We provide justifications for these distributional choices in Section \ref{sec:computation}.

\subsection{Univariate distribution of the dependence model} \label{sec:univariate dependence model}

The key property of the Stable distribution  which allows us to assess the tail behavior of the dependence model $X(s)$ in (\ref{eqn:model}) is that it is closed under convolution; % sums of $\alpha$-Stable variables (Stable variables with concentration parameter $\alpha$) are still $\alpha$-Stable. 
if $S_k$ $\stackrel{\text{indep}}{\sim}$ $\text{Stable}(\alpha,$ $\beta_k,$ $\gamma_k,$ $\delta_k)$ and constant $w_k\geq 0$ for $k=1,\ldots, K$, then
\begin{equation}
\begin{split}
    \sum_{k=1}^K w_{k} S_k 
     \sim &\text{Stable}(\alpha,\bar{\beta},\bar{\gamma},\bar{\delta})\\
\end{split}
\end{equation}
with $\bar{\gamma} = \{\sum_{k=1}^K(w_{k}\gamma_k)^\alpha\}^{1/\alpha}$, $\bar{\beta} =\sum_{k=1}^K\beta_k(w_{k}\gamma_k)^\alpha/\bar{\gamma}^{\alpha}$ and $\bar{\delta}=\sum_{k=1}^K w_{k}\delta_k$. To be able to examine the joint distribution of $(X_i, X_j)$ for model \eqref{eqn:model}, it is desirable for the mixture to have the same distributional support and rate of tail decay as each $S_k$. Thus in \eqref{eqn:scaling_process}, we fixed $\beta_k\equiv 1$, $\delta_k\equiv \delta$ while imposing the constraint $\sum_{k=1}^K w_k=1$. As a result, 
$\bar{\beta} =1$ and $\bar{\delta}=\delta$, which means the univariate support of the $\{R(\bs)\}$ is $[\delta,\infty)$ everywhere.  We also have tail-stationarity for $\{R(\bs)\}$ because $P(R(\bs)>x)\sim 2\bar{\gamma}^{\alpha}(\bs)C_{\alpha} x^{-\alpha}$ for all $\bs\in \mathcal{S}$ as $x\rightarrow\infty$, where
$\bar{\gamma}(\bs)=[\sum_{k=1}^K\{w_{k}(\bs,r_k)\gamma_k\}^\alpha]^{1/\alpha}$ and $C_{\alpha} = \Gamma(\alpha)\sin(\alpha\pi/2)/\pi$.

With the understanding of the distributions of $\{R(\bs)\}$, the tail behavior of the univariate distributions of the dependence model \eqref{eqn:model} can be assessed.
To avoid clutter, we denote $X_j := X(\bs_j)$, $R_j := R(\bs_j)$, $\phi_j := \phi(\bs_j)$, $\bar{\gamma}_j := \bar{\gamma}(\bs_j)$, and  $W_j := g(Z(\bs_{j}))$, $j = 1, \dots, D$.
Proposition \ref{prop:marginal_distr} below shows that the marginal tail behavior of the process $X(\bs_j)$ depends on the spatially varying tail parameter $\phi_j$; the (transformed) Gaussian part dominates when $\phi_j$ is small, while the heavy-tailed $R_j$ dominates when $\phi_j$ is large, and $\phi_j = \alpha$ marks the transition point between the two.

\singlespacing
\begin{proposition}\label{prop:marginal_distr}
The univariate distribution of the process \eqref{eqn:model} at a location $\bs_j$ satisfies
\begin{equation}\label{eqn:marginal}
    1-F_j(x)=\Pr \{R_j^{\phi_j}g(Z_j)>x\}\sim 
    \begin{dcases}
    \revise{\mathbb{E}}(R_j^{\phi_j})x^{-1},&\text{ if }0\leq\phi_j<\alpha,\\
    \frac{2C_\alpha\bar{\gamma}^\alpha_{j}}{1-\alpha/\phi_j} x^{-\frac{\alpha}{\phi_j}},&\text{ if }\phi_j>\alpha,\\
    2C_{\alpha}\bar{\gamma}^\alpha_{j} x^{-1}\log x,&\text{ if }\phi_j=\alpha,
    \end{dcases}
\end{equation}
as $x\rightarrow \infty$, where $C_\alpha=\Gamma(\alpha)\sin(\alpha\pi/2)/\pi$.
\end{proposition}
% \BAS{Move this remark to the appendix?}
\doublespacing
% \begin{remark}
% By Theorem 3.8 and Lemma 3.12 of \citet{nolan2020univariate}, we can write out the exact form of the fractional lower order moment of $R_j\sim \text{Stable}(\alpha,1,\bar{\gamma}_{j},\delta)$ if $\delta=0$:
% \begin{equation*}
%     E(R_j^{\phi_j})=\bar{\gamma}_{j}^{\phi_j} \cos^{-\frac{\phi_j}{\alpha}}\left(\frac{\pi\alpha}{2}\right)\frac{\Gamma(1-\phi_j/\alpha)}{\Gamma(1-\phi_j)}.
% \end{equation*}
% \end{remark}

\subsection{Joint distribution of the dependence model}\label{sec:joint dependence model}

We start by considering the joint behavior of two scaling variables, $R_i$ and $R_j$ (defined as kernel-weighted linear combinations of iid Stable random variables in \eqref{eqn:scaling_process}) at two locations $\bs_i$ and $\bs_j$. For notational simplicity, we denote $w_{kj}=w_k(\bs_j, r_k)$, $k=1,\ldots,K$, and write $\bw_j=(w_{1j},\ldots, w_{Kj})$ and $\mathcal{C}_j=\{k:w_{kj}\neq 0,k=1,\ldots,K\}$,  $j=1,\ldots, D$.  We require that any location $\bs\in \mathcal{S}$ is covered by at least one basis function, thus $\mathcal{C}_j$ cannot be empty for any $j$.
The following proposition describes the dependence between two scaling variables $R_i$ and $R_j$:

\singlespacing
\begin{proposition}\label{prop:R_joint_nonzero}

\begin{enumerate}[(a)]
    \item\label{lem:R_joint} If $\mathcal{C}_i\cap\mathcal{C}_j=\emptyset$, $R_i$ and $R_j$ are independent and $\Pr(R_i>x, R_j>x)$ $\sim$ $4C_\alpha^2\left\{\bar{\gamma}^\alpha_{i}\right\}\left\{\bar{\gamma}^\alpha_{j}\right\}$ $x^{-2\alpha}$ as $x\rightarrow \infty$. If $\mathcal{C}_i\cap\mathcal{C}_j\neq\emptyset$, then
\begin{equation}\label{eqn:R_joint}
    \Pr(R_i>x,R_j>x)=2C_\alpha C_K(\bw_i,\bw_j,\boldsymbol{\gamma})x^{-\alpha},
\end{equation}
where $C_K(\bw_i,\bw_j,\boldsymbol{\gamma})=\sum_{k=1}^K w_{k,\wedge}^\alpha\gamma_k^\alpha$ with $w_{k,\wedge}=\min(w_{ki},w_{kj})$,  $k=1,\ldots,K$.

\item\label{lem:R_max_min} If $\mathcal{C}_i\cap\mathcal{C}_j\neq \emptyset$, we have for two positive constants $\phi_i<\phi_j$ and any $c_i,\;c_j>0$
\begin{equation}\label{eqn:R_max_min} 
        \Pr(\min(c_iR_i^{\phi_i},c_jR_j^{\phi_j})>x)\sim d_\wedge x^{-\alpha/\phi_i},\;
        \Pr(\max(c_iR_i^{\phi_i},c_jR_j^{\phi_j})>x)\sim d_\vee x^{-\alpha/\phi_j},
\end{equation}
where \revise{$d_\wedge$ and $d_\vee$ are defined to satisfy, respectively,}
\begin{equation*}
    \begin{split}
        2C_\alpha\sum_{k\in\mathcal{C}_i\cap\mathcal{C}_j}(w_{ki}\gamma_k)^\alpha c_i^{\alpha/\phi_i}&\leq d_\wedge \leq 2^{\alpha+1}C_\alpha\sum_{k\in\mathcal{C}_i}(w_{ki}\gamma_k)^\alpha c_i^{\alpha/\phi_i},\\
        2C_\alpha\sum_{k\in\mathcal{C}_i\cap\mathcal{C}_j}(w_{kj}\gamma_k)^\alpha c_j^{\alpha/\phi_j}&\leq d_\vee \leq 2^{\alpha+1}C_\alpha\sum_{k\in\mathcal{C}_j}(w_{kj}\gamma_k)^\alpha c_j^{\alpha/\phi_j}.
    \end{split}
\end{equation*}
\end{enumerate}
\end{proposition}

\doublespacing
\begin{remark}
\upshape
If $\bs_i$ and $\bs_j$ are so distant that $\mathcal{C}_i\cap \mathcal{C}_j=\emptyset$, $R_i$ and $R_j$ are completely independent and hence $\chi_{ij}^R=0$. If there is at least one basis function that covers both $\bs_i$ and $\bs_j$, $\chi_{ij}^R=\lim_{x\rightarrow\infty} \Pr(R_j>x\given R_i>x) = C_K(\bw_i,\bw_j,\boldsymbol{\gamma})/\bar{\gamma}_i^\alpha>0$, which means $R_i$ and $R_j$  are asymptotically dependent.
\end{remark}

% To evaluate $\chi_{ij}(u)$ for model \eqref{eqn:model}, the main difficulty lies in the approximation of the joint probability
% \begin{equation*}
%     P(X_i>x,X_j>x)=P\{R_i^{\phi_i}g(Z_i)>x,R_j^{\phi_j}g(Z_j)>x\}.
% \end{equation*}

In the following theorem, we describe the tail dependence properties of our proposed model given in Section \ref{sec:model}.  It shows that the model can be simultaneously AD (or AI) at short spatial distances and AI at long spatial distances.  It makes precise how the interplay between the compact kernels and the spatially varying tail parameter $\phi(\bs)$ determines the class and strength of the spatial tail dependence.  Using the existing results on regular variation from \citet{breiman1965some}, \citet{cline1986convolution} and \citet{engelke2019extremal}, it presents bounds for the dependence coefficients $\chi_{ij}$ and $\eta_{ij}$, which depend on the spatial separation of the locations $\bs_i$ and $\bs_j$, the configuration of the compact kernels, and the value of the tail parameter $\phi(\bs)$ at $\bs_i$ and $\bs_j$.

\singlespacing
\begin{theorem}\label{thm:dependence_properties}

\noindent Under the definitions and notation established in the previous sections, for locations $\bs_i$ and $\bs_j$, we denote
\begin{equation*}
    v_{ki}=\frac{(w_{ki}\gamma_k)^\alpha}{\sum_{k'\in \mathcal{C}_i}(w_{k'i}\gamma_k)^\alpha},\quad 
    v_{kj}=\frac{(w_{kj}\gamma_k)^\alpha}{\sum_{k'\in \mathcal{C}_j}(w_{k'j}\gamma_k)^\alpha},
\end{equation*}
and $v_{k,\wedge}=\min(v_{ki},v_{kj})$, $v_{k,\vee}=\max(v_{ki},v_{kj})$. Also, let $W_j=g(Z_j)$, $j=1,\ldots, D$.
% \begin{enumerate}[(a)]
\begin{enumerate}[label=(\alph*)]
    \item\label{item:local_AD_AI}\label{thm:local} If $\mathcal{C}_i\cap \mathcal{C}_j\neq \emptyset$,
    the dependence class of the pair $(X_i, X_j)^{\trans}$ depends on the tail index parameters $\phi_i$ and $\phi_j$.
    
    \begin{enumerate}[label=(\roman*),ref=(\roman*)]
        \item\label{item:local_AD}\label{item:local_case1} If $\alpha<\phi_i<\phi_j$, the pair $(X_i,X_j)^T$ is asymptotically dependent with $\eta_{ij}=1$ and
        \begin{footnotesize}
        \begin{equation*}
        \begin{split}
         \chi_{ij}&=\mathbb{E}\left\{\min\left(\frac{W_i^{\alpha/\phi_i}}{\mathbb{E}(W_i^{\alpha/\phi_i})},\frac{W_j^{\alpha/\phi_j}}{\mathbb{E}(W_j^{\alpha/\phi_j})}\right)\right\}\sum_{k=1}^K v_{k,\wedge}.
        \end{split}
        \end{equation*}
    \end{footnotesize}
    \item\label{both_small}\label{item:local_case2} If $\phi_i<\phi_j<\alpha$, the pair $(X_i,X_j)^T$ is asymptotically independent with $\chi_{ij}=0$ and
    \begin{equation*}
        \begin{cases}
        \eta_{ij}=\eta_{ij}^W,&\text{ if }\eta_{ij}^W>\phi_j/\alpha,\\
        \eta_{ij}\in [\eta_{ij}^W,\phi_j/\alpha],&\text{ if }\phi_i/\alpha<\eta_{ij}^W<\phi_j/\alpha,\\
        \eta_{ij}\in [\phi_i/\alpha,\phi_j/\alpha],&\text{ if }\eta_{ij}^W<\phi_i/\alpha.\\
        \end{cases}
    \end{equation*}
    
    \item\label{item:local_case3} If $\phi_i<\alpha<\phi_j$, the pair $(X_i,X_j)^T$ is also asymptotically independent with $\chi_{ij}=0$ and
    \begin{equation*}
    \begin{dcases}
     1/\eta_{ij}\in[ (1-\phi_i/\alpha)/(2\eta_{ij}^W)+1,2-\phi_i/\alpha],&\text{if }\eta_{ij}^W\leq (\phi_i/\alpha+\phi_j/\alpha)/2,\\
     1/\eta_{ij}\in[ (1+\phi_j/\alpha)/(2\eta_{ij}^W),2-\phi_i/\alpha],&\text{if }\eta_{ij}^W> (\phi_i/\alpha+\phi_j/\alpha)/2.
    \end{dcases}
    \end{equation*}
    \end{enumerate}

    \item\label{thm:AI_case}\label{thm:long_range} If $\mathcal{C}_i\cap \mathcal{C}_j=\emptyset$, the pair $(X_i,X_j)^T$ is asymptotically independent with $\chi_{ij}=0$. When $\rho_{ij}:= \text{Cor}(Z_i, Z_j)=0$, $X_i$ and $X_j$ are completely independent. When $\rho_{ij}\neq 0$,
    % \begin{enumerate}[label=(\roman*),ref=\arabic{enumi}(\roman*)]
    \begin{enumerate}[label=(\roman*)]
        \item\label{thm:AI_case1}\label{item:long_range_case1} If $\alpha<\phi_i<\phi_j$, 
        \begin{equation*}
            \begin{cases}
            \eta_{ij}=1/2,&\text{ if }\phi_i/\alpha>2,\\
            \eta_{ij}\in [1/2,\alpha/\phi_i],&\text{ if }\phi_i/\alpha<2<\phi_j/(\alpha\eta_{ij}^W),\\
            \eta_{ij}\in [\alpha\eta_{ij}^W/\phi_j,\alpha/\phi_i],&\text{ if }\phi_j/(\alpha\eta_{ij}^W)<2.\\
            \end{cases}
        \end{equation*}
    
        \item\label{thm:AI_case2}\label{item:long_range_case2} If $\phi_i<\phi_j<\alpha$, 
        \begin{equation*}
        \begin{cases}
        \eta_{ij}=\eta_{ij}^W,&\text{ if }\eta_{ij}^W>\phi_j/\alpha,\\
        \eta_{ij}\in [\eta_{ij}^W,\phi_j/\alpha],&\text{ if }\eta_{ij}^W<\phi_j/\alpha.\\
        \end{cases}
         \end{equation*}
    
        \item\label{thm:AI_case3}\label{item:long_range_case3} If $\phi_i<\alpha<\phi_j$, 
        \begin{equation*}
        \begin{dcases}
         1/\eta_{ij}\in[1/(1+\rho_{ij})+1,2],&\text{if }2\eta_{ij}^W\leq \phi_j/\alpha,\\
         1/\eta_{ij}\in[(1+\phi_j/\alpha)/(1+\rho_{ij}),2],&\text{if }2\eta_{ij}^W> \phi_j/\alpha.
        \end{dcases}
        \end{equation*}
    \end{enumerate}
\end{enumerate}
\end{theorem}
\doublespacing
\begin{remark}
\upshape
In Theorem \ref{thm:dependence_properties}\ref{item:local_AD_AI}, when $\mathcal{C}_i\cap \mathcal{C}_j\neq \emptyset$, both AD and AI are possible as the dependence strength is controlled by both $(\phi_i,\phi_j)\trans$ and the weights $(\bw_i\trans,\bw_j\trans)\trans$. Roughly speaking, larger $\min\{\phi_i,\phi_j\}$ or smaller $||\bw_i-\bw_j||$ induces stronger dependence. Additionally, when $\mathcal{C}_i\cap \mathcal{C}_j=\emptyset$, the dependence structure of $(Z_i,Z_j)\trans$ is recovered.
\end{remark}
\begin{remark}
\upshape 
Theorem \ref{thm:dependence_properties} suggests a natural interpretation of the spatially varying tail parameter $\phi(\bs)$.  For a given location $\bs_i$, if $\phi(\bs_i)$ is larger than $\alpha$, then $X(\bs_i)$ is asymptotically dependent with $X(\bs_j)$ at location $\bs_j$, provided that \revise{$\phi(\bs_j)$} is also larger than $\alpha$ \emph{and} $\bs_i$ shares a kernel with \revise{$\bs_j$}. However, if $\phi(\bs_i)$ is smaller than $\alpha$, then $X(\bs_i)$ can never be asymptotically dependent with any other location.  Therefore, the surface $\phi(\bs)$ represents the \emph{potential strength of tail dependence} of the process at all locations $\bs$ with any other location.
\end{remark}

\subsection{Examples}\label{sec:examples}

In this section, we provide a few concrete examples for the dependence model $\{X(\bs)\}$ defined in \eqref{eqn:model} and selections for the Stable variable parameters.

\begin{example}[\citet{huser2019modeling} process] \upshape
This is a scale mixture model defined in \eqref{eqn:scalemix_model} where $\{Z(\bs)\}$ is a stationary spatial process with hidden regular variation, the link function $g(\cdot)$ transforms the margins of $\{Z(\bs)\}$ to standard Pareto (corresponding to $\delta=1$ in \eqref{eqn:type_II}) and
\begin{equation*}
    R\given \phi^H \sim \text{Pareto}\left(\frac{1-\phi^H}{\phi^H}\right),\; \phi^H\in [0,1].
\end{equation*}
This model is tail equivalent to a special case of \eqref{eqn:model} obtained by fixing $K\equiv 1$, $r=\infty$ and $\phi(\bs)\equiv \phi$, and also making $\delta=1$, which means $R(\bs)\equiv S_1\sim \text{Stable}(\alpha, 1, \gamma_1, 1)$. If we denote $\alpha/\phi=(1-\phi^H)/\phi^H$, the tail behavior in \eqref{eqn:marginal} corresponds to that of Equation (9) in \citet{huser2019modeling}: the case when $\phi\in (0,\alpha)$ corresponds to $\phi^H\in [0,1/2)$ which induces asymptotic independence, while $\phi\in (\alpha,\infty)$ corresponds to $\phi^H\in (1/2,1)$ which induces asymptotic dependence. When $\phi=\alpha$, the tail decays with the same order as a Gamma random variable with rate $1$ and shape $2$, which is the same as the survival function of the case $\phi^H=1/2$.

Furthermore, since $K\equiv 1$ and $r=\infty$, the interval for $\eta_{ij}$ from Theorem \ref{thm:dependence_properties}\ref{item:local_AD_AI} reduces to a singleton and we have
\begin{equation*}
    \chi_{ij} = 
    \mathbb{E}\left[\min\left\{\frac{W_i^{\alpha/\phi}}{\mathbb{E}(W_i^{\alpha/\phi})},\frac{W_j^{\alpha/\phi}}{\mathbb{E}(W_j^{\alpha/\phi})}\right\}\right]\mathbbm{1}(\phi<\alpha),
\end{equation*}
which is exactly equation (10) in \citet{huser2019modeling} given $\alpha/\phi=(1-\phi^H)/\phi^H$. Similarly, we can show that $\eta_{ij}$ under the stationary version of \eqref{eqn:model} is also equivalent to that of \citet[][see their equation (11)]{huser2019modeling}.
\end{example}

\begin{example}[\citet{huser2019modeling} process with spatially-varying $R$]
\upshape
Similar to the generalization from the \citet{huser2017bridging} process to model \eqref{eqn:huser_model} by \citet{hazra2021realistic}, we can allow $K>1$ and $r<\infty$ while holding the parameters for the Stable variables constant at all knots. That is,
\begin{equation*}
    \phi(\bs)\equiv \phi, \; S_k \iid \text{Stable}(\alpha,1,\gamma,\delta), k=1,\ldots, K.
\end{equation*}
To achieve local asymptotic dependence (i.e., scenario Theorem \ref{thm:dependence_properties}\ref{item:local_AD_AI}\ref{item:local_AD}) with a constant $\phi$, the value of $\phi$ has to be greater than $\alpha$. Then,  as in \citet{hazra2021realistic}, local asymptotic independence is not possible under this construction because $\phi(\bs)\equiv \phi>\alpha$ and \ref{item:local_AD_AI}\ref{item:local_AD} always holds locally. However, this case achieves more flexible local dependence properties than the model specified in \eqref{eqn:huser_model} as the value of $\phi\in (\alpha, \infty)$ is not fixed (recall that $\phi$ has to be fixed at $0$ in \citet{hazra2021realistic} because the \citet{huser2017bridging} model can only achieve asymptotic dependence when $\phi=0$).
\end{example}

\begin{illustration}[Empirical evaluations of the bounds in Theorem \ref{thm:dependence_properties}]
\upshape
We simulate $N=300,$$000,$$000$ data sets on the domain $\mathcal{S}=[0,10]\times [0,10]$ using a $\{\phi(\bs):\;\bs\in\mathcal{S}\}$ surface that varies smoothly between 0.4 to 0.7; see Figure \ref{fig:phi_surface_demon}. For each simulation, we generate independent Stable variables with $\alpha = 0.5$ at 9 knots on a uniform grid and average them using compact Wendland basis functions centered at the knots \citep{Wendland1995}. 
We empirically estimate the $\chi_{ij}(u)$ and $\eta_{ij}(u)$ functions defined in \eqref{eqn:chi_measure} and \eqref{eqn:eta_measure} between pairs of the sample points on a grid of $u$ values, using the $N$ independent simulations of $(R_i^{\phi_i}W_i,R_j^{\phi_j}W_j)$. 

Corresponding to Theorem \ref{thm:dependence_properties}\ref{thm:local}\ref{item:local_case1}, Figure \ref{fig:emp_chi_example_a} shows the empirical estimated $\chi_{12}(u)$ and $\eta_{12}(u)$ for the $1$st and $2$nd sample points over a grid of $u$. These two points share a common Wendland  kernel and $\alpha = 0.5 < \phi_2 < \phi_1$, so they are asymptotically dependent. 
Corresponding to Theorem \ref{thm:dependence_properties}\ref{thm:local}\ref{item:local_case2}, Figure \ref{fig:emp_chi_example_b} shows the empirical estimated $\chi_{34}(u)$ and $\eta_{34}(u)$ for the $3$rd and $4$th sample points. These two points share  a common Wendland  kernel, but $\phi_3 < \phi_4 < \alpha$, so they are asymptotically independent. 
Figure \ref{fig:emp_chi_example_c} shows the empirical estimated $\chi_{45}(u)$ and $\eta_{45}(u)$ for the $4$th and $5$th sample points. This corresponds to Theorem~\ref{thm:dependence_properties}\ref{thm:local}\ref{item:local_case3}, as the two points also share  a common Wendland kernel, yet $\phi_4 < \alpha < \phi_5$, so they are still asymptotically independent. 
Finally, Figure \ref{fig:emp_chi_example_d} shows the empirical estimated $\chi_{15}(u)$ and $\eta_{15}(u)$ for the $1$st and $5$th sample points. Corresponding to Theorem \ref{thm:dependence_properties}\ref{thm:long_range}\ref{item:long_range_case1}, these two locations do not share any  common Wendland kernel, so even-though $\phi_5>\phi_1>\alpha$, we still have asymptotically independence.
\end{illustration}

\begin{figure}[ht]
    \centering
  \begin{minipage}[c]{0.45\textwidth}
  \centering
    \includegraphics[width=\textwidth]{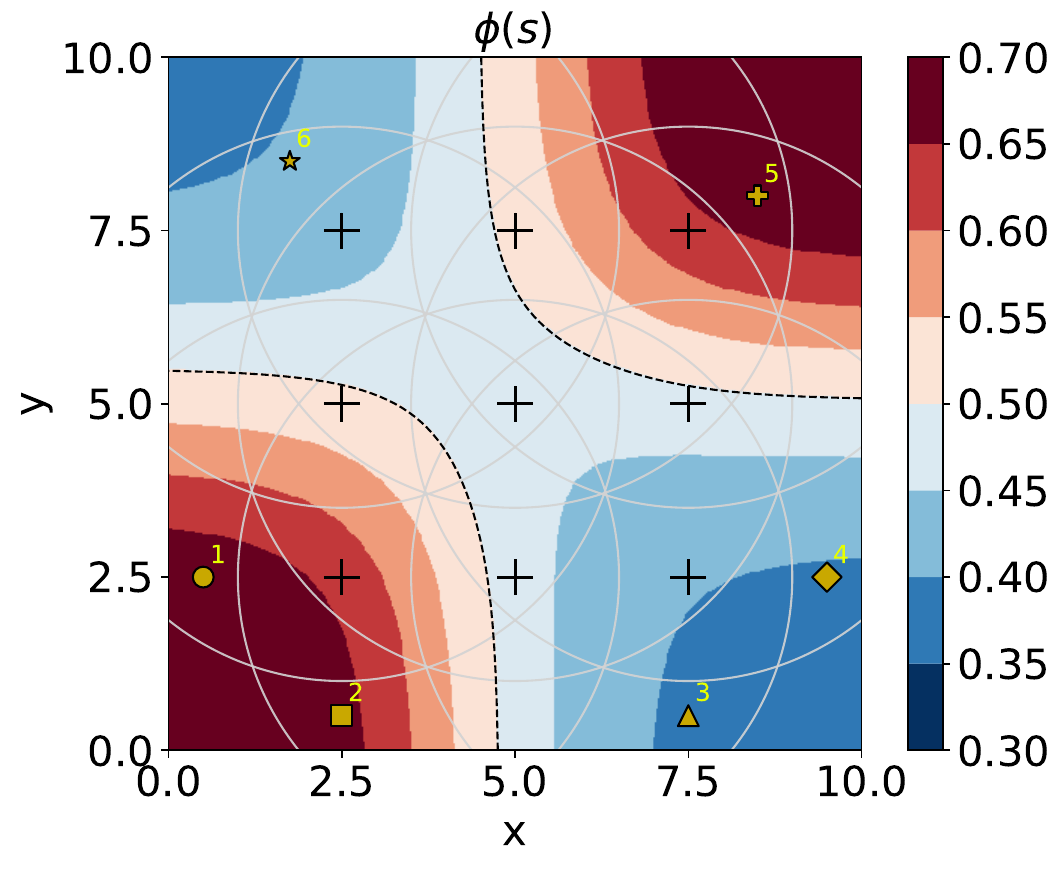}
  \end{minipage}
  % \hfill
  \hspace{5mm}
  \begin{minipage}[c]{0.4\textwidth}
    \caption{A $\phi(\bs)$ surface on $[0,10]^2$, in which the dashed line marks the transition between local AI and AD. The points with `+' are centers for the Stable variables and the compact Wendland basis functions. The points with other signs/marker-styles are randomly chosen sample points that we use to verify the dependence properties in Theorem \ref{thm:dependence_properties}.
    } \label{fig:phi_surface_demon}
  \end{minipage}
\end{figure}

\begin{figure}[h]
    \centering
    % First row
    \begin{minipage}{0.496\textwidth}
        \begin{subfigure}[t]{\textwidth}
            \centering
            \includegraphics[width=0.492\textwidth, clip=true, trim=20 30 0 0]{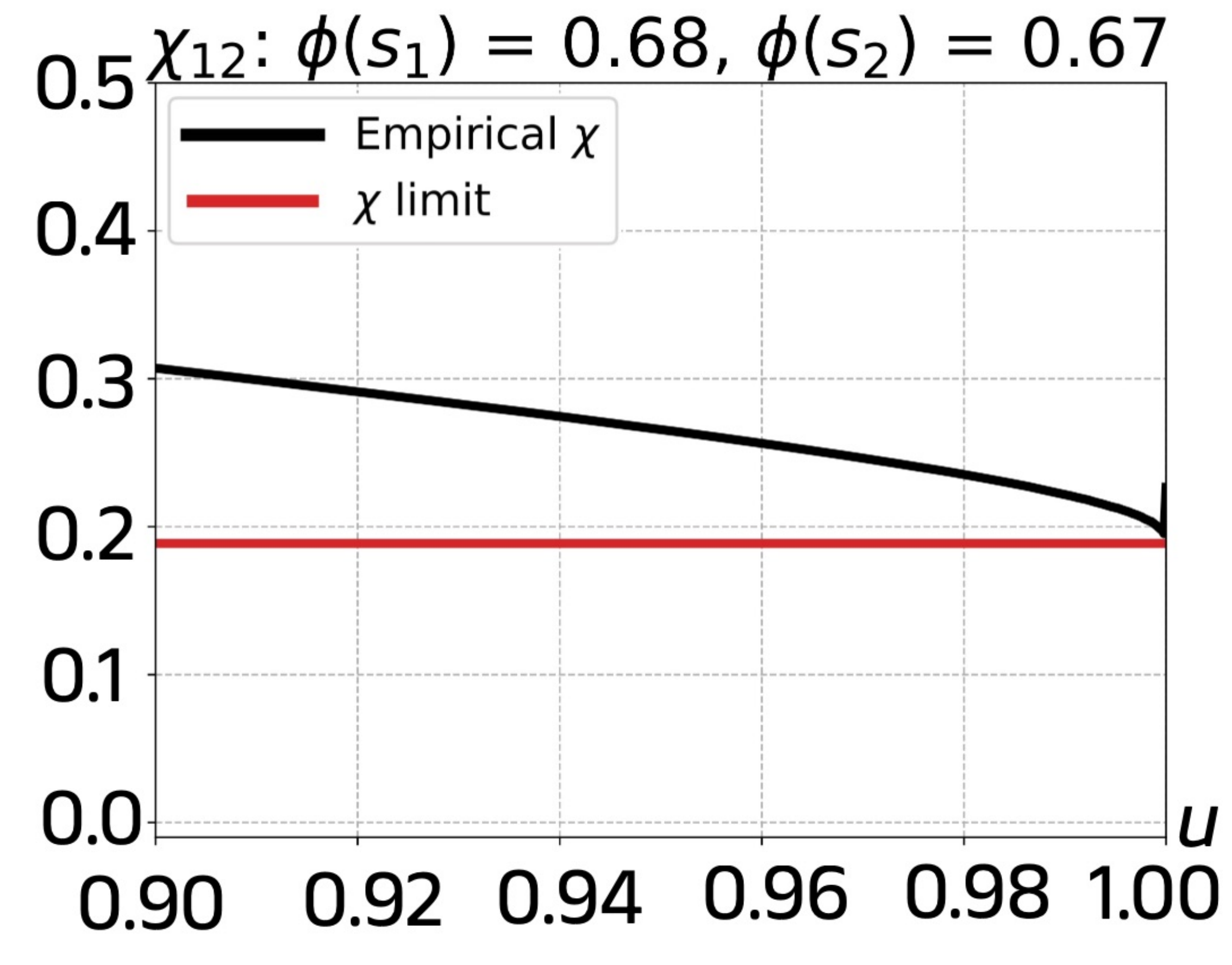}
            \includegraphics[width=0.492\textwidth, clip=true, trim=20 0 0 0]{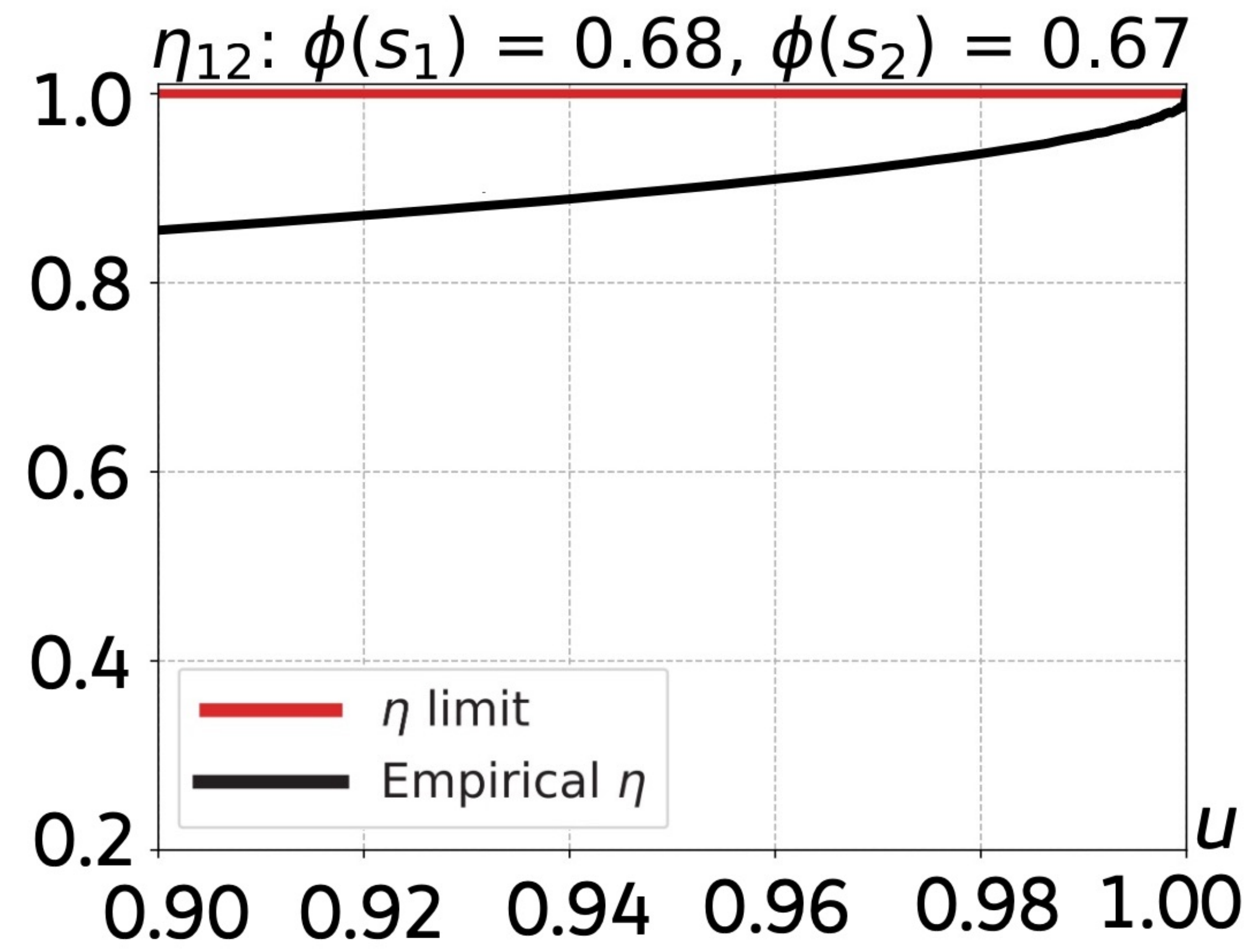}
            \caption{Thm \ref{thm:dependence_properties} \ref{thm:local}\ref{item:local_case1} $\chi_{12}$ and $\eta_{12}$}
            \label{fig:emp_chi_example_a}
        \end{subfigure}
    \end{minipage}
    \hfill
    \begin{minipage}{0.496\textwidth}
        \begin{subfigure}[t]{\textwidth}
            \centering
            \includegraphics[width=0.492\textwidth, clip=true, trim=20 30 0 0]{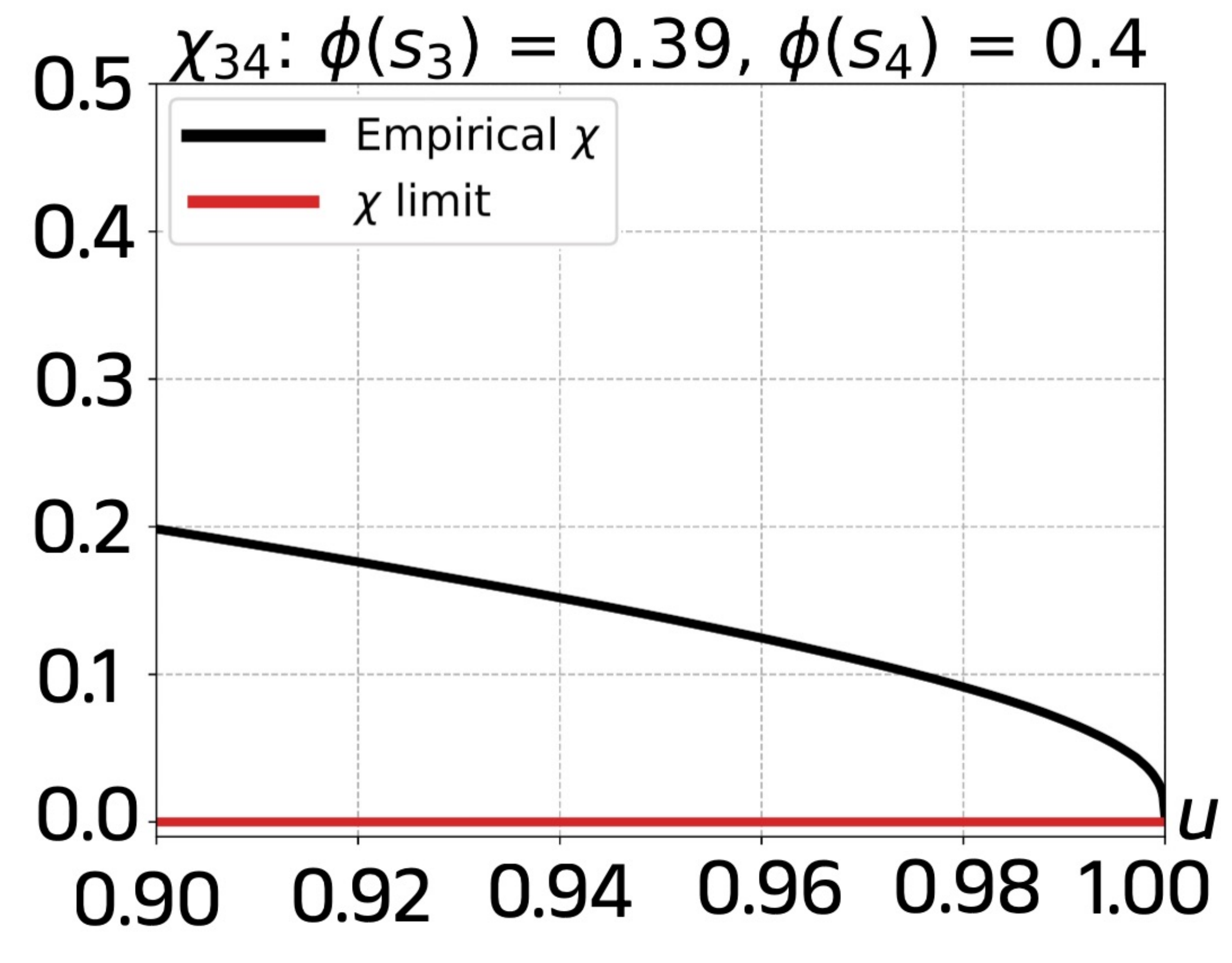}
            \includegraphics[width=0.492\textwidth, clip=true, trim=20 0 0 0]{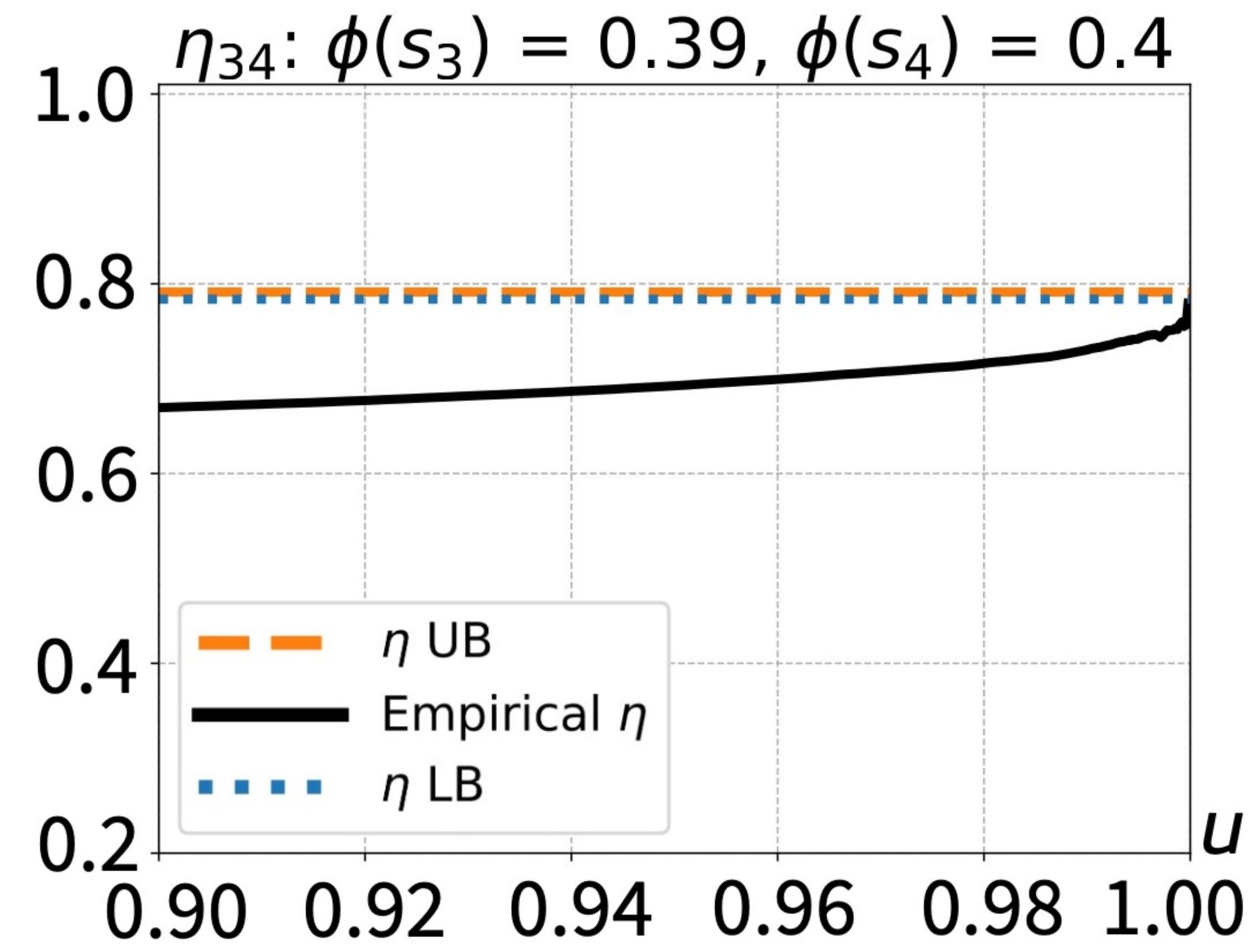}
            \caption{Thm \ref{thm:dependence_properties}\ref{thm:local}\ref{item:local_case2} $\chi_{34}$ and $\eta_{34}$}
            \label{fig:emp_chi_example_b}
        \end{subfigure}
    \end{minipage}

    % Second row
    \begin{minipage}{0.496\textwidth}
        \begin{subfigure}[t]{\textwidth}
            \centering
            \includegraphics[width=0.492\textwidth, clip=true, trim=20 30 0 0]{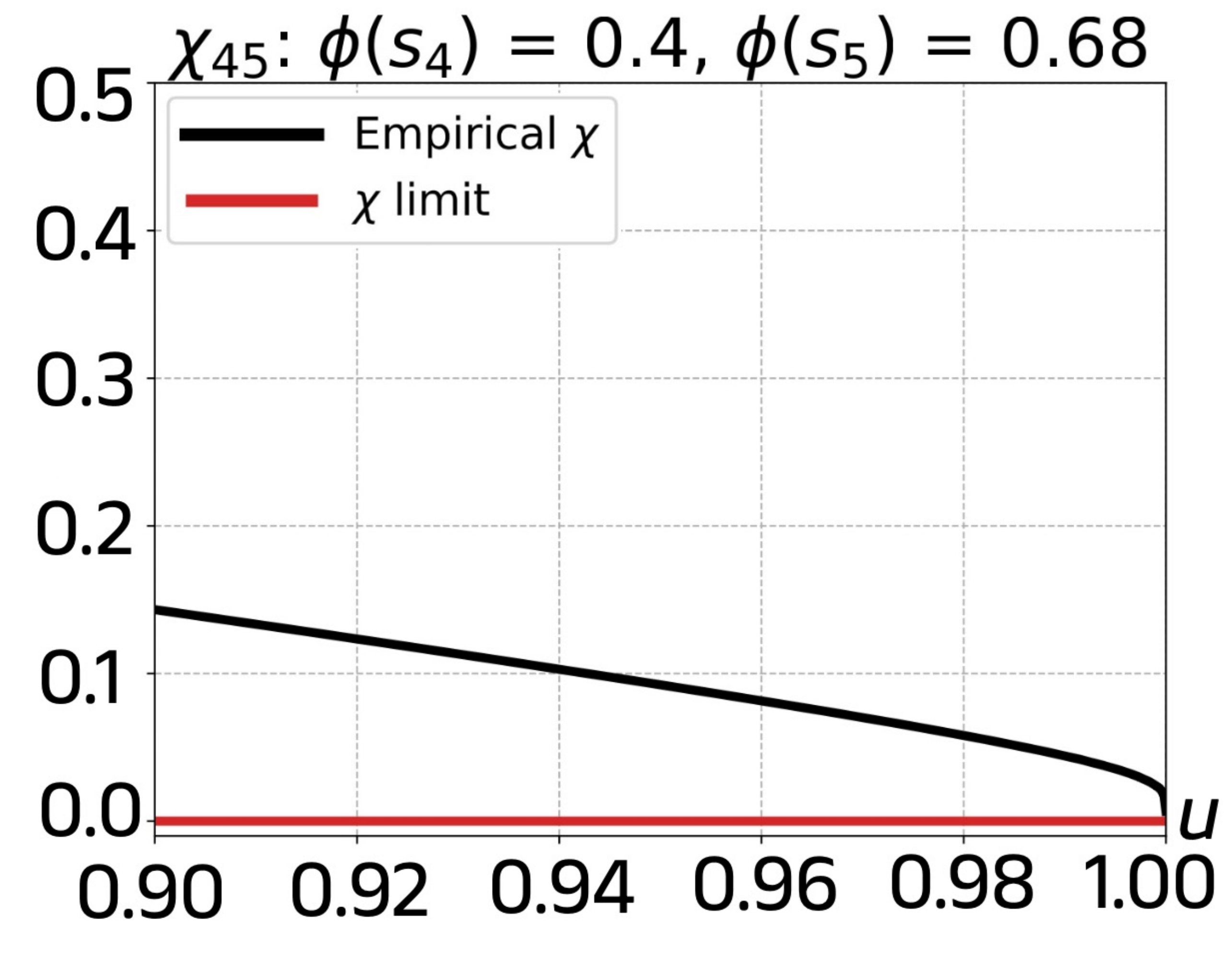}
            \includegraphics[width=0.492\textwidth, clip=true, trim=20 0 0 0]{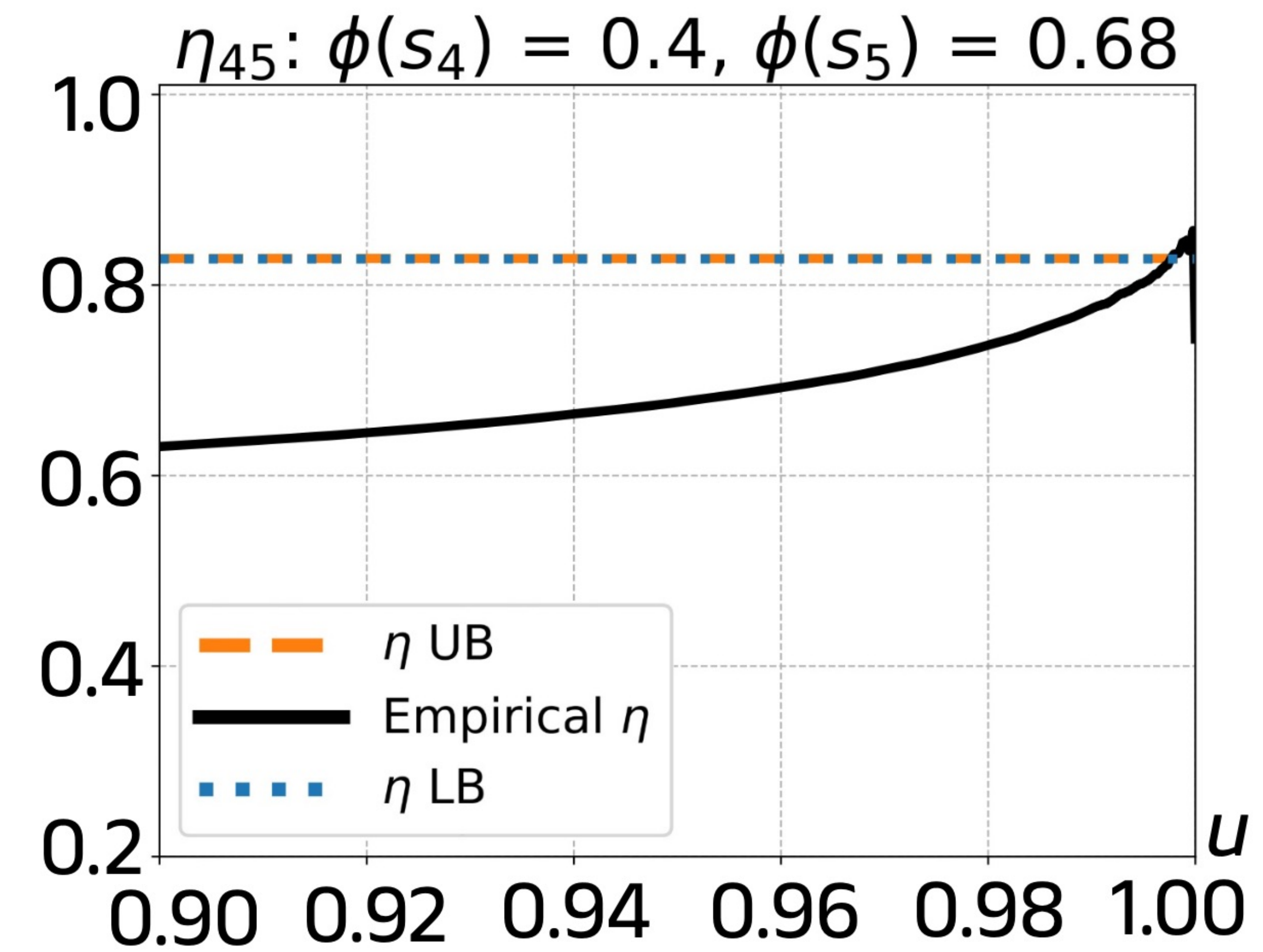}
            \caption{Thm \ref{thm:dependence_properties}\ref{thm:local}\ref{item:local_case3} $\chi_{45}$ and $\eta_{45}$}
            \label{fig:emp_chi_example_c}
        \end{subfigure}
    \end{minipage}
    \hfill
    \begin{minipage}{0.496\textwidth}
        \begin{subfigure}[t]{\textwidth}
            \centering
            \includegraphics[width=0.492\textwidth, clip=true, trim=20 30 0 0]{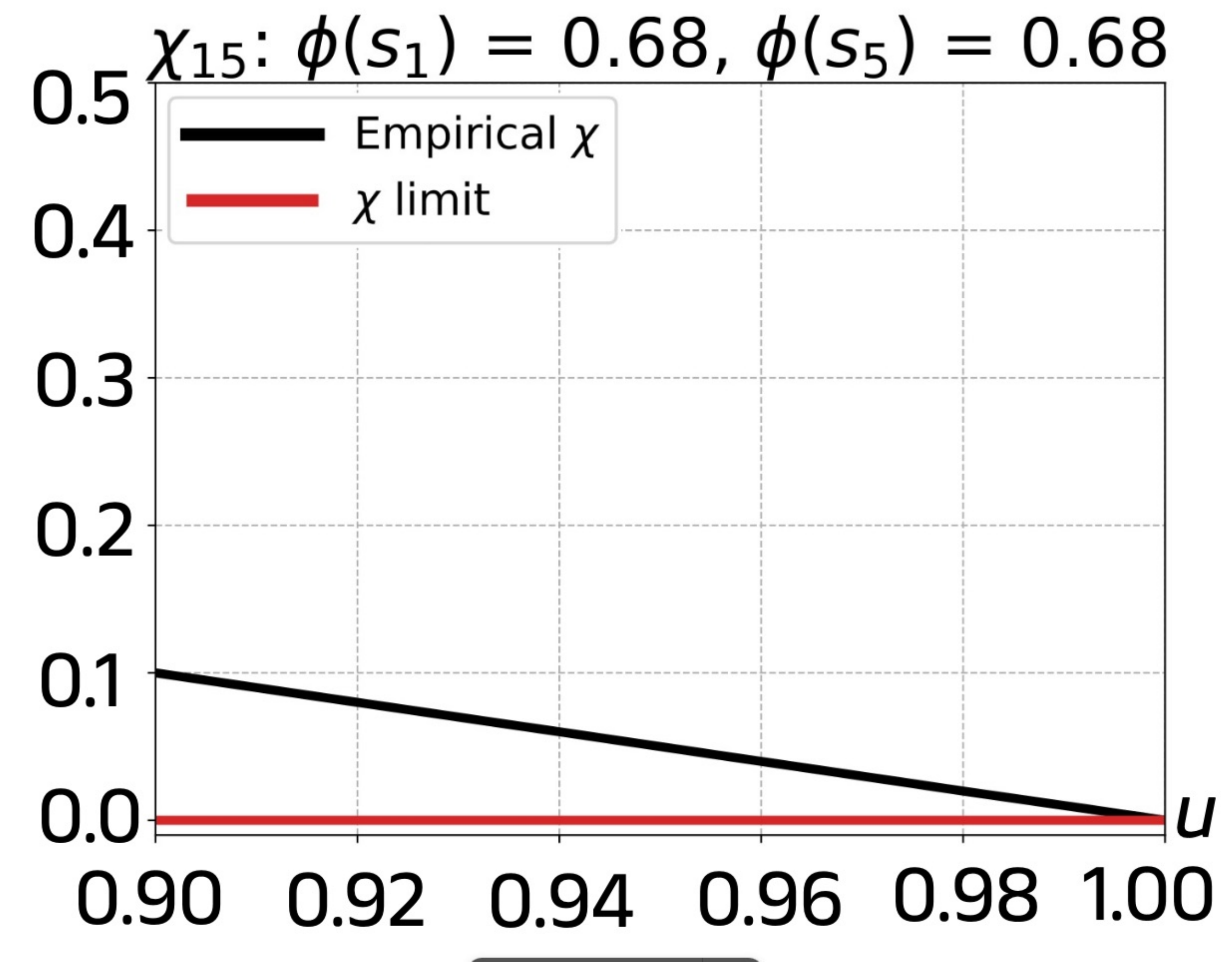}
            \includegraphics[width=0.492\textwidth, clip=true, trim=20 0 0 0]{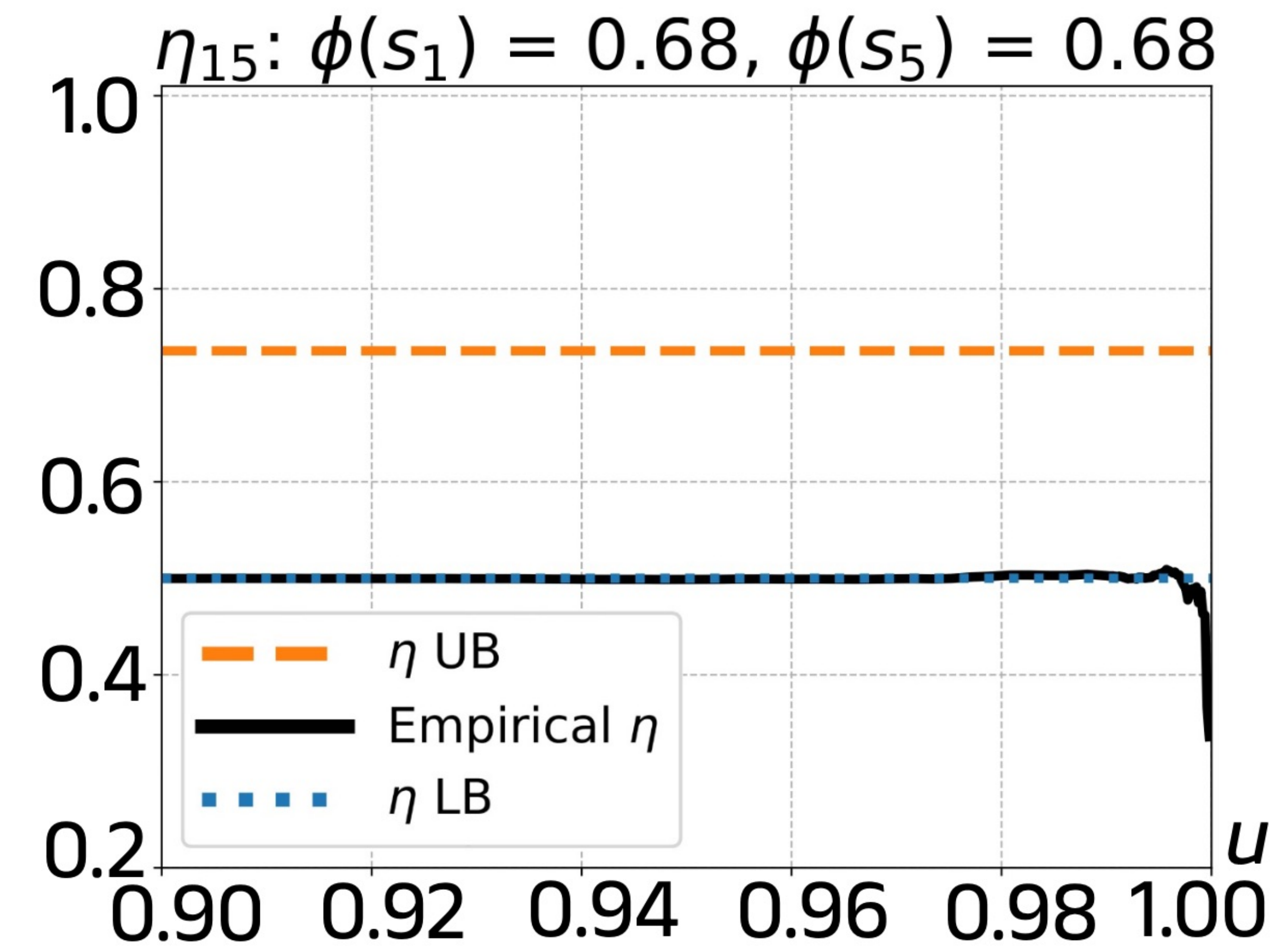}
            \caption{Thm \ref{thm:dependence_properties}\ref{thm:long_range}\ref{item:long_range_case1} $\chi_{15}$ and $\eta_{15}$}
            \label{fig:emp_chi_example_d}
        \end{subfigure}
    \end{minipage}

    \caption{Empirical estimates of the dependence coefficients $\chi_{ij}(u)$ and $\eta_{ij}(u)$ between sample points $i$ and $j$ in Figure \ref{fig:phi_surface_demon}. The red lines mark the theoretical limit and the orange dashed and blue dotted lines respectively represent theoretical upper and lower bounds. The $\chi_{ij}(u)$'s are direct empirical estimates; due to numerical instability, particularly as the probability $u$ approaches its upper limit of 1, we estimate the $\eta_{ij}(u)$'s using the Hill estimator with marginal transformation to standard Exponential \citep{R_mev}.}
    \label{fig:emp_chi_example}
\end{figure}

\section{Bayesian inference}\label{sec:inference}

We define a Bayesian hierarchical model based on the dependence model \eqref{eqn:model} in Section \ref{sec:model_construction} and use a MCMC algorithm to fit to the data. The dependence model \eqref{eqn:model} is displayed again here for convenience, now with each time replicate denoted with a subscript $t = 1, \dots, T$:
\begin{equation*}
    X_t(\bs)=R_t(\bs)^{\phi(\bs)}g(Z_t(\bs)).
\end{equation*}

\subsection{Hierarchical model} \label{sec:hierarchical_model}

Let $\{Y_t(\bs):\bs\in \mathcal{S}, t = 1, \dots, T\}$ denote the spatial process on the scale of the observations. We link this  process to our dependence model using a marginal probability integral transform within the hierarchical model as
\begin{equation}\label{eqn:prob_inte_trans}
   % F_{Y \given \btheta_Y(\bs),t}(Y_t(\bs)) = F_{X\given\phi(\bs), \bar{\gamma}(\bs), t}(\underbrace{R_t(\bs)^{\phi(\bs)}g(Z_t)(\bs)}_{X_t(\bs)}).
   F_{Y \given \btheta_Y(\bs),t}(Y_t(\bs)) = F_{X\given\phi(\bs), \bar{\gamma}(\bs), t}(X_t(\bs)).
\end{equation}

In principle, the dependence model can be used for any marginal distribution. 
Here we opt for the block-maxima approach---considering annual maxima at time $t$ as the observed process $Y_t(s)$, assuming a Generalized Extreme Value (GEV) marginal distribution. 
We let the GEV marginal parameters vary in space and time as
\begin{equation*}
    Y_t(\bs)\sim \text{GEV}(\mu_t(\bs), \sigma_t(\bs), \xi_t(\bs)).
\end{equation*}
Let $\btheta_{\text{GEV},t}(\bs) = (\mu_t(\bs), \sigma_t(\bs),\xi_t(\bs))\trans$ be the vector of marginal GEV parameters at location $\bs$ and time $t$. Given the time-varying parameters, it is reasonable to then assume conditional temporal independence among the annual maxima. % $\{F_{Y\given\btheta_{\text{GEV},t}(\bs)}(Y_t(\bs)):\bs\in\mathcal{S}\}$, $t=1,\ldots, T$.

% Dependence Model and Likelihood

Next, we define a hierarchical model based on the mixture \eqref{eqn:model}. Conditioning on the scaling variables at the knots $\boldsymbol{S}_t$, 
the full conditional likelihood for the observation vector at time $t$ is
\begin{equation}\label{eq:likelihood}
    \mathcal{L}(\boldsymbol{Y}_t \given \boldsymbol{R}_t, \boldsymbol{\gamma}, \boldsymbol{\phi},\boldsymbol{\theta}_{\text{GEV}}, \boldsymbol{\rho}) = \varphi_D(\boldsymbol{Z}_t)\left\lvert\frac{\partial\boldsymbol{Z}_t}{\partial \boldsymbol{Y}_t}\right\rvert
    % S_{kt} \given \gamma &\sim \text{Stable}(\alpha \equiv 0.5, 1, \gamma_k, \delta \equiv 0), k = 1, \dots, K \notag
\end{equation}

Since we have assumed conditional temporal independence by introducing the time-varying marginal parameters, likelihoods across the independent time replicates are multiplied together for the joint likelihood. In the likelihood \eqref{eq:likelihood}, $\varphi_D$ is the $D$-variate Gaussian density function with covariance matrix $\bm{\Sigma_{\rho}}$ and $\partial\bm{Z}_t/\partial \bm{Y}_t$ is the Jacobian. Additional details are included in the Appendix \ref{sec:Appendix MCMC}.

For $\bm{\Sigma_{\rho}}$, we use a locally isotropic, non-stationary Mat\'{e}rn covariance function \citep{paciorek2006spatial, risser2015regression}
\begin{equation}
    C(\bs,\bs') = \zeta(\bs)\zeta(\bs')\frac{\sqrt{\rho(\bs)\rho(\bs')}}{\{\rho(\bs)+\rho(\bs')\}/2}\mathcal{M}_{\nu}\left(\frac{\lVert \bs-\bs' \rVert}{\sqrt{\{\rho(\bs)+\rho(\bs')\}/2}}\right),
\end{equation}
where the standard deviation process $\zeta(\bs)\equiv 1$ for $\bs\in \mathcal{S}$, $\mathcal{M}_{\nu}(\cdot)$ is the Mat\'{e}rn correlation function with range $1$ and smoothness $\nu$, and $\rho(\bs)$ is the spatially varying range parameter.

As described in Section \ref{sec:model_construction}, we form a spatially varying surface $R(\bs)$ by combining $K$ compactly supported Wendland kernel functions, each centered at a knot, each scaled by a corresponding  Stable variable. In addition, we construct the $\{\phi(\bs)\}$ and  $\{ \rho(\bs) \}$ surfaces using Gaussian kernel functions centered at the knots. 

The Wendland kernel function is parameterized as $w_k^{(S)}(\bs) \propto \bigg(1-\dfrac{||\bs - \boldsymbol{b}_k^{(S)}||^2}{r}\bigg)^l_+$, in which $l = 2$, $\boldsymbol{b}_k^{(S)}$ are a grid of knots over the spatial domain $\mathcal{S}$, and $r$ is the radius of the kernel function.

%The kernel functions are compactly-supported so that some sites have disjoint sets of spatial indices with non-zero kernel function values. As for the Gaussian kernel function, it is parameterized as

% \begin{equation*} % 
%     w_k^{(\phi)}(\bs) \propto \exp\left(-\dfrac{\lVert \bs - \boldsymbol{b}_k^{(\phi)} \rVert^2}{2h^{(\phi)}}\right) \quad \text{and} \quad w_k^{(\rho)}(\bs) \propto \exp\left(-\dfrac{\lVert \bs - \boldsymbol{b}_k^{(\rho)} \rVert^2}{2h^{(\rho)}}\right) 
% \end{equation*}
% in which $\boldsymbol{b}_k^{(\phi)}$ and $\boldsymbol{b}_k^{(\rho)}$ are the grid of locations for the dependence model parameters $\phi$ and $\rho$; $h^{(\phi)}$ and $h^{(\rho)}$ are the bandwidth parameter for their respective Gaussian kernel functions.

% \BAS{I removed the formulas for the Gaussian kernel functions---I think we can assume that people know this already.}

The priors for the dependence model parameters are $\phi_k \iid \text{Beta}(5,5)$ and $\rho_k \iid \text{halfNormal}(0,2), k=1, \ldots, K$, where ``halfNormal" refers to the positively truncated normal distribution. The Beta prior for $\phi_k$ is centered at the transition boundary between AI and AD, and places less mass near the edges of the support which correspond to extremely strong or weak dependence scenarios.  In the following, we fix the scale parameter $\gamma$ of the Stable distribution (somewhat arbitrarily) at 0.5.  Varying $\gamma$ does not play any role in modulating the tail dependence characteristics of the model (see Theorem \ref{thm:dependence_properties}), so fixing it at a convenient value results in almost no loss of flexibility.

\subsection{Computation}\label{sec:computation}

To estimate the posterior distribution of the model parameters, we sequentially update each parameter using an adaptive random walk Metropolis (RWM) algorithm \citep{shaby2010exploring}.
Since we assume independence across time, we can update $R_t$'s in parallel  across $t = 1, \ldots, T$ at each MCMC iteration.

The probability integral transform in \eqref{eqn:prob_inte_trans} and the Jacobian term in the likelihood \eqref{eq:likelihood} require evaluation of the marginal distribution and density functions of the dependence model $X(\bs)$.  For general Stable variables $S_1, \ldots, S_K$, this is difficult.  However, under the special case of $\alpha = 1/2$, sometimes called a L\`evy distribution, we can obtain simpler analytical forms.  Furthermore, fixing $\alpha = 1/2$ sacrifices no flexibility with respect to the tail dependence characteristics described in Theorem \ref{thm:dependence_properties}. 
When $\alpha = 1/2$, the survival function for the mixture in \eqref{eqn:model} with the Type II (i.e. location-shifted) Pareto link function in \ref{eqn:type_II} is
\begin{equation}\label{eqn:CDF_X(s)}
    \begin{split}
        1-F_{X_j}(x) = P(R_j^{\phi_j}W_j>x)
        &=\sqrt{\frac{\bar{\gamma}_{j}}{2\pi}}\int_{0}^\infty \frac{r^{\phi_j-3/2}}{x+r^{\phi_j}}\exp\left\{-\frac{\bar{\gamma}_{j}}{2r}\right\}dr.
    \end{split}
\end{equation}
Using the Type II Pareto link function, rather than the standard Pareto, does not change the tail properties but can be advantageous for MCMC sampling because it aligns the support of the random scaling factor and the transformed Gaussian process. The trade-off is that the Type II Pareto, unlike the standard Pareto, does not give a closed form for the univariate distribution function for $X_j(s)$, and therefore requires numerical integration. 
Using Leibniz rule to take derivative with respect to $x$, we can get the univariate density function for $X_j(s)$,
\begin{equation} \label{eqn:pdf_X(s)}
    f_{X_j}(x)=\sqrt{\frac{\bar{\gamma}_{j}}{2\pi}}\int_{0}^\infty \frac{r^{\phi_j-3/2}}{(x+r^{\phi_j})^2}\exp\left\{-\frac{\bar{\gamma}_{j}}{2r}\right\}dr,
\end{equation}
which also requires numerical integration.

We evaluate the numerical integrals (e.g. the univariate distributions \eqref{eqn:CDF_X(s)} and density functions \eqref{eqn:pdf_X(s)} of the dependence model) using \texttt{GSL} libraries \citep{gough2009gnu} in \texttt{C++}. The MCMC sampler is implemented in \texttt{python}, and the parallel updates are implemented via the \texttt{mpi4py} \citep{mpi4py} module/interface to \texttt{OpenMPI}. After parallelization, on an AMD Milan EPYC CPU, running a chain, on datasets of 500 spatial locations and 64 time replicates, to approximately 15,000 iterations takes about 270 hours. \revise{Compared with the stationary model of \citet{huser2019modeling}, fitting our model % (i.e. with the addition of  spatially varying scales $R(s)$, spatially varying tail dependence parameters $\phi(s)$, and nonstationary covariance, )
approximately triples the computational cost in our  implementation.}

\subsection{Simulation and Coverage Analysis}\label{sec:simulation_studies}\label{sec:simulation_scenarios}\label{sec:coverage_analysis}
We present simulation results and conduct coverage analysis to investigate whether the MCMC procedure is able to draw accurate inference on model parameters, assuming marginally GEV responses.
We use $D = 500$ sites uniformly drawn from the square $\mathcal{S} = [0,10]^2$ and $\mathcal{T} = 64$ time replicates. The latent Gaussian process $Z_t(\bs)$ is generated with a locally isotropic, non-stationary Mat\'{e}rn covariance function as specified in Section \ref{sec:hierarchical_model}, with $\nu = 0.5$.
For each time replicate, we specified $K=9$ knot locations $\{\boldsymbol{b}_k^{(S)}, k = 1, \dots, K\}$ over a \revise{regularly spaced grid} in the spatial domain $\mathcal{S}$. As specified  in Section \ref{sec:hierarchical_model}, we generate independent L\'{e}vy random variables with $\gamma_k = 0.5$ at those nine pre-specified knot locations and interpolate them to site locations using the Wendland  kernel functions  with radius $r=4$ centered at those knots; see Figure %\ref{fig:point_space} 
\ref{fig:simulation_scenarios} for illustration. 

\begin{figure}[H]
    \centering
    \includegraphics[width=1\linewidth]{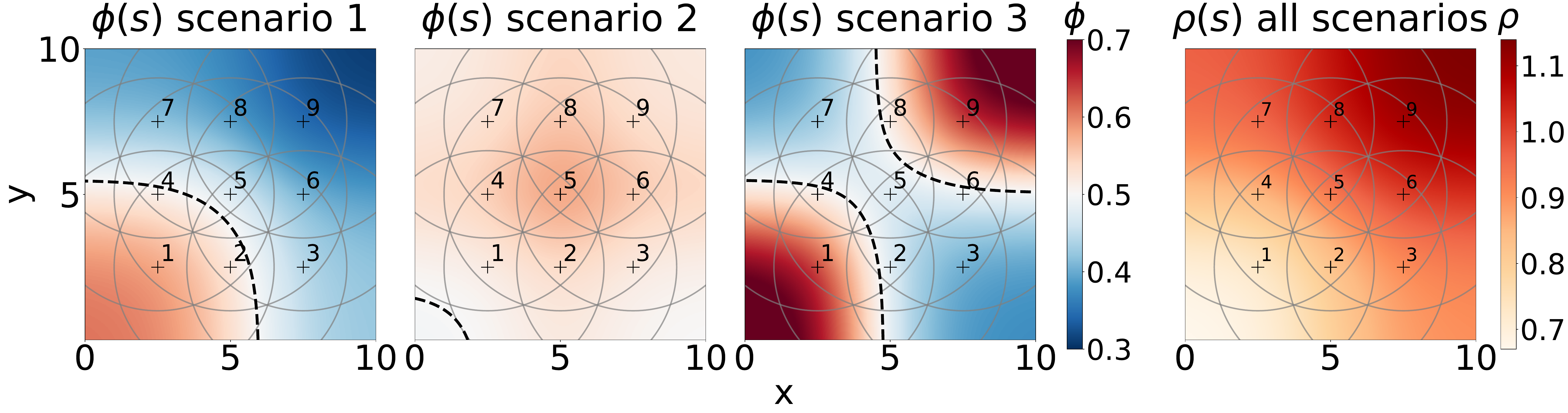}
    \caption{Dependence model parameter surfaces, $\phi(\bs)$ and $\rho(\bs)$, for the three simulation scenarios, which correspond to different levels of non-stationarity. \revise{The dashed line marks the transition between local AI and AD.} The `+' denotes the center of the kernel functions, and the gray circles denotes the radii.}
    \label{fig:simulation_scenarios}
\end{figure}

% \begin{figure}[H]
%     \centering
%     % \begin{minipage}[b]{0.15\textwidth}
%     %     \centering
%     %     \includegraphics[width=1.2\textwidth, clip=true, trim=0 0 40 30]{Empirical_Coverage_Mu_scenario1.png}
        
%     %     \vspace{3pt}
        
%     %     \includegraphics[width=1.2\textwidth, clip=true, trim=0 0 40 30]{Empirical_Coverage_Sigma_scenario1.png}
%     %     \vspace{-4pt}
%     % \end{minipage}
%     \begin{minipage}[b]{0.16\textwidth}
%     \centering
%     \includegraphics[width=1.2\textwidth]{Empirical_Coverage_MuSigma_scenario1.pdf}
%     \end{minipage}
%     \hfill
%     \begin{minipage}[b]{0.83\textwidth}
%         \centering
%         \includegraphics[width=0.48\textwidth, clip=true, trim=55 0 0 0]{Empirical_Coverage_all_Phi_scenario1.pdf}
%         \includegraphics[width=0.48\textwidth, clip=true, trim=55 0 0 0]{Empirical_Coverage_all_Range_scenario1.pdf}
%     \end{minipage}
%     \caption{Empirical coverage rates of credible intervals of the marginal parameters $\mu$ and $\sigma$ (left), the dependence parameters $\phi_k$ (middle), and $\rho_k$ (right), $k=1,\dots, 9$, in simulation scenario 1.}
%     \label{fig:scenario1_coverage}
% \end{figure}

\begin{figure}[!ht]
    \centering
    \begin{minipage}[b]{0.197\textwidth}
        \centering
        \includegraphics[width=\linewidth]{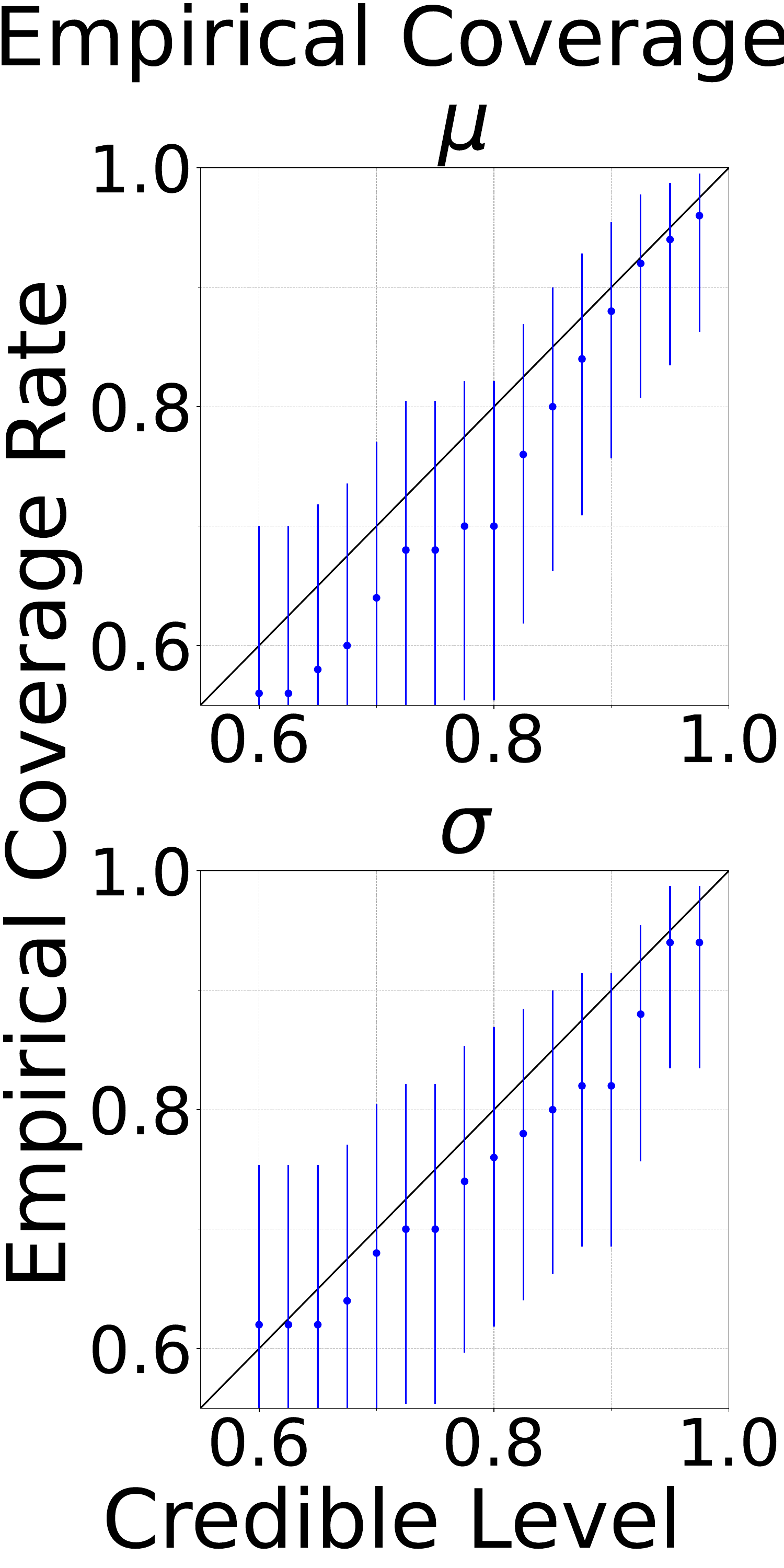}
    \end{minipage}
    % \hfill
    \begin{minipage}[b]{0.393\textwidth}
        \centering
        \includegraphics[width=\linewidth, clip=true, trim=0 0 8 0]{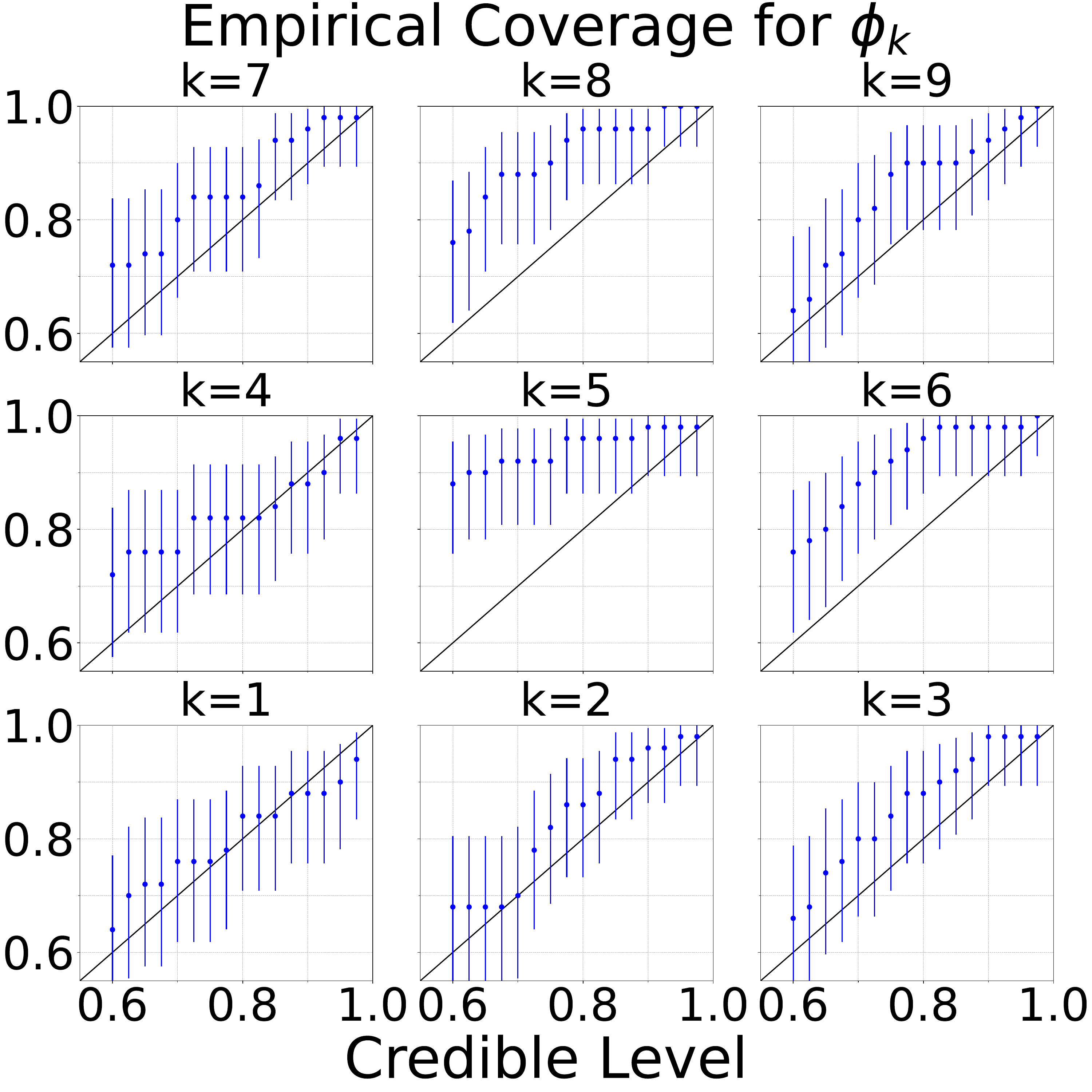}
    \end{minipage}
    % \hfill
    \begin{minipage}[b]{0.393\textwidth}
        \centering
        \includegraphics[width=\linewidth, clip=true, trim=8 0 0 0]{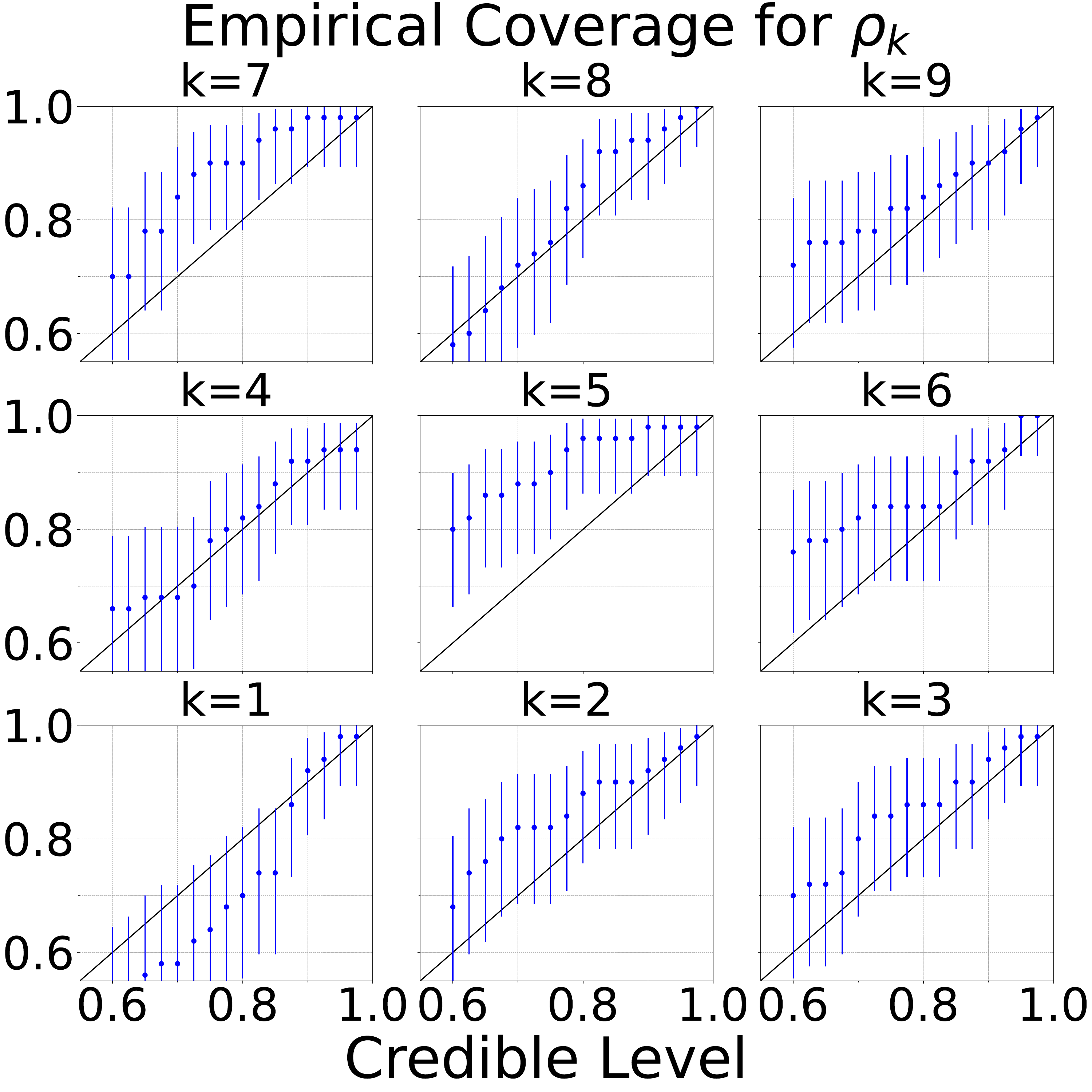}
    \end{minipage}
    \caption{Empirical coverage rates of credible intervals of the marginal parameters $\mu$ and $\sigma$ (left), the dependence parameters $\phi_k$ (middle), and $\rho_k$ (right), $k=1,\dots, 9$, in simulation scenario 1.}
    \label{fig:scenario1_coverage}
\end{figure}

We consider three scenarios for the scaling parameter $\phi(\bs)$ and the range parameter $\rho(\bs)$ so that all cases in Theorem \ref{thm:dependence_properties} are represented; we generate the nine $\phi_k$ and $\rho_k$ at the same nine knots locations ($\boldsymbol{b}_k^{(S)} = \boldsymbol{b}_k^{(\phi)} = \boldsymbol{b}_k^{(\rho)}$), and interpolate them to the site locations using the Gaussian  kernel functions  with bandwidth $h^{(\phi)} = h^{(\rho)} = 4$ centered at those knots. The interpolated $\{\phi(\bs)\}$ and $\{\rho(\bs)\}$ surfaces are shown in Figure \ref{fig:simulation_scenarios}.
Finally, we generated dependence model variables $X_t(\bs)$ under three scenarios and transformed to marginal GEV distribution with parameters $(\mu(\bs), \sigma(\bs), \xi(\bs))\trans = (0, 1, 0.2)\trans$. For the purpose of simulations, GEV parameters are set to be spatially and temporally constants with no covariate. %since both $\xi(\bs)$ and the scaling parameter $\phi(\bs)$ control the tail behavior, 
For computational expediency in the coverage analysis, $\xi(\bs)$ is not updated in the simulations. 

We study the coverage properties of the posterior inference based on the MCMC samples for the posterior credible intervals with 50 simulated datasets drawn under each of the scenarios.
Figure~\ref{fig:scenario1_coverage} shows the empirical coverage rates of the scaling parameter $\phi$ and the range parameter $\rho$ at the nine knot locations (i.e. $k = 1, ..., 9$), as well as the location $\mu$ and scale $\sigma$ of the marginal GEV parameters in simulation scenario 1. Standard binomial confidence intervals are included on the coverage plots. In all scenarios, we see that the sampler generates well-calibrated posterior inference for the GEV parameters with close to nominal frequentist coverage, and slightly over-covers for the the scaling and range parameters. The Appendix \ref{sec:Appendix Simulation} includes the additional empirical coverage rates of the parameters from simulation scenario 2 and 3 in Figures \ref{fig:scenario2_coverage} and \ref{fig:scenario3_coverage}, which show similar characteristics as the coverage plots from simulation scenario 1.

\section{Extreme \emph{in situ} of Daily Precipitation} \label{sec:application}
\subsection{Data Analysis}\label{sec:data analysis}

In this section, we analyze extreme daily measurements from a gauged network of \emph{in situ} weather stations from the Global Historical Climate Network \citep[GHCN;][]{Menne2012}. A subset of the extreme summertime measurements from GHCN stations (see Figure \ref{fig:application_US}) over the central United States was originally analyzed in \cite{zhang2022accounting} using the \citet{huser2019modeling} copula with a spatio-temporally varying marginal model similar to the one we use below. However, exploratory analysis shown in Figure~\ref{fig:application_chi_select} %the left-hand panel of Figure~%\ref{fig:emp_chi_moving_windows} 
% \ref{fig:application_chi} 
suggests that applying a single dependence class to such a large spatial domain is inappropriate, as there are areas in the central U.S. domain where the extreme summer precipitation appear to be locally AI and others that appear to be AD. Furthermore, when applying a single dependence class to the entire domain, \cite{zhang2022accounting} found that (overall) the extreme precipitation measurements are asymptotically \textit{independent}, meaning that the risk of concurrent extremes would be underestimated for sub-regions that exhibit asymptotic dependence. For comparison, we also approximately replicate the analysis of \citet{zhang2022accounting}, who used the \citet{huser2019modeling} copula on a similar dataset (results shown in Table \ref{table:application_knot} and Figure \ref{fig:application_ll}). Although we cannot make direct comparisons to the broadly similar \citet{hazra2021realistic} model because it does not permit spatial prediction at un-observed locations, our results from Section~\ref{sec:results} would suggest spatially heterogeneous transition between AD and AI (see Figure \ref{fig:application_surface}), and thus the \citet{hazra2021realistic} model would fit poorly because it only allows AD at short spatial lags.

% \begin{figure}[h]
%     \centering
%     \begin{subfigure}[t]{0.75\textwidth}%{0.542\textwidth}
%         \centering
%         \includegraphics[height=4.5cm]%[height=3.5cm]
%         {stations_train_and_test_combined.pdf}
%         \caption{}
%         \label{fig:application_US}
%     \end{subfigure}
%     \hfill
%     \begin{subfigure}[t]{0.45\textwidth}
%         \centering
%         \includegraphics[height=4.5cm]{stations13.pdf}
%         \caption{}
%         \label{fig:application_knot}
%     \end{subfigure}

%     \caption{(\subref{fig:application_US}) Location of the 590 observations/data sites and the 99 out-of-sample testing sites. 
%     (\subref{fig:application_knot}) An illustration of the mixture component setup for model $k13r4b4$. The blue dots represent observation sites, red `+' symbols represent the centers of the kernel functions, and the shaded circles represent areas covered by each Wendland kernel with specified radius $r$.}
%     \label{fig:application_combined}
% \end{figure}

\begin{figure}[h]
    \centering
  \begin{minipage}[c]{0.69\textwidth}
  \centering
    \begin{subfigure}[b]{\textwidth}
        \includegraphics[height=4.5cm]{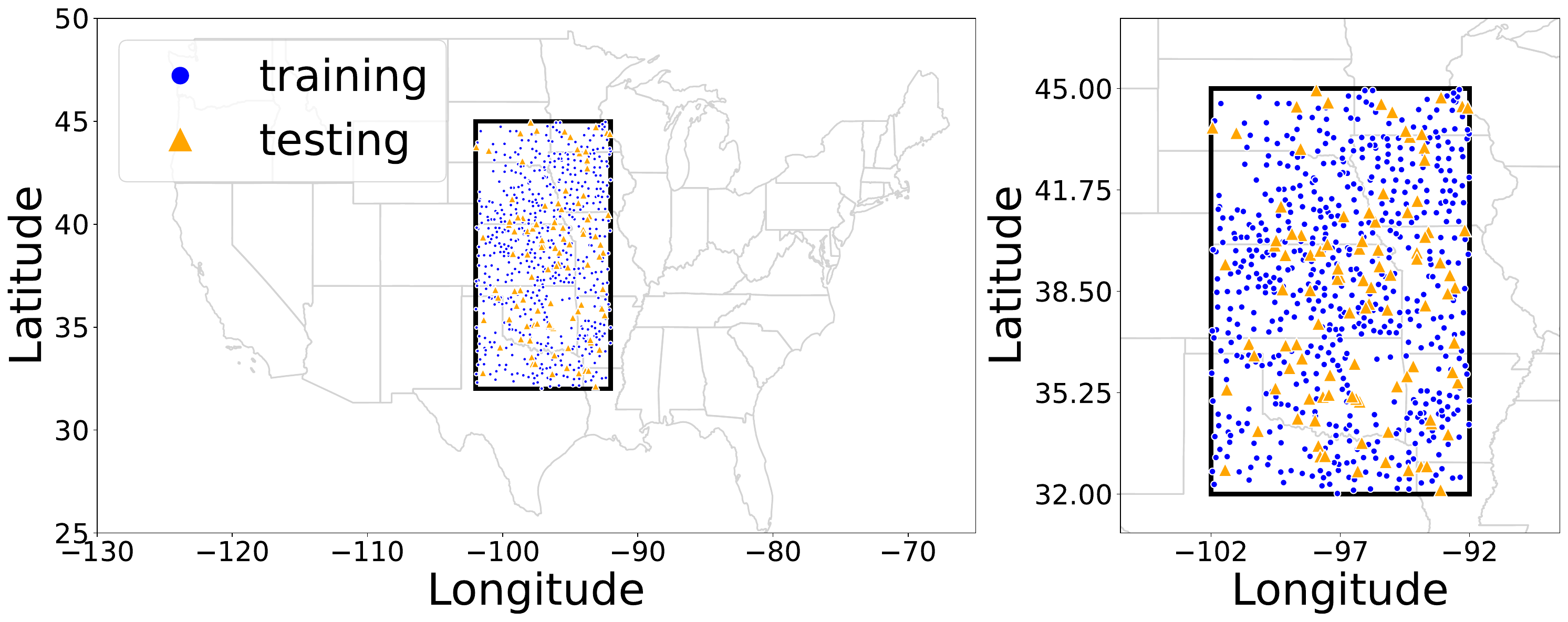}
        \caption{}
        \label{fig:application_US}
    \end{subfigure}
    \begin{subfigure}[b]{\textwidth}
        \centering
        \includegraphics[height=4.5cm]{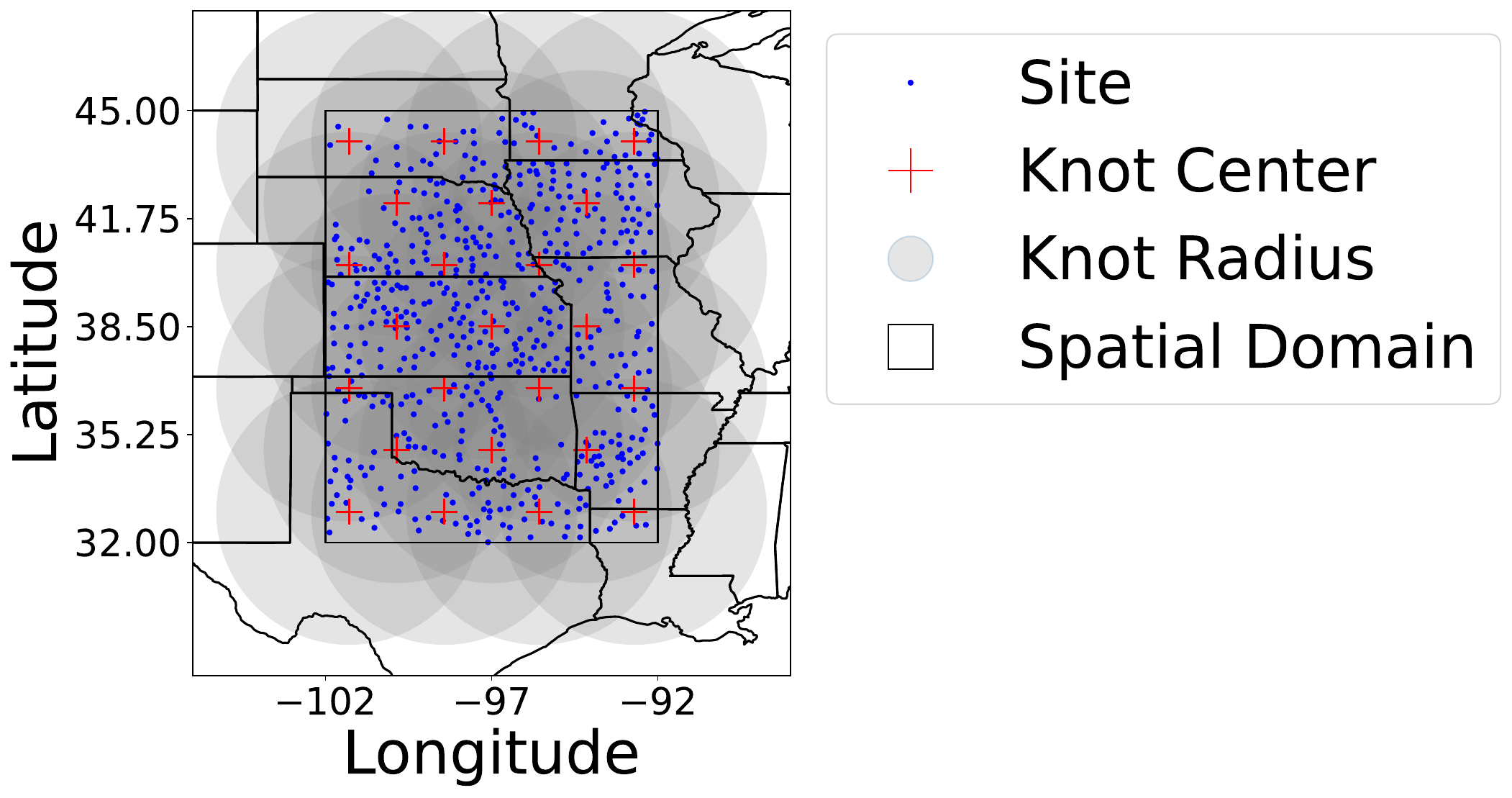}
        \caption{}
        \label{fig:application_knot}
    \end{subfigure}
  \end{minipage}
  \hspace{5pt}
  \begin{minipage}[c]{0.25\textwidth}
    \caption{(\subref{fig:application_US}) Location of the 590 observations/data sites and the 99 out-of-sample testing sites. 
    (\subref{fig:application_knot}) An illustration of the mixture component setup for model \texttt{k25r4b4}. The blue dots represent observation sites, red `+' symbols represent the centers of the kernel functions, and the shaded circles represent areas covered by each Wendland kernel with specified radius $r$.}
    \label{fig:application_combined}
  \end{minipage}
\end{figure}

While the GHCN database contains over twenty thousand stations over the contiguous United States, following \cite{risser2019probabilistic} we analyze a quality-controlled subset of the network from a recent 75-year period, namely those stations that have a minimum of 90\% non-missing daily precipitation measurements over 1949 through 2023. We then restrict our attention to summertime daily measurements (those from June, July, and August, or JJA) in the central United States region (defined by the $[102^\circ\text{W},92^\circ\text{W}]\times[32^\circ\text{N},45^\circ\text{N}]$ longitude-latitude bounding box), resulting in a set of $D=590$ stations. Our focus on the central United States in the summer season is intentional, since the majority of extreme precipitation in this region and season results from severe convective storms, which can be highly localized and therefore are particularly challenging to model statistically. Let $P_{tm}(\bs_j)$ represent the daily precipitation measurement in millimeters for day $m = 1, \dots, 92$ (the JJA season has 92 days) in year $t = 1949, \dots, 2023$ at station $\bs_j$. We then analyze the summertime maxima, denoted $Y_t(\bs_j) = \max_m\{P_{tm}(\bs_j)\}$. We only record the JJA maxima at station $\bs_j$ in year $t$ if that season has \revise{at least two-thirds non-missing measurements}, otherwise, $Y_t(\bs_j)$ is considered missing.

\revise{As we are analyzing block maxima, one might wonder whether each field of seasonal maxima contains observations from many individual storms.  Out of the 92 days in June-July-August, on average the station-specific maxima occur on approximately 75 different days (left panel of Figure~\ref{fig:event_timing}). The most extreme events occur on fewer days: 99th percentile events arise from approximately 5-8 days (orange curve in the right panel of Figure~\ref{fig:event_timing}). The fact that extreme precipitation arises from so many separate storms is not surprising; summertime precipitation in the central United States is primarily convective (i.e., from thunderstorms), which is characterized by highly-localized and short-lived storms. % This is why we originally chose to study extreme summertime rainfall in the central US in \cite{zhang2022accounting}: convective precipitation was challenging for our prior trend detection work \citep{Risser2019}. 
}

Corresponding to the block-maxima structured data, we assume that marginally 
\[
Y_t(\bs) \sim \text{GEV}(\mu_t(\bs), \sigma_t(\bs), \xi_t(\bs)).
\]
To ensure the independence over time and account for potential systematic increase or decrease in rainfall due to global warming, we have assumed a time-varying component for the location parameter. To account for the physical features of the terrain, we incorporate spatially-varying covariate with thin plate splines to smooth over the spatial domain. That is, 
\[
\mu_t(\bs) = \mu_0(\bs) + \mu_1(\bs) \cdot t, \quad \sigma_t(\bs) \equiv \sigma(\bs), \quad \xi_t(\bs) \equiv \xi(\bs),
\]
where the spatially varying intercept and slope parameters for the GEV location is specified as a spline as
\[
\mu_i(\bs) = \beta_0 + \beta_1 \cdot \text{elev}(\bs) + \sum_{i=1}^{11} w_i f_{tps}(\bs), \quad i \in \{0, 1\},
\]
where $f_{tps}$ denotes the thin-plate spline kernel (using 11 degrees of freedom to smooth over the spatial domain).  The GEV scale and shape parameters are specified as linear functions of elevation, as
\begin{align*}
    \log(\sigma(\bs)) &= \beta_0 + \beta_1 \cdot \text{elev}(\bs) \\
    \xi(\bs) &= \beta_0 + \beta_1 \cdot \text{elev}(\bs).
\end{align*}

\revise{For justification of these specific marginal modeling choices, please see Appendix \ref{sec:Appendix_GEV}.}

We place knots in an isomorphic grid across the spatial domain.  We considered thirteen different model configurations, with various combinations of different knot grids, Wendland kernel radii, Gaussian kernel bandwidths, and restricting the marginal parameters to be spatially constant. Table \ref{table:application_knot} describes the thirteen different setups, and Figure \ref{fig:application_knot} illustrates the mixture component setup of the model that we eventually chose. 
\begin{figure}[h]
\centering
    \includegraphics[width=\textwidth]{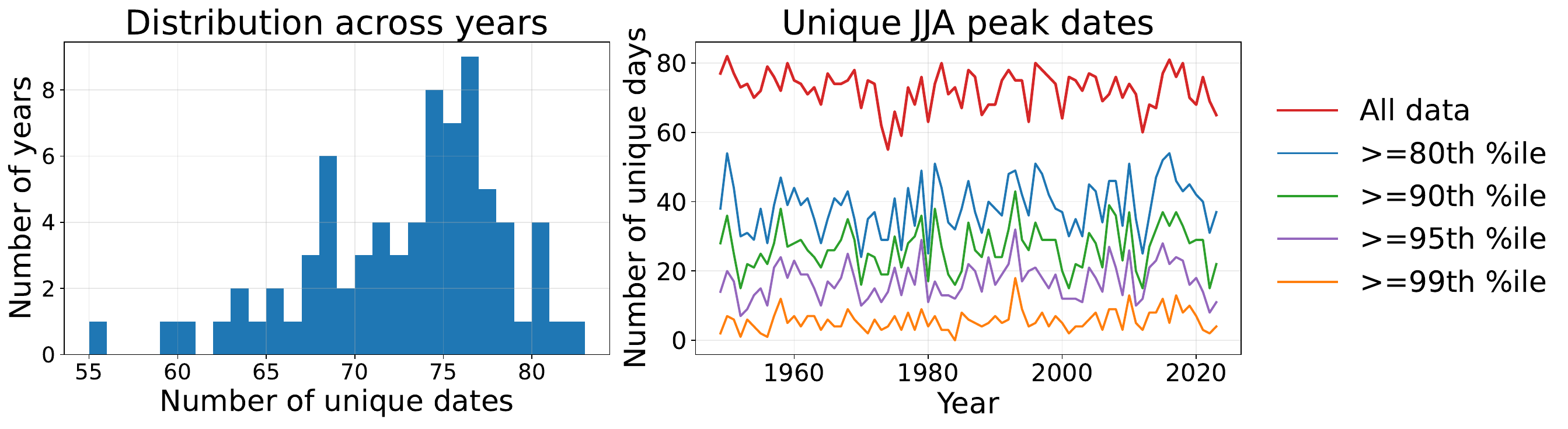}
\caption{\revise{Across 590 locations, histogram (left) of the unique dates of the yearly maxima across all sites; line charts (right) of the unique dates of the upper percentile maxima.}}
\label{fig:event_timing}
\end{figure}
To obtain MCMC starting values for the marginal parameters, we fit simple GEV regressions, assuming conditional independence across space and time. For the dependence model parameters, we fit empirical variograms to obtain initial values for the range parameter $\rho$, set $\phi(\bs)$ to start at 0.5 everywhere, and used the median observation to calculate a starting point for the random scale factor $R$. We then ran the MCMC chains for approximately 10,000 iterations (or until convergence) such that after discarded a burn-in period, we obtain at least 5,000 posterior samples.

\begin{table}[h]
\caption{13 model configurations spanning different knot grids, spatial extent of basis functions, and constraints on the marginal parameters. Model naming convention is as follows: $k$, $r$, $b$, and $m$ respectively denote the number of knots, radius of the compact Wendland  kernel, the bandwidth of the Gaussian  kernel, and restriction indicator on marginal GEV parameters. Here the effective range refers to the distance at which the Wendland kernel function becomes 0 (as it is compact) or when the Gaussian kernel function drops below 0.05. The \texttt{H-W Stationary} model is the stationary \cite{huser2019modeling} process used by \cite{zhang2022accounting}.}
\label{table:application_knot}

\rowcolors{2}{gray!15}{white} % alternate row color starting from second row

\centering
\renewcommand{\arraystretch}{1.2}
\begin{tabular}{
    >{\raggedright\arraybackslash}p{3.8cm}
    >{\centering\arraybackslash}p{1.2cm} 
    >{\centering\arraybackslash}p{1.2cm} 
    >{\centering\arraybackslash}p{1.2cm} 
    >{\centering\arraybackslash}p{1.2cm} 
    >{\centering\arraybackslash}p{1.2cm} 
    >{\raggedright\arraybackslash}p{2.7cm}}
\toprule
\textbf{Model Name} & \multicolumn{2}{c}{\textbf{\# of Knots}} & \multicolumn{3}{c}{\textbf{Basis Effective Range}} & \textbf{Constraint} \\
\cmidrule(lr){2-3} \cmidrule(lr){4-6}
& $S, \phi$ & $\rho$ & $S$ & $\phi$ & $\rho$ &  \\
\midrule
1. H-W Stationary     & 1   & 1   & $\infty$ & $\infty$ & $\infty$ & None \\
2. k13r4b4            & 13  & 13  & 4        & 4.89     & 4.89     & None \\
3. k13r4b4m           & 13  & 13  & 4        & 4.89     & 4.89     & Fix $\mu, \sigma, \xi$ \\
4. k25r2b0.67         & 25  & 25  & 2        & 2        & 2        & None \\
5. k25r2b0.67m        & 25  & 25  & 2        & 2        & 2        & Fix $\mu, \sigma, \xi$ \\
6. k25r2b2            & 25  & 25  & 2        & 3.46     & 3.46     & None \\
7. k25r2b2m           & 25  & 25  & 2        & 3.46     & 3.46     & Fix $\mu, \sigma, \xi$ \\
8. k25r4b4            & 25  & 25  & 4        & 4.89     & 4.89     & None \\
9. k25r4b4m           & 25  & 25  & 4        & 4.89     & 4.89     & Fix $\mu, \sigma, \xi$ \\
10. k41r1.6b0.43      & 41  & 41  & 1.6      & 1.6      & 1.6      & None \\
11. k41r1.6b0.43m     & 41  & 41  & 1.6      & 1.6      & 1.6      & Fix $\mu, \sigma, \xi$ \\
12. k41r2b0.67        & 41  & 41  & 2        & 2        & 2        & None \\
13. k41r2b0.67m       & 41  & 41  & 2        & 2        & 2        & Fix $\mu, \sigma, \xi$ \\
\bottomrule
\end{tabular}
\end{table}

\subsection{Model Evaluation}\label{sec:model evaluation}

To evaluate model fit, we incorporate additional observations as out-of-sample test data. These stations have $>85\%$ and $<= 90\%$ non-missing daily precipitation data measurements over the same time period and spatial domain. The JJA maxima are chosen according to the same criteria as the training dataset, enforcing at least 2/3 non-missing values in any given year.  This results in 99 stations in the test set, drawn as the yellow triangles in Figure \ref{fig:application_US}.
We evaluate  model fit and compare the thirteen models based on their predictive log-likelihoods at these 99 stations.

To obtain the log-likelihoods at the out-of-sample sites, \revise{we} draw from predictive distributions of all model parameters at each MCMC iteration. For models that have fixed marginal parameters, we use the initial estimates from the GEV fit to interpolate to the testing sites. Figure \ref{fig:application_ll} displays the boxplots of log-likelihoods at the testing sites for the thirteen models we considered. Based on its superior log-likelihood performance, we decided to use the \texttt{k25r4b4} model (Figure \ref{fig:application_ll}).

It appears that models which jointly estimate the marginal parameters within the MCMC tend to perform systematically better than models with fixed marginal parameters.  This suggests that the common practice of performing the analysis in two steps---first estimating the marginal parameters, then performing the dependence analysis ``downstream'' by plugging in the empirical marginal estimates to transform the data to convenient margins---is prone to underperform. We believe that it is worth the extra effort to estimate the marginal model parameters together with the dependence model parameters in one unified hierarchical model. 

\begin{figure}[h]
    \centering
  \begin{minipage}[c]{0.85\textwidth}
    \includegraphics[width=\textwidth]{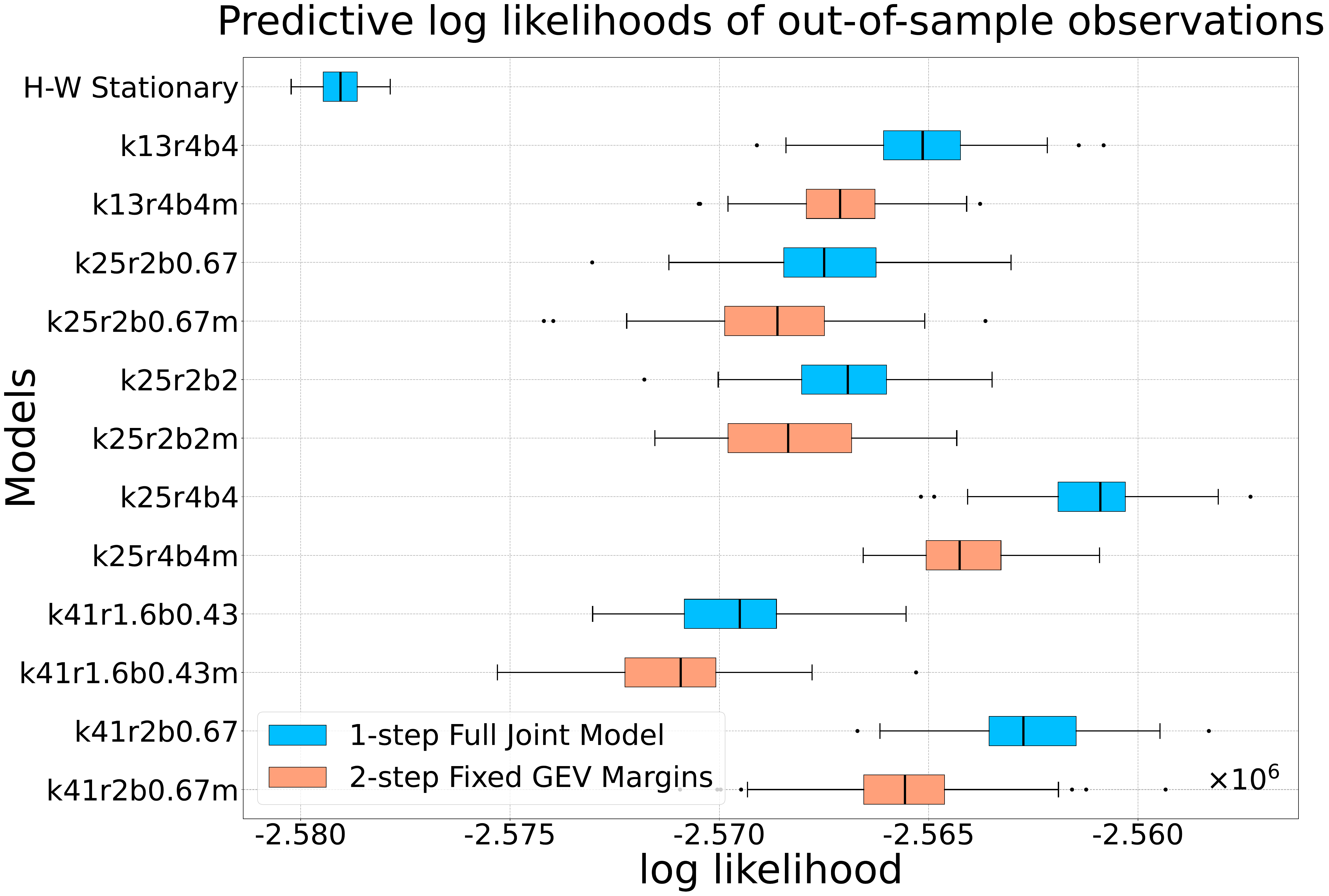}
  \end{minipage}\vfill
  \begin{minipage}[c]{\textwidth}
    \caption{Boxplots of predictive log-likelihood for the thirteen models. Higher log-likelihood is better. ``Blue'' models estimate the GEV parameters in the MCMC process, while the ``orange models'' have some restrictions on their marginal model parameters; ``blue models'' perform better than their ``orange counterparts''.} \label{fig:application_ll}
  \end{minipage}
\end{figure}

% We decided to use the $k13r4b4$ model based on its superior log-likelihood performance (Figure \ref{fig:application_ll}) and its parsimony. The $k25r2b0.67$ model shows comparable performance to our chosen model, but results in overly complex parameter surfaces.  We therefore favor the simpler $k13r4b4$ model.  In addition, the \textit{stationary} \citet{huser2019modeling}-equivalent model performs roughly comparably on average, but with much greater variation. This makes sense based on the posterior mean $\phi(\bs)$ surface (Figure \ref{fig:application_surface}), which varies in space but remains in the AI regime.  The \citet{huser2019modeling}-equivalent model can capture the overall behavior of the tail dependence strength but not its heterogeneity.  Thus, $k13r4b4$ is the most parsimonious candidate model that can still faithfully capture local variation in tail dependence.

\begin{figure}[h]
    \centering
    \begin{subfigure}[t]{0.24\textwidth}
        \centering
        \includegraphics[width=\textwidth, clip=true, trim=5 5 55 10]{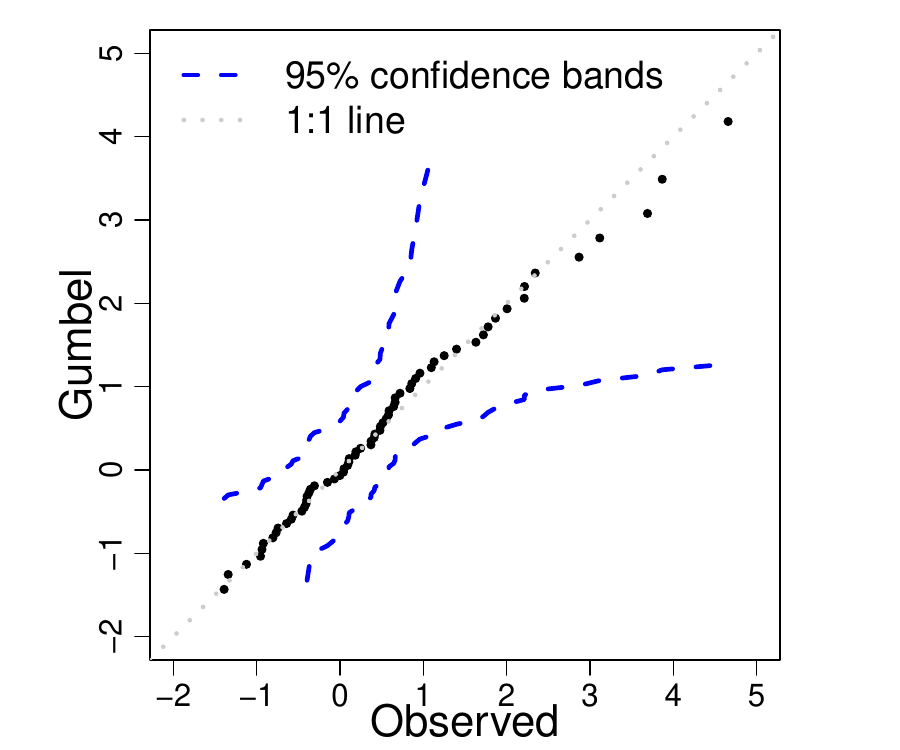}
        \caption{site 23}
        \label{fig:QQ_site_23}
    \end{subfigure}
    \hfill
    \begin{subfigure}[t]{0.24\textwidth}
        \centering
        \includegraphics[width=\textwidth, clip=true, trim=5 5 55 10]{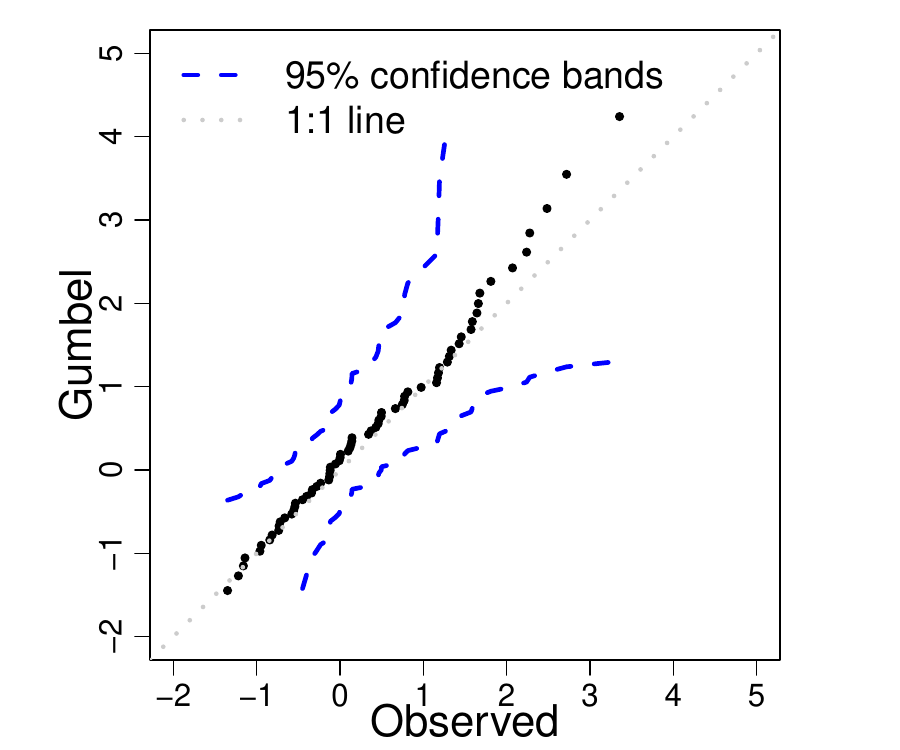}
        \caption{site 41}
        \label{fig:QQ_site_41}
    \end{subfigure}
    \hfill
    \begin{subfigure}[t]{0.24\textwidth}
        \centering
        \includegraphics[width=\textwidth, clip=true, trim=5 5 55 10]{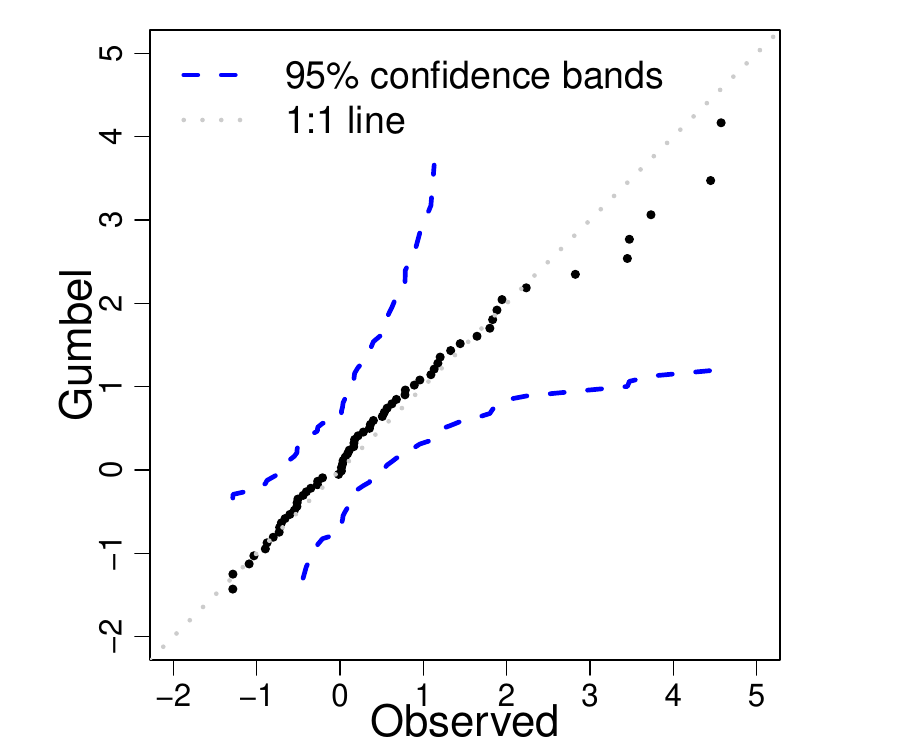}
        \caption{site 82}
        \label{fig:QQ_site_82}
    \end{subfigure}
    \hfill
    \begin{subfigure}[t]{0.24\textwidth}
        \centering
        \includegraphics[width=\textwidth, clip=true, trim=5 5 55 10]{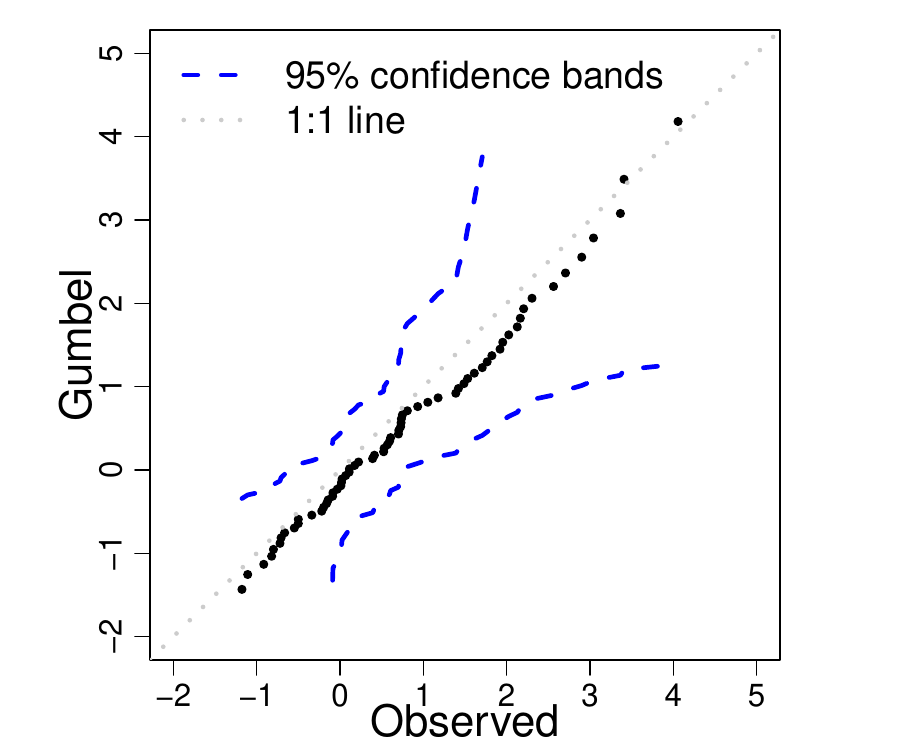}
        \caption{site 94}
        \label{fig:QQ_site_94}
    \end{subfigure}

    \caption{QQ-plots of four randomly selected holdout locations, comparing observed and predicted marginal quantiles for the \texttt{k25r4b4} model. 95\% confidence envelopes are also shown. Marginal values are transformed to a Gumbel distribution.}
    \label{fig:application_QQ}
\end{figure}

In addition, we looked at the empirical quantile plots of the out-of-sample observations to evaluate the model's marginal fit. We compared GEV quantiles based on posterior predictive draws of the marginal parameters at the holdout sites to the empirical quantiles at the same holdout sites. Figure \ref{fig:application_QQ} shows QQ-plots from four sites randomly selected from the 99 testing sites, for the best-performing \texttt{k25r4b4} model. The model provides a decent marginal fit, as the 95\% confidence band contains the 1:1 line in each case. 
% Some locations show better marginal fit than others, which is not a surprise since the joint likelihood must balance the joint and marginal fits, rather than exclusively optimizing the marginal fit at each location.

\subsection{Results}\label{sec:results}

We now present results from the chosen \texttt{k25r4b4} model. First, we examine the estimates for $\phi(\bs)$ at the knot locations, specifically whether they fall within $(0, 1/2]$ or $(1/2, 1)$, corresponding to whether the data-generating process is asymptotically independent or dependent. 
% It appears that this dataset demonstrates asymptotically independent behavior at all spatial ranges with high probability.  Recall that the \citet{hazra2021realistic} cannot capture this behavior.  Nonetheless, this is a somewhat disappointing result, as it does not properly highlight one of the key features of our model, which is that it simultaneously allows asymptotic dependence at short distances and asymptotic independence at long distances.  However, as the credible intervals for $\phi(\bs)$ show, there is non-negligible posterior probability that this is the case.  The tail dependence also appears to be non-stationary, as the estimated value of of $\phi(s)$ changes appreciably over the spatial domain; see Figure \ref{fig:application_surface}. 
It appears that this dataset exhibits asymptotically dependent behavior at short spatial ranges on the western portion of the spatial domain, and asymptotically independent behavior at short spatial ranges on the central and eastern portions of the spatial domain. Previous spatial models, including the \citet{hazra2021realistic} model, cannot capture this spatially heterogeneous behavior. This highlights one of the key features of our model, which is that it simultaneously allows  AI at long distances and either AI or AD at short distances. In this dataset, the tail dependence appears to be non-stationary, as the estimated value of the $\phi(\bs)$ changes appreciably over the spatial domain; see Figure \ref{fig:application_surface}.
% Next, the posterior mean $\rho(\bs)$ surface shows that the correlation range of the latent Gaussian process is estimated to be much longer in some parts of the spatial domain than others.  Specifically, southeast Minnesota/northeast Iowa and the Oklahoma panhandle show very short ranges.  
Next, the posterior mean $\rho(\bs)$ surface shows that the correlation range of the latent Gaussian process is also estimated to be variable across the spatial domain.
Table \ref{table:application_estimate} reports the posterior means and the 95\% equi-tail credible intervals for the dependence parameter $\phi(\bs)$ and range parameter $\rho(\bs)$ at the kernel knot locations.
Finally we include the posterior mean marginal parameter surfaces in Figure \ref{fig:application_surface}. Interestingly, the slope parameter $\mu_1(\bs)$ is positive in about half of the spatial domain and negative in the other half.  The largest values correspond to a change of 7 millimeters per century in the marginal GEV location.  
% The estimated GEV shape parameter surface shows very little spatial variation and is heavier-tailed that we would have expected.

% Posterior surface of phi and rho
% \begin{figure}[H]
%     \centering
%     \includegraphics[width=0.477\linewidth]{Surface_phi.pdf}
%     \includegraphics[width=0.34\linewidth, clip=true, trim=210 0 0 0]{Surface_rho.pdf}
    
%     \caption{Interpolated surfaces of posterior mean for dependence model of $k25r4b4$: tail dependence parameter $\phi(\bs)$ and range parameter $\rho(\bs)$.}
%     \label{fig:application_surface}
% \end{figure}

% \begin{figure}[h]
% \begin{adjustwidth}{-1.5cm}{}
% \centering
% \begin{minipage}{0.75\linewidth}
%     \centering
%     \includegraphics[width=0.571\linewidth]{Surface_phi.pdf}
%     \includegraphics[width=0.41\linewidth, clip=true, trim=210 0 0 0]{Surface_rho.pdf}
% \end{minipage}%
% \hfill
% \begin{minipage}{0.24\linewidth}
%     \captionof{figure}{Interpolated surfaces of posterior mean for dependence model of \texttt{k25r4b4}: tail dependence parameter $\phi(\bs)$ and range parameter $\rho(\bs)$. Dashed black line marks $\phi(\bs) = 0.5$, the transition between AD and AI.}
%     \label{fig:application_surface}
% \end{minipage}
% \end{adjustwidth}
% \end{figure}

\begin{figure}[h]
\centering
    \centering
    \hspace{-0.1\linewidth}
    \includegraphics[width=0.351\linewidth]{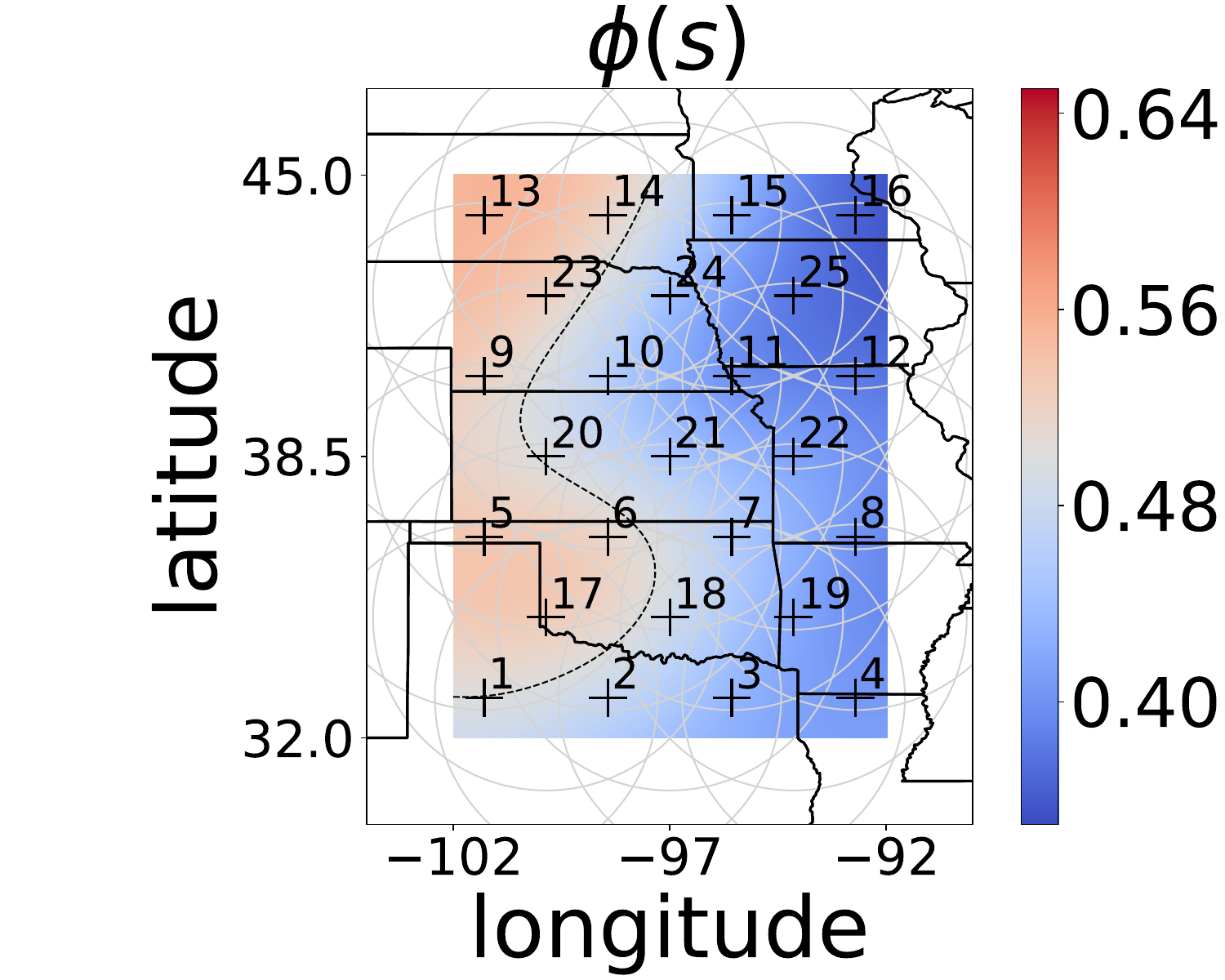}
    \includegraphics[width=0.289\linewidth, clip=true, trim=130 0 0 0]{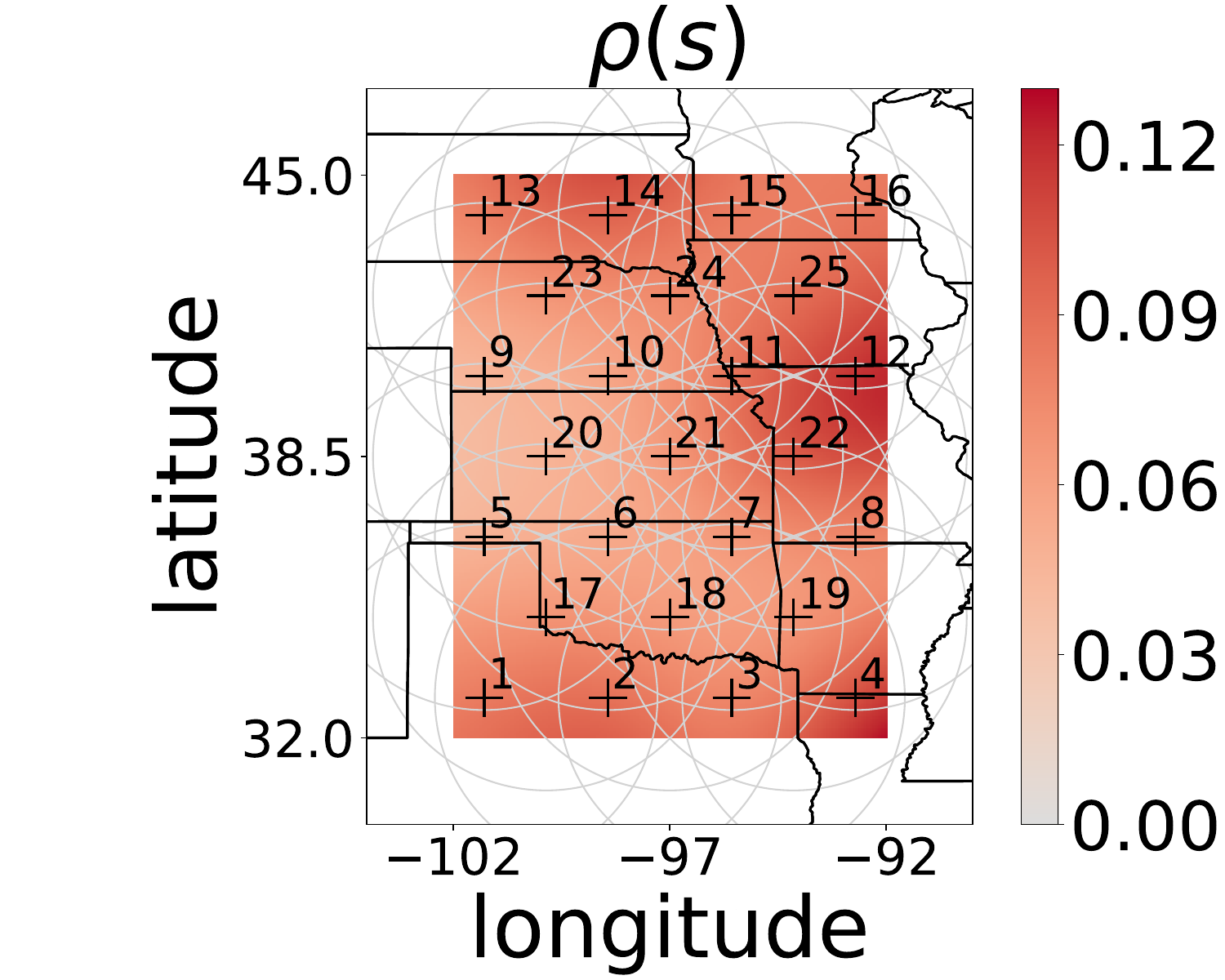}\linebreak
    \includegraphics[width=0.279\linewidth]{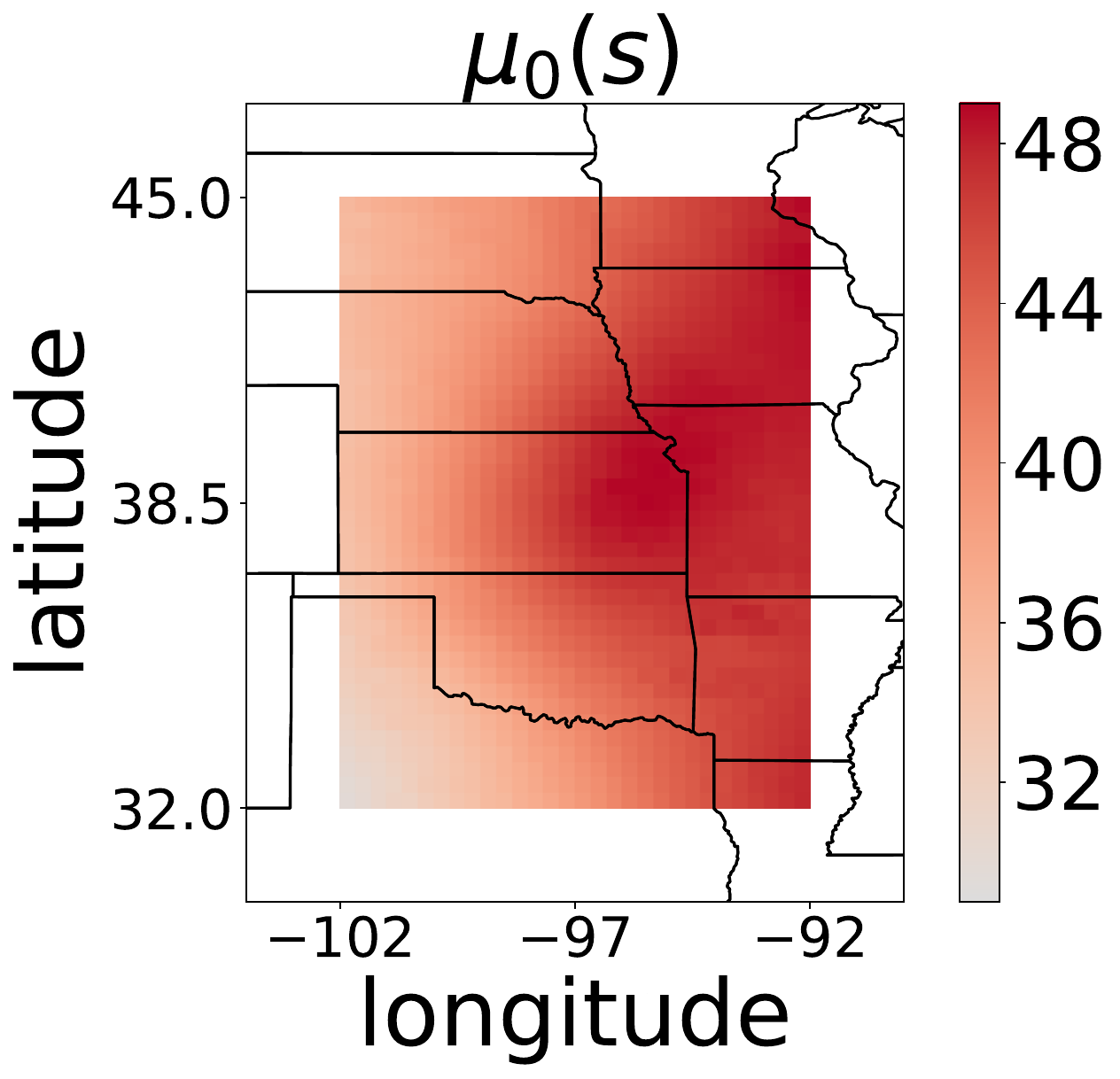}
    \includegraphics[width=0.225\linewidth, clip=true, trim=130 0 0 0]{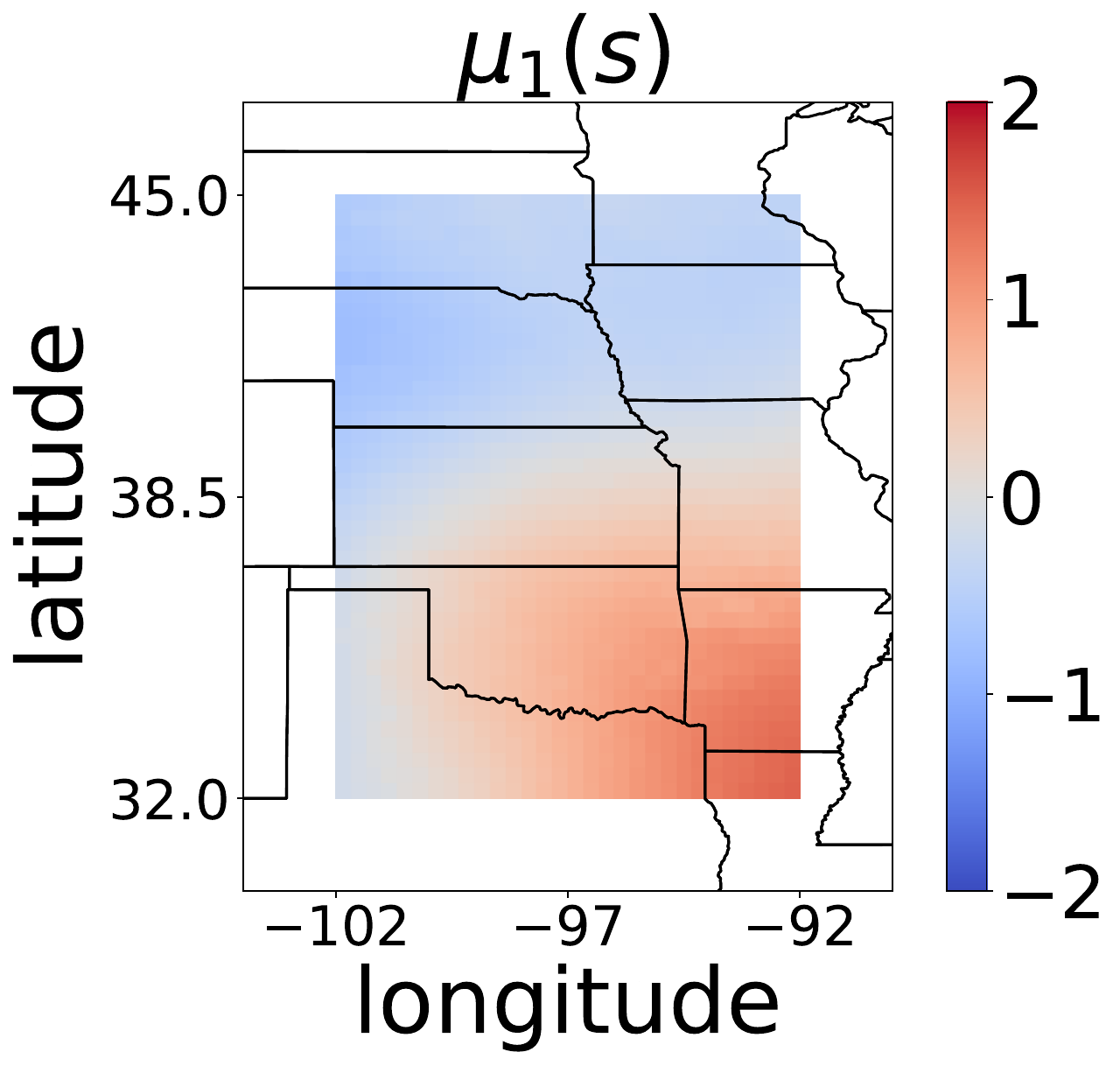}
    \includegraphics[width=0.225\linewidth, clip=true, trim=130 0 0 0]{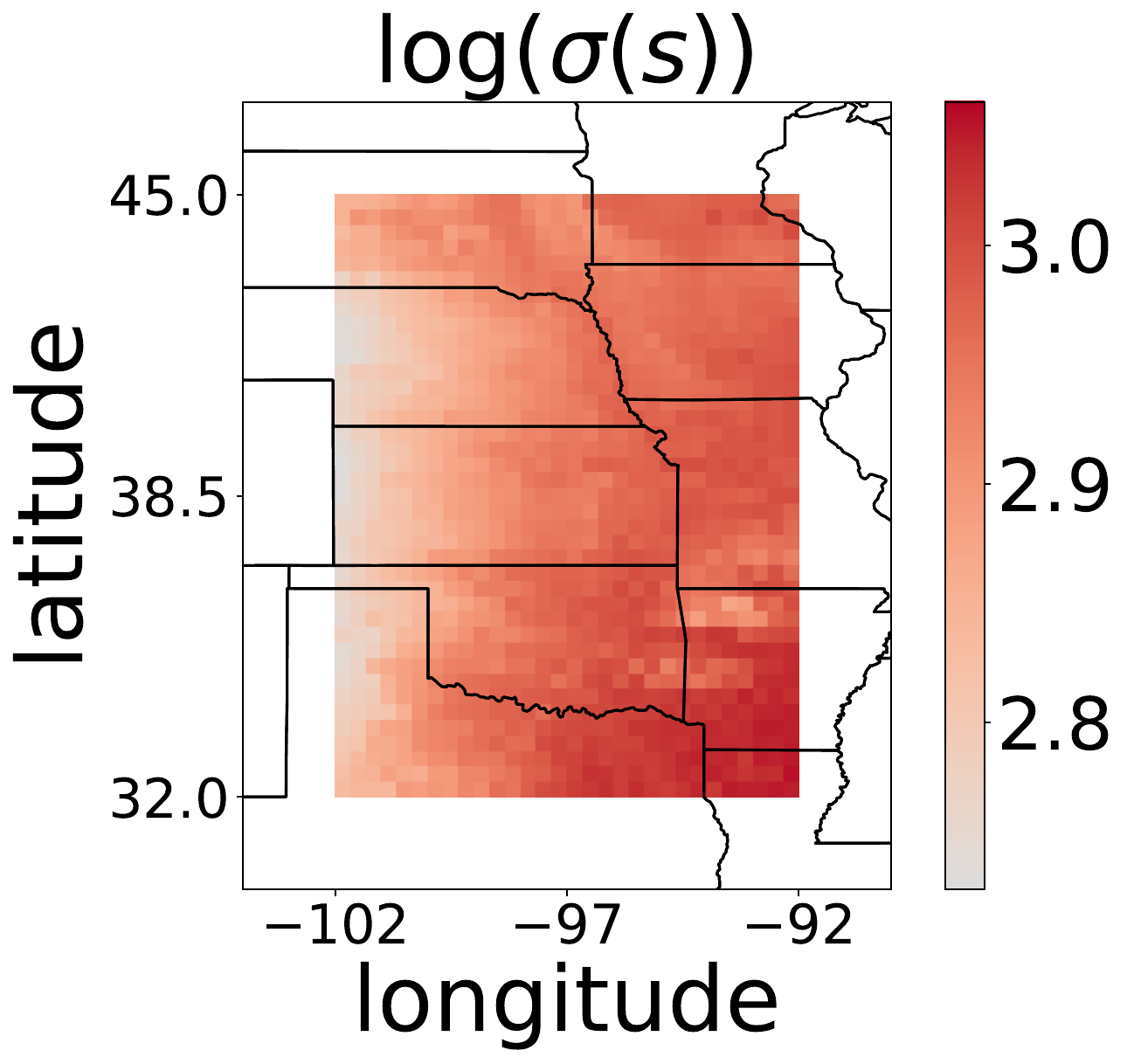}
    \includegraphics[width=0.247\linewidth, clip=true, trim=130 0 0 0]{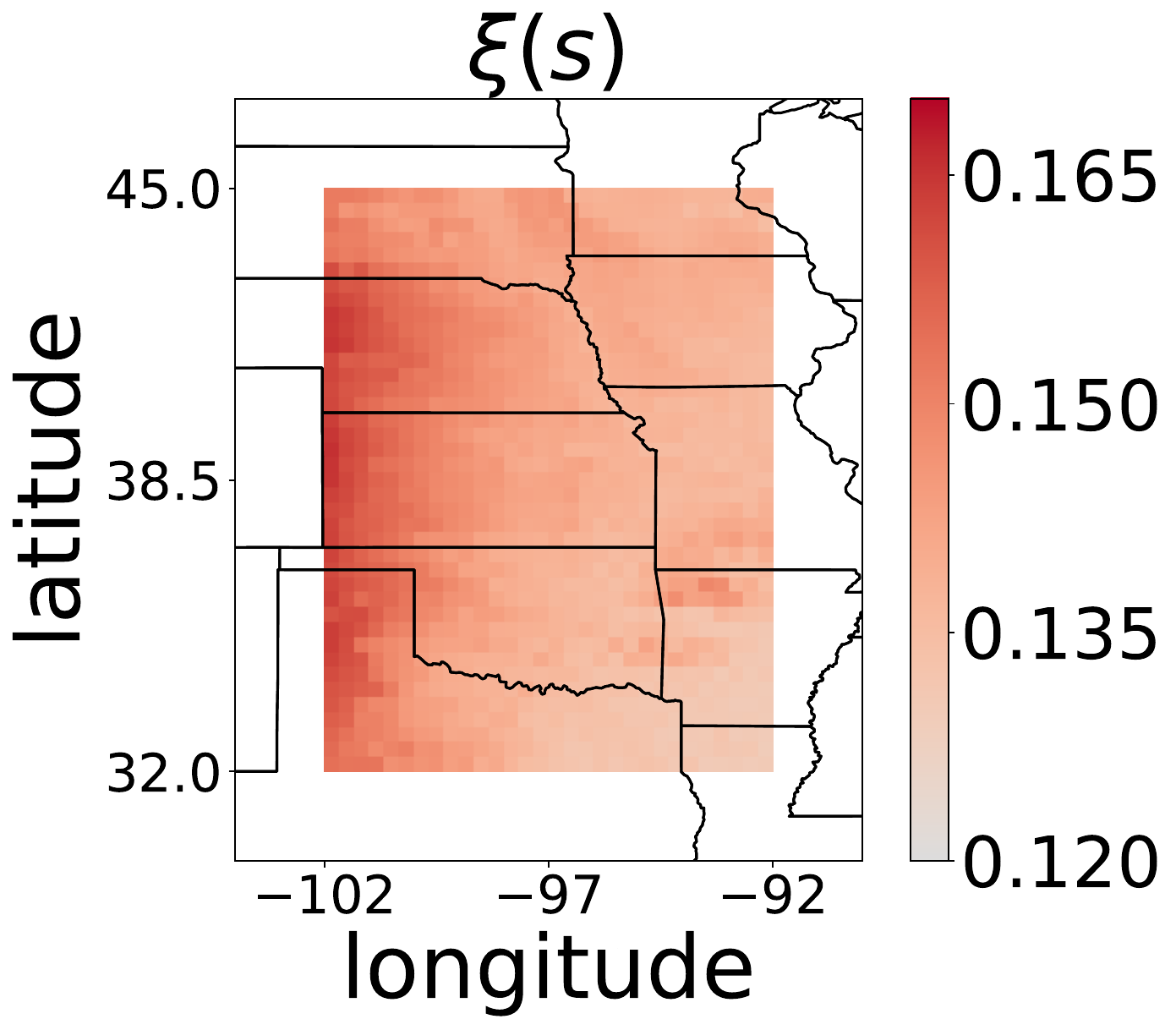}
    
    \caption{Interpolated surfaces of posterior mean for dependence and marginal models of \texttt{k25r4b4}. The tail dependence parameter $\phi(\bs)$ and range parameter $\rho(\bs)$ on top and the marginal parameters on bottom. For the $\phi(\bs)$ surface, the dashed black line marks $\phi(\bs) = 0.5$, the transition between AD and AI.}
    % \label{fig:application_GEVsurface}
    \label{fig:application_surface}
\end{figure}

\begin{table}[h]
\begin{adjustwidth}{-1.5cm}{-1.5cm}
\caption{Posterior mean and 95\% equi-tail credible interval of the \texttt{k25r4b4} model. $k$ indicates knot index.}
\label{table:application_estimate}
\centering
\small
\begin{tabular}{c|c|clccclccc}
\hline
$k$ & $\phi_k$                                                        & $\rho_k$                                                        &  & \multicolumn{1}{c|}{$k$} & \multicolumn{1}{c|}{$\phi_k$}                                                        & $\rho_k$                                                        &  & \multicolumn{1}{c|}{$k$} & \multicolumn{1}{c|}{$\phi_k$}                                                        & $\rho_k$                                                        \\ \cline{1-3} \cline{5-7} \cline{9-11} 
1   & \begin{tabular}[c]{@{}c@{}}0.469 \\ (0.271, 0.694)\end{tabular} & \begin{tabular}[c]{@{}c@{}}0.082 \\ (0.008, 0.168)\end{tabular} &  & \multicolumn{1}{c|}{2}   & \multicolumn{1}{c|}{\begin{tabular}[c]{@{}c@{}}0.394 \\ (0.184, 0.613)\end{tabular}} & \begin{tabular}[c]{@{}c@{}}0.147 \\ (0.044, 0.228)\end{tabular} &  & \multicolumn{1}{c|}{3}   & \multicolumn{1}{c|}{\begin{tabular}[c]{@{}c@{}}0.408 \\ (0.199, 0.616)\end{tabular}} & \begin{tabular}[c]{@{}c@{}}0.068 \\ (0.003, 0.165)\end{tabular} \\ \cline{1-3} \cline{5-7} \cline{9-11} 
4   & \begin{tabular}[c]{@{}c@{}}0.426 \\ (0.218, 0.672)\end{tabular} & \begin{tabular}[c]{@{}c@{}}0.178 \\ (0.091, 0.269)\end{tabular} &  & \multicolumn{1}{c|}{5}   & \multicolumn{1}{c|}{\begin{tabular}[c]{@{}c@{}}0.549 \\ (0.316, 0.764)\end{tabular}} & \begin{tabular}[c]{@{}c@{}}0.024 \\ (0.001, 0.066)\end{tabular} &  & \multicolumn{1}{c|}{6}   & \multicolumn{1}{c|}{\begin{tabular}[c]{@{}c@{}}0.619 \\ (0.355, 0.822)\end{tabular}} & \begin{tabular}[c]{@{}c@{}}0.035 \\ (0.002, 0.095)\end{tabular} \\ \cline{1-3} \cline{5-7} \cline{9-11} 
7   & \begin{tabular}[c]{@{}c@{}}0.500 \\ (0.276, 0.723)\end{tabular} & \begin{tabular}[c]{@{}c@{}}0.051 \\ (0.002, 0.145)\end{tabular} &  & \multicolumn{1}{c|}{8}   & \multicolumn{1}{c|}{\begin{tabular}[c]{@{}c@{}}0.346 \\ (0.168, 0.583)\end{tabular}} & \begin{tabular}[c]{@{}c@{}}0.034 \\ (0.001, 0.100)\end{tabular} &  & \multicolumn{1}{c|}{9}   & \multicolumn{1}{c|}{\begin{tabular}[c]{@{}c@{}}0.561 \\ (0.321, 0.775)\end{tabular}} & \begin{tabular}[c]{@{}c@{}}0.024 \\ (0.001, 0.072)\end{tabular} \\ \cline{1-3} \cline{5-7} \cline{9-11} 
10  & \begin{tabular}[c]{@{}c@{}}0.416 \\ (0.212, 0.639)\end{tabular} & \begin{tabular}[c]{@{}c@{}}0.036 \\ (0.001, 0.114)\end{tabular} &  & \multicolumn{1}{c|}{11}  & \multicolumn{1}{c|}{\begin{tabular}[c]{@{}c@{}}0.305 \\ (0.121, 0.509)\end{tabular}} & \begin{tabular}[c]{@{}c@{}}0.097 \\ (0.004, 0.254)\end{tabular} &  & \multicolumn{1}{c|}{12}  & \multicolumn{1}{c|}{\begin{tabular}[c]{@{}c@{}}0.366 \\ (0.172, 0.590)\end{tabular}} & \begin{tabular}[c]{@{}c@{}}0.161 \\ (0.034, 0.337)\end{tabular} \\ \cline{1-3} \cline{5-7} \cline{9-11} 
13  & \begin{tabular}[c]{@{}c@{}}0.538 \\ (0.301, 0.785)\end{tabular} & \begin{tabular}[c]{@{}c@{}}0.073 \\ (0.004, 0.187)\end{tabular} &  & \multicolumn{1}{c|}{14}  & \multicolumn{1}{c|}{\begin{tabular}[c]{@{}c@{}}0.588 \\ (0.319, 0.821)\end{tabular}} & \begin{tabular}[c]{@{}c@{}}0.177 \\ (0.049, 0.297)\end{tabular} &  & \multicolumn{1}{c|}{15}  & \multicolumn{1}{c|}{\begin{tabular}[c]{@{}c@{}}0.461 \\ (0.224, 0.693)\end{tabular}} & \begin{tabular}[c]{@{}c@{}}0.079 \\ (0.007, 0.185)\end{tabular} \\ \cline{1-3} \cline{5-7} \cline{9-11} 
16  & \begin{tabular}[c]{@{}c@{}}0.333 \\ (0.179, 0.499)\end{tabular} & \begin{tabular}[c]{@{}c@{}}0.083 \\ (0.021, 0.146)\end{tabular} &  & \multicolumn{1}{c|}{17}  & \multicolumn{1}{c|}{\begin{tabular}[c]{@{}c@{}}0.605 \\ (0.376, 0.827)\end{tabular}} & \begin{tabular}[c]{@{}c@{}}0.072 \\ (0.006, 0.162)\end{tabular} &  & \multicolumn{1}{c|}{18}  & \multicolumn{1}{c|}{\begin{tabular}[c]{@{}c@{}}0.528 \\ (0.289, 0.766)\end{tabular}} & \begin{tabular}[c]{@{}c@{}}0.042 \\ (0.001, 0.100)\end{tabular} \\ \cline{1-3} \cline{5-7} \cline{9-11} 
19  & \begin{tabular}[c]{@{}c@{}}0.359 \\ (0.161, 0.574)\end{tabular} & \begin{tabular}[c]{@{}c@{}}0.021 \\ (0.000, 0.062)\end{tabular} &  & \multicolumn{1}{c|}{20}  & \multicolumn{1}{c|}{\begin{tabular}[c]{@{}c@{}}0.429 \\ (0.213, 0.660)\end{tabular}} & \begin{tabular}[c]{@{}c@{}}0.059 \\ (0.008, 0.108)\end{tabular} &  & \multicolumn{1}{c|}{21}  & \multicolumn{1}{c|}{\begin{tabular}[c]{@{}c@{}}0.415 \\ (0.195, 0.627)\end{tabular}} & \begin{tabular}[c]{@{}c@{}}0.060 \\ (0.003, 0.137)\end{tabular} \\ \cline{1-3} \cline{5-7} \cline{9-11} 
22  & \begin{tabular}[c]{@{}c@{}}0.486 \\ (0.251, 0.741)\end{tabular} & \begin{tabular}[c]{@{}c@{}}0.176 \\ (0.028, 0.301)\end{tabular} &  & \multicolumn{1}{c|}{23}  & \multicolumn{1}{c|}{\begin{tabular}[c]{@{}c@{}}0.590 \\ (0.323, 0.829)\end{tabular}} & \begin{tabular}[c]{@{}c@{}}0.060 \\ (0.004, 0.146)\end{tabular} &  & \multicolumn{1}{c|}{24}  & \multicolumn{1}{c|}{\begin{tabular}[c]{@{}c@{}}0.426 \\ (0.208, 0.653)\end{tabular}} & \begin{tabular}[c]{@{}c@{}}0.031 \\ (0.001, 0.094)\end{tabular} \\ \cline{1-3}
25  & \begin{tabular}[c]{@{}c@{}}0.356 \\ (0.147, 0.589)\end{tabular} & \begin{tabular}[c]{@{}c@{}}0.065 \\ (0.002, 0.173)\end{tabular} &  &                          &                                                                                      &                                                                 &  &                          &                                                                                      &                                                                
\end{tabular}
\end{adjustwidth}
\end{table}

To assess the tail dependence of the fitted process, we use a moving-window approach, in which we divide  the spatial domain into $17 \times 7$ sub-regions and empirically estimate and visualize $\chi$ in each local window.  The left-hand panel of Figure \ref{fig:application_chi} shows the empirical $\hat{\chi}_u(h)$, while the right hand panel of Figure \ref{fig:application_chi} shows the analogous model-based $\hat{\chi}_u(h)$.  To obtain the empirical version, we first use the \texttt{k25r4b4} model-fitted marginal GEV parameters to transform the observations $Y_t(s)$ to the $X_t(s)$ scale, then empirically estimate $\hat{\chi}_u(h)$.  We do this across three quantile levels, $u = 0.9$, $0.95$ and $0.99$, and three  distances, approximately  $75$km, $150$km, and $225$km.  We obtained the analogous model-based version of $\hat{\chi}_u(h)$ by making many unconditional draws from the fitted model, then using the same empirical estimation strategy as in the left-hand panel.

% The empirical and model-based estimates of $\chi_u(h)$ show similar spatial patterns, but at different levels.  The most prominent feature is the patch of high values along the border between Texas and Oklahoma, as well as the more subtle patch of high values near the order between Nebraska, Iowa, and South Dakota. These features appear in both the left- and right-hand panels of Figure \ref{fig:application_chi}. 
% We also notice that as the quantile increases, the estimates $\chi_u(h)$  decrease for all spatial distances, which is the expected behavior for asymptotic independence.  This is consistent with the results shown in Figure \ref{fig:application_surface}, which shows that  $\phi(\bs) < 0.5$ across the entire domain. That the overall level of the model-based $\hat{\chi}_u(h)$ surfaces is higher than that of their empirical analogues, even as the spatial patterns and trends in $u$ are similar, is not terribly concerning; empirical estimates of $\hat{\chi}_u(h)$ are necessarily very noisy, so we would not expect perfect alignment, even if the model perfectly captured the data-generating process. 

The empirical and model-based estimates of $\chi_u(h)$ show similar spatial patterns at all levels. The most prominent feature is the patch of high values along the border between Texas and Oklahoma, as well as the patch of high values near the border between Nebraska and Missouri. These features appear in both the left- and right-hand panels of Figure \ref{fig:application_chi}.
We also notice that as the quantile increases, the estimates $\chi_u(h)$ decrease for all spatial distances, in the central and eastern region of the spatial domain, which is the expected behavior for asymptotic independence. This is consistent with the results shown in Figure \ref{fig:application_surface} for the regions where $\phi(\bs) < 0.5$.
That the overall level of some model-based $\hat{\chi}_u(h)$ surfaces is higher than that of their empirical analogues, even as the spatial patterns and trends in $u$ are similar, is not terribly concerning; empirical estimates of $\hat{\chi}_u(h)$ are necessarily very noisy, so we would not expect perfect alignment, even if the model perfectly captured the data-generating process. \revise{Confidence bands are included in the appendix \ref{sec:Appendix_chi}.}

% Empirical moving window chi
\begin{figure}[h]
    \centering
    
    \begin{minipage}[b]{0.474\textwidth}
        \centering
        Dataset $\chi$
        \vspace{5pt}
        
        \includegraphics[width=\textwidth, clip=true, trim=0 30 100 0]{Surface_data_chi_fittedGEV_h=75.pdf}
    \end{minipage}
    \hfill
    \begin{minipage}[b]{0.495\textwidth}
        \centering
        Model-based $\chi$
        \vspace{5pt}
        
        \includegraphics[width=\textwidth, clip=true, trim=65 30 0 0]{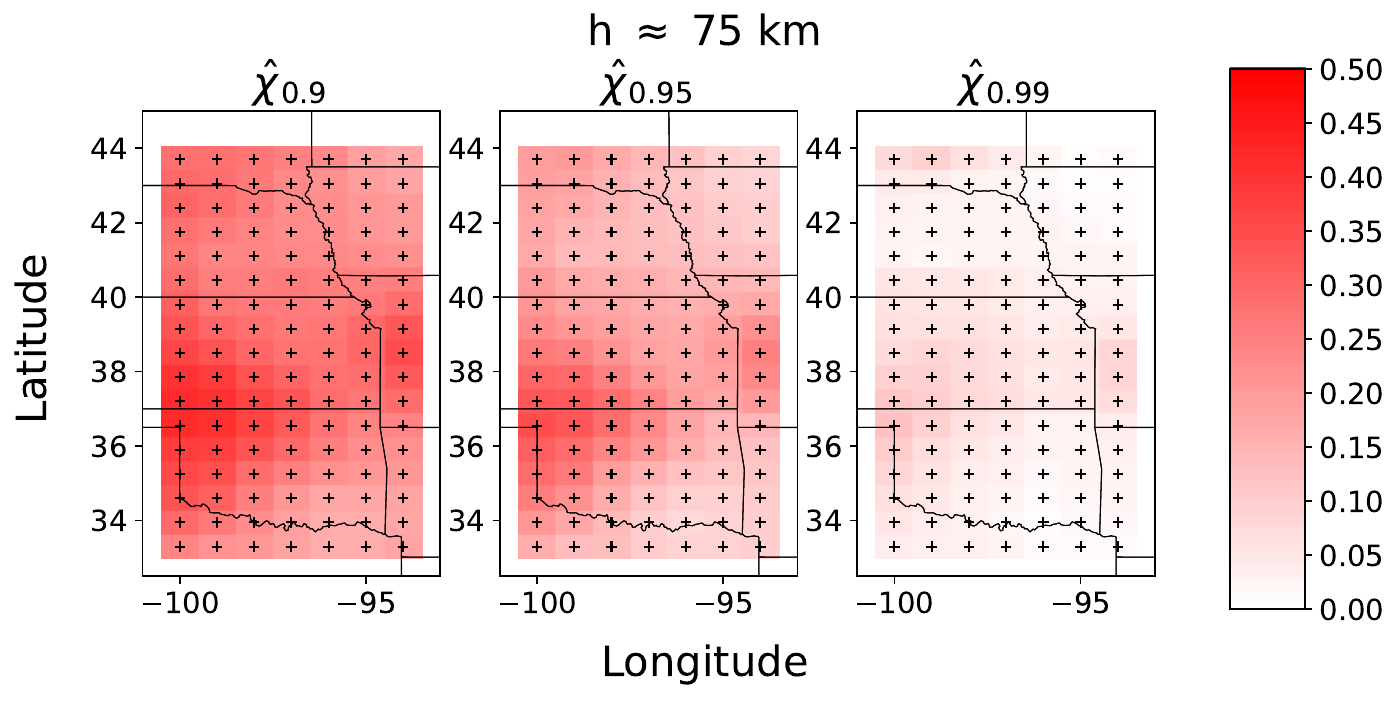}
    \end{minipage}
    \vfill
    \begin{minipage}[b]{0.474\textwidth}
        \centering
        \includegraphics[width=\textwidth, clip=true, trim=0 30 100 0]{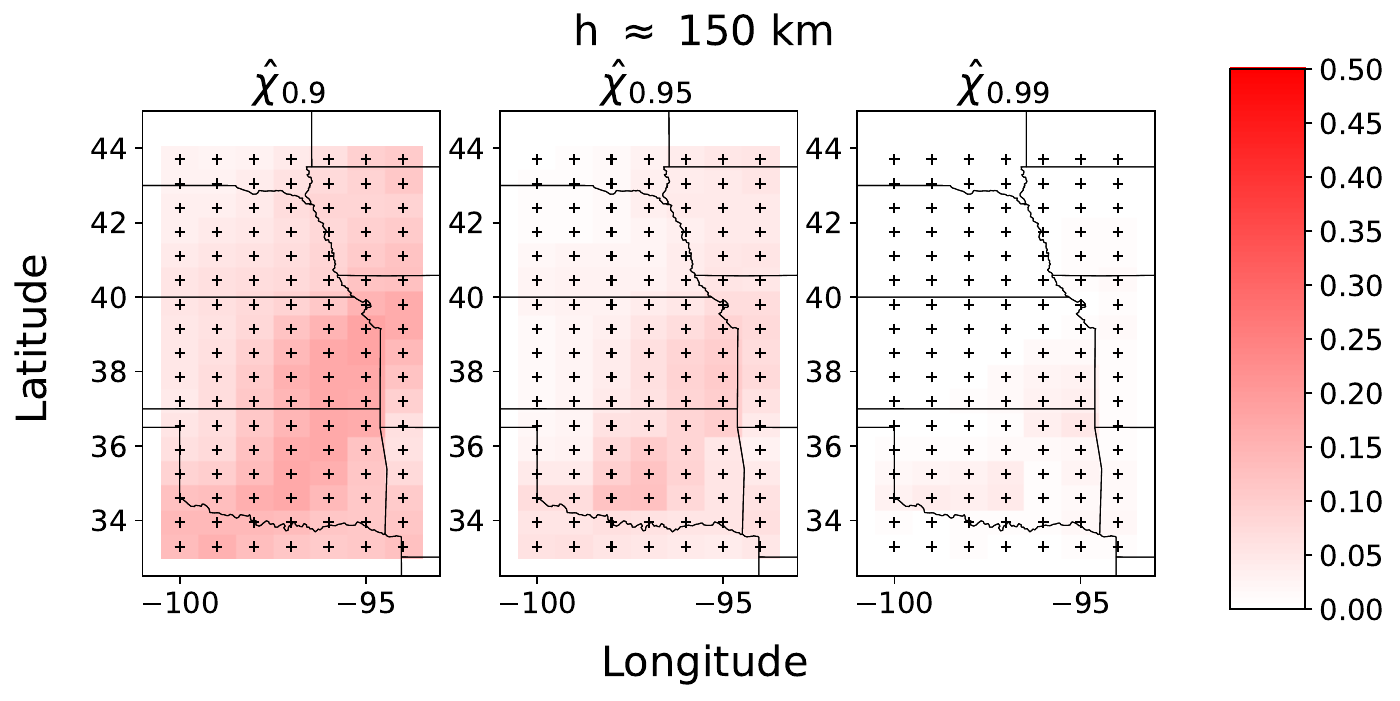}
    \end{minipage}
    \hfill
    \begin{minipage}[b]{0.495\textwidth}
        \centering
        \includegraphics[width=\textwidth, clip=true, trim=65 30 0 0]{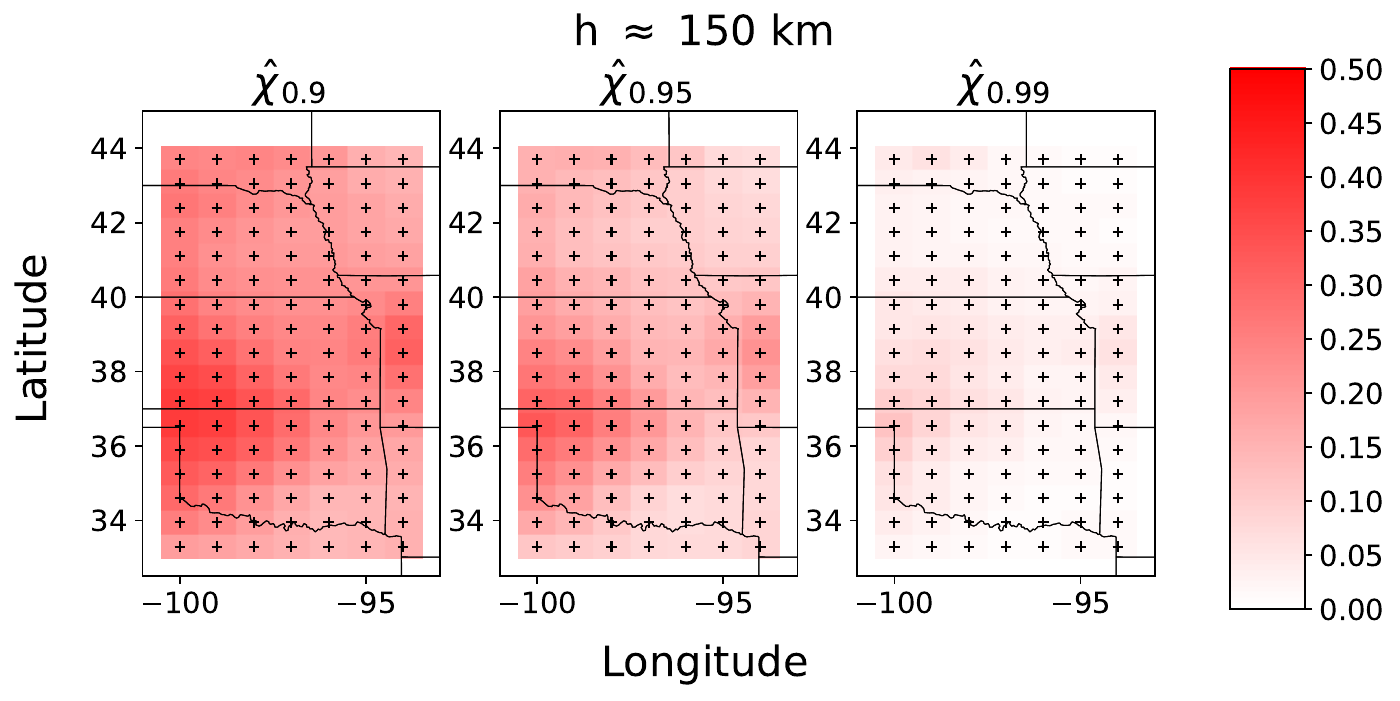}
    \end{minipage}
    \vfill
    \begin{minipage}[b]{0.474\textwidth}
        \centering
        \includegraphics[width=\textwidth, clip=true, trim=0 30 100 0]{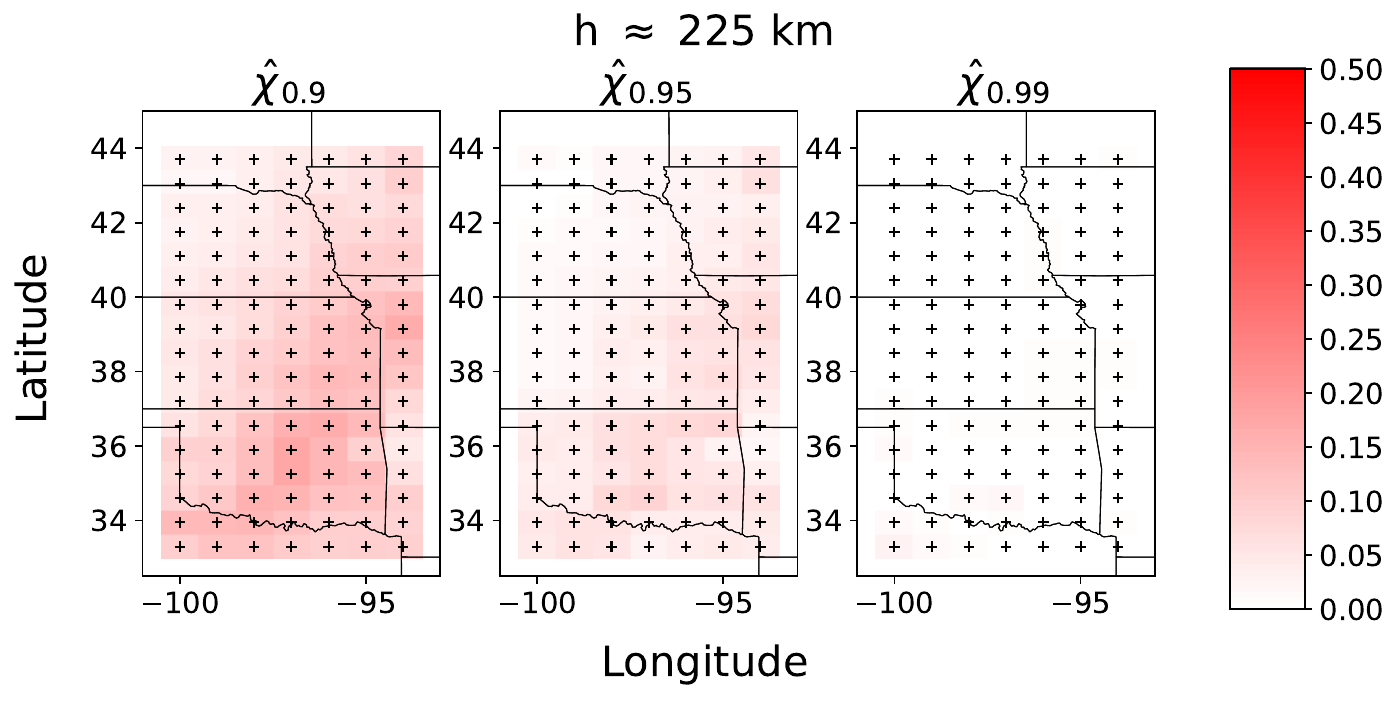}
    \end{minipage}
    \hfill
    \begin{minipage}[b]{0.495\textwidth}
        \centering
        \includegraphics[width=\textwidth, clip=true, trim=65 30 0 0]{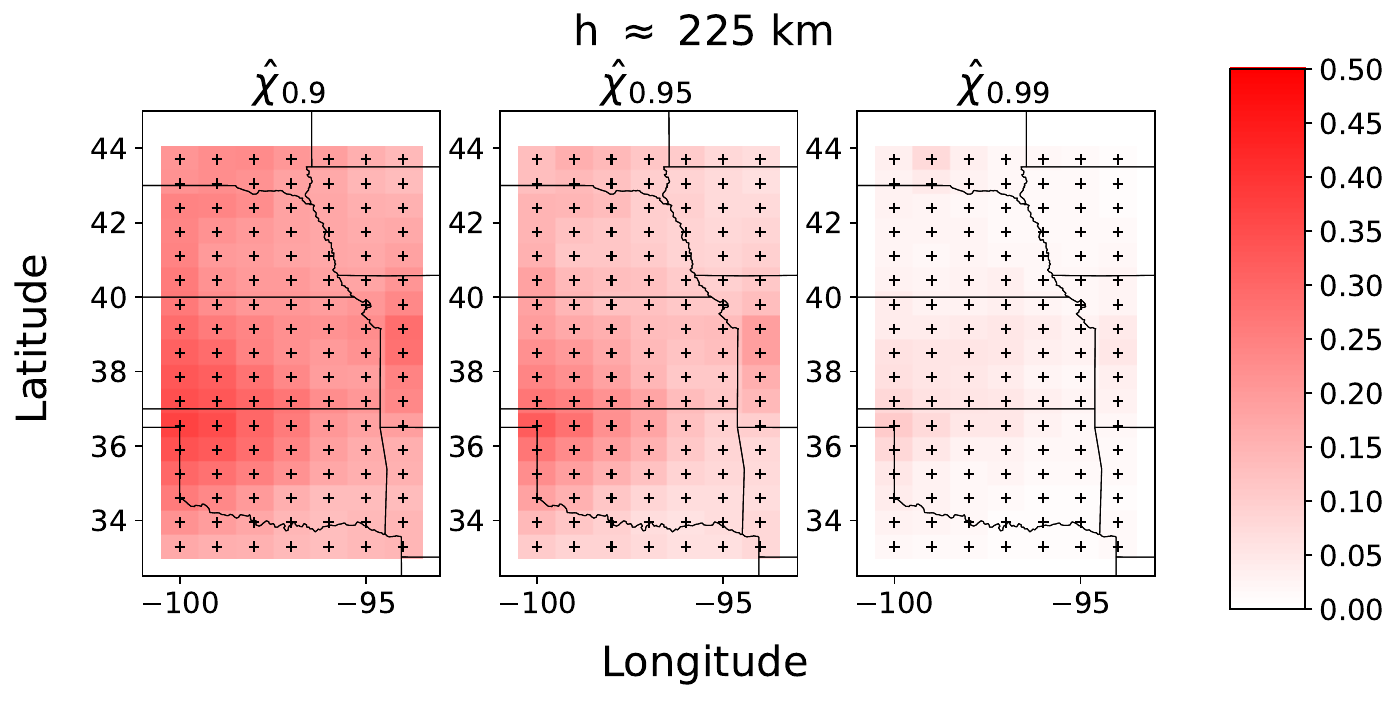}
    \end{minipage}
    \caption{Moving-window  estimates of $\chi_u(h)$  across three quantiles $u$ and three spatial lags $h$. The left-hand panel shows the dataset empirical estimates of $\chi_u(h)$, and the right-hand panel shows the model-based estimates of $\chi_u(h)$, based on the \texttt{k25r4b4} model.}
    \label{fig:application_chi}
\end{figure}

\section{Discussion}\label{sec:discussion}

In this article, we have proposed a modeling approach that extends the random scale construction to obtain more flexible local and long-range tail dependence behaviors. Our proposed mixture model is capable of simultaneously exhibiting asymptotic independence at long ranges and either asymptotic dependence or independence at short ranges.  The model is also able to capture non-stationary tail dependence structure with a spatially varying tail parameter that allows short-range asymptotic independence in some parts of the domain and short-range asymptotic dependence in other parts of the domain.
Our model permits trivial unconditional simulation, which easily allows for direct Monte Carlo risk estimates based on complicated functionals like areal sums or extrema.  It is also straightforward to do spatial prediction (i.e. interpolation) by drawing from relevant posterior predictive distributions.  And because our approach is fully Bayesian, all variation in any downstream analysis is accounted for in a coherent way.

Our analysis of extreme summertime precipitation in the central US highlighted some of the model's key features. We expected the data to exhibit \revise{short-range AI in some regions of the study area and short-range AD in other regions} . As expected, the posterior mean surface of the dependence parameter showed \revise{both AI and AD at short ranges} in parts of the spatial domain, and AI at long ranges. In addition, a desirable feature of a coherent flexible Bayesian model like ours is that any risk estimates based on posterior samples from our model will reflect the ambiguity in the local tail dependence regime.
% There was appreciable posterior probability that this was the case, but, somewhat disappointingly, the posterior mean surface of the dependence parameter showed asymptotic independence at all ranges, across the entire spatial domain.  Nevertheless, a desirable feature is that any risk estimates based on posterior samples from our model will reflect the ambiguity in the local tail dependence regime.

An interesting side result that merits further study was that models which jointly estimated marginal and dependence parameters always performed better (in holdout set predictive log-likelihoods) than analogous ``two-step'' models in which marginal parameters were estimated and plugged in before fitting the dependence model. 
% We believe this phenomenon merits further investigation. 
Jointly estimating the marginal and dependence models as we have done incurs large costs in terms of software implementation and computational complexity.  However, it seems to confer significant advantages in terms of model fit, in addition to naturally propagating variation between marginal and dependence models, and on to predictions.

Fitting our model with MCMC allows inference on spatial extreme-value datasets with relatively large numbers of locations. We have defined the model conditionally as a Bayesian hierarchical model, for which standard MCMC techniques can be used. Computation is facilitated by paralleling over time and migrating any required numerical integration to \texttt{C++}. Even so, the lack of closed form marginal transformations creates a significant computational challenge that scales with the total number of observations. This computational bottleneck would be even more evident if using a Generalized Pareto response for analyzing peaks-over-threshold data. While including a nugget term as in \citet{zhang2021hierarchical} would alleviate much of the computation, it still adds an additional layer of numerical integration to the marginal transformation, which is already not in closed form.

\revise{It would be useful to find ways to alleviate the heavy computational burden required to fit our full model.  Neural Bayes estimators in the style of \citet{saintsbury-dale-2024a} are likely not effective in our context because of the high dimensionality of our parameter space (i.e. thousands of parameters).  Most of the computational effort in our sampler arises from the numerical integration required for joint estimation of the marginal and dependence components, which we identify as a target for future improvements, possibly including the use of neural emulators.}

Finally, while many of the key features of our proposed model are also found in recent ``single-site conditioning'' models \citep{wadsworth2022higher}, to our knowledge, ours is the first fully-specified joint model to possess them.  In some ways, our model can be seen as an alternative to \citet{wadsworth2022higher}-type models.  These single-site conditioning models are more parsimonious and faster to fit than what we have proposed here.  However, that our model is a well-defined joint probability model gives it significant advantages in both interpretability and applicability.

\section{Disclosure statement}\label{disclosure-statement}

No competing interest is declared.

\section{Data Availability Statement}\label{data-availability-statement}

Deidentified data have been made available at the following URL: \url{https://github.com/muyangshi/GEV_ScaleMixture/}.

{
% \spacingset{1.2}
% \setlength{\bibsep}{1.5pt}
\addcontentsline{toc}{section}{References}
\bibliography{main}
}

% \addcontentsline{toc}{section}{References}
%   \bibliography{main.bib}

\appendix
\section*{Appendix}

\section{Technical proofs} \label{sec:Appendix Proof}
\subsection{Properties of stable distribution}\label{sec:properties_stable}

% \LZadd{
The Stable distribution is important in both theory and application because it is the generalized central limit of random variables without second (or even first) order moments. %It is widely used in signal processing in the presence of impulsive interference \citep{nikias1995signal}. 
Assume $S\sim \text{Stable}(\alpha, \beta, \gamma, \delta)$ under the 1-parameterization \citep{nolan2020univariate} where $\alpha\in(0,2]$ is the concentration parameter, $\beta\in[-1,1]$ is the skewness parameter, $\gamma>0$ is the scale parameter and $\delta\in\mathbb{R}$ is the location parameter. Then $S$ has the characteristic function
\begin{equation*}
    E\exp(iuS) = \exp[-\gamma^\alpha|u|^\alpha\{1-i\beta\omega(u)\text{sign}(u)\}+i\delta u],\; u\in \mathbb{R},
\end{equation*}
where $\text{sign}(u)$ is the sign of $u$ and $\omega(u) = \tan(\pi\alpha/2)\mathbbm{1}(\alpha\neq 1)-2\pi^{-1}\log|u|\mathbbm{1}(\alpha= 1)$.  If $\alpha<1$ and $\beta=1$, $S$ has support  {$[\delta,\infty]$}. %; see Figure~\ref{fig:stabledist}.
Moreover, if $\alpha<2$ and $0<\beta\leq 1$, the tail of $S$ is Pareto-like and satisfies $\pr(S>x)\sim \gamma^{\alpha}(1+\beta)C_{\alpha} x^{-\alpha}$ with $C_{\alpha}=\Gamma(\alpha)\sin(\alpha\pi/2)/\pi$ as $x\rightarrow\infty$.
% }

% \begin{figure}
%     \centering
%     \includegraphics[width=0.4\linewidth]{figures/stabledist.png}
%     \caption{Stable distribution density, fixing $\alpha=0.5$, $\beta=1$ and $\delta=0$ while varying the scale parameter $\gamma$.}
%     \label{fig:stabledist}
% \end{figure}

% \LZadd{
More importantly, the Stable distributions are closed under convolution; sums of $\alpha$-Stable variables (Stable variables with concentration parameter $\alpha$) are still $\alpha$-Stable. If $S_k$ $\stackrel{\text{indep}}{\sim}$ $\text{Stable}(\alpha,$ $\beta_k,$ $\gamma_k,$ $\delta_k)$ and constant $w_k\geq 0$ for $k=1,\ldots, K$, then
\begin{equation}
    \sum_{k=1}^K w_{k} S_k 
     \sim \text{Stable}(\alpha,\bar{\beta},\bar{\gamma},\bar{\delta})
\end{equation}
with $\bar{\gamma} = \{\sum_{k=1}^K(w_{k}\gamma_k)^\alpha\}^{1/\alpha}$, $\bar{\beta} =\sum_{k=1}^K\beta_k(w_{k}\gamma_k)^\alpha/\bar{\gamma}^{\alpha}$ and $\bar{\delta}=\sum_{k=1}^K w_{k}\delta_k$. 
% }

% \LZadd{
To be able to examine the joint distribution of $(X_i, X_j)$ for model \eqref{eqn:model}, it is desirable for the mixture to have the same distributional support and rate of tail decay as each $S_k$. Thus in \eqref{eqn:scaling_process}, we fixed $\beta_k\equiv 1$, $\delta_k\equiv \delta$ while imposing the constraint $\sum_{k=1}^K w_k=1$. As a result, 
$\bar{\beta} =1$ and $\bar{\delta}=\delta$, which means the univariate support of the $\{R(\bs)\}$ is $[\delta,\infty)$ everywhere.  In consequence, we have $\pr(R(\bs)>x)\sim 2\bar{\gamma}^{\alpha}(\bs)C_{\alpha} x^{-\alpha}$ for all $\bs\in \mathcal{S}$ as $x\rightarrow\infty$ with $\bar{\gamma}(\bs)=[\sum_{k=1}^K\{w_{k}(\bs,r_k)\gamma_k\}^\alpha]^{1/\alpha}$, which means the process $\{R(\bs)\}$ has tail-stationarity.
% }

\subsection{Proof of Proposition~\ref{prop:marginal_distr}}
We begin by recalling a couple of useful theoretical results.
\begin{lemma}[\citet{breiman1965some}]\label{lem:breiman} \upshape
Assume $X_1$ and $X_2$ are two independent random variables that are both supported by $\mathbb{R}^+$, and that $\pr(X_1>x)\in RV_{-\alpha}, \alpha\geq 0$.
\begin{enumerate}[(a)]
    \item\label{breiman} If there exists $\epsilon>0$ such that $E(X_2^{\alpha+\epsilon})<\infty$, then
    \begin{equation}\label{eqn:breiman}
        \pr(X_1X_2>x)\sim E(X_2^\alpha)\pr(X_1>x).
    \end{equation}
    \item Under the assumptions of part \eqref{breiman}, we have
    \begin{equation*}
        \sup_{x\geq y}\left|\frac{\pr(X_1X_2>x)}{\pr(X_1>x)}-E(X_2^\alpha)\right|\rightarrow 0, \; y\rightarrow \infty.
    \end{equation*}
    \item\label{easy_breiman} If $\pr(X_1>x)\sim cx^{-\alpha}$, \eqref{eqn:breiman} holds under $E(X_2^\alpha)<\infty$.
    % \item If $\pr(X_2>x)=o(\pr(X_1X_2>x))$, then $\pr(X_1X_2>x)\in RV_{-\alpha}$.
    \item If $\pr(X_2>x)=o(\pr(X_1>x))$, then $\pr(X_1X_2>x)\in RV_{-\alpha}$.
\end{enumerate}
\end{lemma}

\begin{lemma}[Theorem 4(v) of \citet{cline1986convolution}]\label{lem:ET_same_convolution} \upshape
Let $X_1\sim F_1$ and $X_2\sim F_2$ be two random variables that are both exponential tailed with the same rate, i.e. $F_i\in \text{ET}_{\alpha,\beta_i}$, $\alpha>0$, $\beta_i>-1$, $i=1,2$. Then
\begin{equation*}
    \pr(X_1+X_2>x)\sim \alpha\frac{\Gamma(\beta_1+1)\Gamma(\beta_2+1)}{\Gamma(\beta_1+\beta_2+1)}l(x)x^{\beta_1+\beta_2+1}\exp(-\alpha x),
\end{equation*}
where $l(\cdot)$ is slowly varying.
\end{lemma}

\begin{proof}[Proof of Proposition~\ref{prop:marginal_distr}]
When $0<\phi_j<\alpha$, the $\phi_j$th moment of $R_j$ exists: $E(R_j^{\phi_j})<\infty$. This is sometimes called the fractional lower order moment. Recall that marginally $\pr(W_j>x)=x^{-1}$. Thus applying Lemma~\ref{lem:breiman} part~\ref{easy_breiman} to $R_jW_j$ yields $\pr(R_j^{\phi_j}W_j>x)=E(R_j^{\phi_j})x^{-1}$.

When $\phi_j>\alpha$, $E(R_j^{\phi_j})=\infty$. However, the tail of $R_j$ is still regularly varying:
\begin{equation*}
    \pr(R_j^{\phi_j}>x)\sim 2C_{\alpha}\bar{\gamma}^\alpha_{j} x^{-\frac{\alpha}{\phi_j}}
\end{equation*}
and
\begin{equation*}
    E\left(W_j^\frac{\alpha}{\phi_j}\right)=\int_1^\infty w^{\frac{\alpha}{\phi_j}-2}dw=\frac{1}{1-\alpha/\phi_j}.
\end{equation*}
From Lemma~\ref{lem:breiman}\ref{easy_breiman},
\begin{equation*}
    \pr(R_j^{\phi_j}W_j>x)\sim E\left(W_j^\frac{\alpha}{\phi_j}\right)\pr(R_j^{\phi_j}>x),
\end{equation*}
from which the result follows.

When $\phi_j=\alpha$, $R_j^{\phi_j}$ and $W_j$ are regularly varying with the same index $-1$. Therefore, $\log R_j^{\phi_j}$ and $\log W_j$ are both in $ET_{1,0}$. By Lemma \ref{lem:ET_same_convolution},
\begin{equation*}
    \pr(R_j^{\phi_j}W_j>x)\sim 2C_{\alpha}\bar{\gamma}^\alpha_{j} x^{-1}\log x.
\end{equation*}
\end{proof}

\begin{remark} \upshape
By Theorem 3.8 and Lemma 3.12 of \citet{nolan2020univariate}, we can write out the exact form of the fractional lower order moment of $R_j\sim \text{Stable}(\alpha,1,\bar{\gamma}_{j},\delta)$ if $\delta=0$:
\begin{equation*}
    E(R_j^{\phi_j})=\bar{\gamma}_{j}^{\phi_j} \cos^{-\frac{\phi_j}{\alpha}}\left(\frac{\pi\alpha}{2}\right)\frac{\Gamma(1-\phi_j/\alpha)}{\Gamma(1-\phi_j)}.
\end{equation*}
\end{remark}

\subsection{Limiting angular measure for linear combinations of independent stable variables}\label{sec:cooley_res_tweak}

% \LZadd{
Consider
% \singlespacing
\begin{align*}
R_1&=w_{11}S_{1}+\ldots +w_{K1}S_{K},\\
R_2&=w_{12}S_{1}+\ldots +w_{K2}S_{K},\\
&\quad\vdots\\
R_m&=w_{1m}S_{1}+\ldots +w_{Km}S_{K},
\end{align*}
% \doublespacing
with deterministic coefficient vectors $\bw_i=(w_{1i},\ldots, w_{Ki})^{\rm T}$ , $i=1,\ldots, m$. The random vector $\bm S = (S_{1},\ldots,S_{K})^{\rm T}$ is composed of regularly varying variables.
In matrix notation, we write $\bm \Psi = (\bw_1,\ldots, \bw_m)^{\rm T}\in \mathbb{R}^{m\times K}$ and
\[
(R_1, \ldots, R_m)^{\rm T} = \bm \Psi \bm S.
\]
To show Proposition~\ref{prop:R_joint_nonzero}, we first slightly reformulate Corollary 1 of \citet{Cooley2019} where we consider only non-negative weights in $\bm \Psi$.
% }

% \LZadd{
\begin{lemma}[\citeauthor{Cooley2019}, \citeyear{Cooley2019}]\label{lem:cooley-thibaud} \upshape
Assume the independently and identically distributed random variables $\tilde{S}_k$, $k=1,\ldots,K$ are regularly varying at infinity, i.e.,
\begin{equation}\label{eq:rv-marg}
n \Pr(\tilde{S}_k/b_n > x) \rightarrow x^{-\alpha}, \quad n\rightarrow\infty,
\end{equation}
where $b_n>0$ and $\alpha>0$. Also, the normalizing sequence $b_n\rightarrow\infty$, and $\bm\Psi$ is a matrix with $m$ rows, $K$ columns and has only non-negative entries and at least one positive entry. Then the random vector $\bm\Psi \tilde{\bm S}$ is regularly varying at infinity with tail index $\alpha$ and it has angular measure (in dimension $K$ and with respect to any chosen norm $\|\cdot\|$) given by
\begin{equation}\label{eq:angular-measure-general}
H_{\bm\Psi \tilde{\bm S}}(\cdot) = \sum_{k=1}^K \left\|{\bm\omega}_k\right\|^\alpha \times \delta_{{\bm\omega}_k/\|{\bm\omega}_k\|}(\cdot),
\end{equation}
where $\delta$ is the Dirac mass function, and we set the sum term for $k$  to $0$ if $\left\|{\bm\omega}_k\right\|=0$.
\end{lemma}
% }

\vskip 0.2cm
Since $\bm\Psi \tilde{\bm S}$ exhibits regularly variation, the scaling property holds for a given norm $||\cdot||$. Define the unit ball $\mathbb{S}^+_{m-1}=\{\bm x\in \mathbb{R}^m_+:||\bm x||=1\}$. For a set $C(r, B)=\{\bm x\in \mathbb{R}^m_+: ||\bm x||>r, ||\bm x||^{-1}\bm x\in B\}$ with $r>0$ and $B\subset \mathbb{S}^+_{m-1}$ being a Borel set. Then we have
\begin{equation*}
n\pr(b_n^{-1}\bm\Psi \tilde{\bm S}\in C(r, B))\sim r^{-\alpha}H_{\bm\Psi \tilde{\bm S}}(B), \text{ as }n\rightarrow\infty.
\end{equation*}
If we let $B=\mathbb{S}^+_{m-1}$, the event $b_n^{-1}\bm\Psi \tilde{\bm S}\in C(r, B)$ is equivalent to $||\bm\Psi\tilde{\bm S}||/b_n >r$. Also note that the angular measure $H_{\bm\Psi \tilde{\bm S}}$ is usually not a probability measure since its total measure is $m_{\bm\Psi \tilde{\bm S}} =H_{\bm\Psi \tilde{\bm S}}(\mathbb{S}_{m-1}^+) = \sum_{k=1}^K \left\|{\bm\omega}_k\right\|^\alpha$. Therefore,
\begin{equation}\label{eqn:important_cooley_res}
    n \,\Pr\left(\|\bm\Psi \tilde{\bm S}\|/b_n > r\right) \sim m_{\bm\Psi \tilde{\bm S}}r^{-\alpha}, \text{ as }n\rightarrow\infty. 
\end{equation}
%The regular variation condition~\eqref{eq:rv-marg}
 
Using this result, marginal and joint upper tail behaviour of $(R_1, \ldots, R_m)^{\rm T}$ can be derived when the coefficients in $\bm S_t$ are regularly varying. Proposition~\ref{prop:R_joint_nonzero} constitutes a special case, focusing on two linear combinations of stable distributions, which are a specific example of regularly varying random variables.

\begin{proof}[Proof of Proposition~\ref{prop:R_joint_nonzero}  ] \upshape
    \begin{enumerate}[(a)]
        \item The case of $\mathcal{C}_i\cap\mathcal{C}_j=\emptyset$ is clear because independence holds and individually $\pr(R_i>x)\sim 2\bar{\gamma}^\alpha_{i}C_\alpha x^{-\alpha}$ and $\pr(R_j>x)\sim 2\bar{\gamma}^\alpha_{j}C_\alpha x^{-\alpha}$.

        When $\mathcal{C}_i\cap\mathcal{C}_j=\emptyset$, we show Expression~\eqref{eqn:R_joint} by applying Lemma~\ref{lem:cooley-thibaud}. First we notice that the upper tails of the stable distributions with the same concentration $\alpha$ and skewness $\beta=1$ are equivalent up to a positive scaling constant. More specifically, we denote $\tilde{S}_k=S_k/\gamma_k$, $k=1,\ldots, K$, and then all $\tilde{S}_k$'s are iid $\text{Stable}(\alpha,$ $1,$ $1,$ $\delta)$ variables because the $S_k$'s share the location $\delta$ in Equation~\eqref{eqn:scaling_process}. 

        Therefore, Lemma~\ref{lem:cooley-thibaud} is applicable for the pair $(R_i, R_j)^{\rm T}$ under the new weights $(w_{1i}\gamma_1,\ldots, w_{Ki}\gamma_K)^{\rm T}$ and $(w_{1j}\gamma_1,\ldots, w_{Kj}\gamma_K)^{\rm T}$. Since $\pr(\tilde{S}_k>x)\sim 2C_\alpha x^{-\alpha}$ as $x\rightarrow \infty$, we can therefore choose
        \begin{equation*}
            b_n=(2nC_\alpha)^{1/\alpha}
        \end{equation*}
        so $n\pr(\tilde{S}_k/b_n>x)\rightarrow x^{-\alpha}$ as $n\rightarrow\infty$. Choosing the componentwise min-operator as the norm, Expression~\eqref{eqn:important_cooley_res} can be re-written as:
        \begin{equation*}
            n\pr(\min(R_i,R_j)/b_n>r)\sim r^{-\alpha}\sum_{k=1}^K \min(w_{ki}\gamma_k, w_{kj}\gamma_k)^\alpha=r^{-\alpha}\sum_{k=1}^Kw_{k,\wedge}^\alpha\gamma_k^\alpha,
        \end{equation*}
        in which $w_{k,\wedge}=\min(w_{ki},w_{kj})$. The previous display immediately induces Equation~\eqref{eqn:R_joint}.

        \item We only prove the inequality concerning $\min(c_iR_i^{\phi_i},c_jR_j^{\phi_j})$. The proof for $\max(c_iR_i^{\phi_i},c_jR_j^{\phi_j})$ is analogous. 

        If $\mathcal{C}_i=\mathcal{C}_j$, $R_i$ and $R_j$ share the same non-zero indices. Since $1/\phi_i-1/\phi_j>0$, $c_i^{1/\phi_i}w_{ki}/x^{1/\phi_i-1/\phi_j}<c_j^{1/\phi_j}w_{kj}$, $k=1,\ldots, K$ for sufficiently large $x$. Therefore 
        \begin{equation*}
            \begin{split}
               \pr\{\min(c_iR_i^{\phi_i},c_jR_j^{\phi_j})>x\}&=P\left(\frac{c_i^{1/\phi_i}}{x^{1/\phi_i-1/\phi_j}}R_i>x^{1/\phi_j}, c_j^{1/\phi_j}R_j>x^{1/\phi_j}\right)\\
                &=P\left(\frac{c_i^{1/\phi_i}}{x^{1/\phi_i-1/\phi_j}}R_i>x^{1/\phi_j}\right)\sim 2C_\alpha \sum_{k\in \mathcal{C}_i}(w_{ki}\gamma_k)^{\alpha} c_i^{\alpha/\phi_i} x^{-\alpha/\phi_i}.
            \end{split}
        \end{equation*}

        If either $\mathcal{C}_i\setminus\mathcal{C}_j$ or $\mathcal{C}_j\setminus\mathcal{C}_i$ is non-empty, we define
        \begin{equation}\label{eqn:cap_minus}
            R_{i,\setminus}=\sum_{k\in \mathcal{C}_i\setminus \mathcal{C}_j}w_{ki}S_k,\quad R_{i,\cap}=\sum_{k\in \mathcal{C}_i\cap \mathcal{C}_j}w_{ki}S_k,
        \end{equation}
        and $R_i=R_{i,\setminus}+R_{i,\cap}$. Similarly we define $R_{j,\setminus}$ and $R_{j,\cap}$. Then the classic $c_r$-inequality gives us $R_{i,\cap}^{\phi_i}\leq R_i^{\phi_i}\leq 2^{\phi_i}(R_{i,\cap}^{\phi_i}+R_{i,\setminus}^{\phi_i})$. Similarly, it is also true that $R_{j,\cap}^{\phi_j}\leq R_j^{\phi_j}\leq 2^{\phi_j}(R_{j,\cap}^{\phi_j}+R_{j,\setminus}^{\phi_j})$. Therefore,
        \begin{equation}\label{eqn:bounds_min}
            \begin{split}
               \pr\{\min(c_iR_{i,\cap}^{\phi_i},c_j&R_{j,\cap}^{\phi_j})>x\}\leq P\{\min(c_iR_i^{\phi_i},c_jR_j^{\phi_j})>x\}\leq\\
                 &P\{2^{\phi_i}c_i(R_{i,\cap}^{\phi_i}+R_{i,\setminus}^{\phi_i})>x,2^{\phi_j}c_j(R_{j,\cap}^{\phi_j}+R_{j,\setminus}^{\phi_j})>x\}.
            \end{split}
        \end{equation}
        The lower bound in \eqref{eqn:bounds_min} can be approximated using the first case $\mathcal{C}_i=\mathcal{C}_j$. The upper bound in \eqref{eqn:bounds_min} can be further bounded by $\pr(\min(2^{\phi_i}c_iR_{i,\cap}^{\phi_i},2^{\phi_j}c_j R_{j,\cap}^{\phi_j})+\max(2^{\phi_i}c_i R_{i,\setminus}^{\phi_i}, 2^{\phi_j}c_j R_{j,\setminus}^{\phi_j})>x)$, which can be approximated using Lemma \ref{lem:CE_convolution} because
        \begin{equation*}
         \begin{split}
            \pr\{\min(2^{\phi_i}c_iR_{i,\cap}^{\phi_i},2^{\phi_j}c_j R_{j,\cap}^{\phi_j})>x\}&\sim 2C_\alpha\sum_{k\in \mathcal{C}_i\cap\mathcal{C}_j}(w_{ki}\gamma_k)^\alpha c_i^{\alpha/\phi_i}x^{-\alpha/\phi_i},\\
            \pr\{\max(2^{\phi_i}c_i R_{i,\setminus}^{\phi_i}, 2^{\phi_j}c_j R_{j,\setminus}^{\phi_j})>x\}&\sim 2^{\alpha+1}C_\alpha\sum_{k\in\mathcal{C}_i\setminus\mathcal{C}_j}(w_{ki}\gamma_k)^\alpha c_i^{\alpha/\phi_i} x^{-\alpha/\phi_i}.
            \end{split}
        \end{equation*}
        The second inequality holds due to the independence between $R_{i,\setminus}$ and $R_{j,\setminus}$. Combining the approximations of the two bounds yields the stated range of the constant $d_\wedge$.
    \end{enumerate}
\end{proof}

% \begin{lemma}\label{lemma:opposite_tail}
%     When $\mathcal{C}_i\cap \mathcal{C}_j\neq \emptyset$, then
%     \begin{equation*}
%         \pr(\theta_i R_i^{\alpha}>q^{-1}, \theta_j R_j^{\alpha}<q^{-1})\sim o(q)
%     \end{equation*}
%     as $q\rightarrow 0$, in which the positive constants $\theta_i=\frac{1}{2C_\alpha\bar{\gamma}^\alpha_{i}}$ and $\theta_j=\frac{1}{2C_\alpha\bar{\gamma}^\alpha_{j}}$.
% \end{lemma}
% \begin{proof}
%     Intuitively, the components $R_{i,\cap}$, $R_{i,\setminus}$, $R_{j,\cap}$, $R_{j,\setminus}$ are all regularly varying variables if not degenerate. Therefore, when $\mathcal{C}_i\cap \mathcal{C}_j\neq \emptyset$, the probability of $R_i$ being large while $R_j$ being small should be very low.
% \end{proof}

\subsection{Proof of Theorem~\ref{thm:dependence_properties}}
To examine the joint tail of $(R_i, R_j)$ and $(X_i, X_j)$, we begin by recalling the useful results from the literature. The first result is an easy but useful inequality
\begin{equation}\label{eqn:sandwich}
    \begin{split}
        \pr(\min(R_1,R_2)+\min(W_1,W_2)>x)\leq &\pr(R_1+W_1>x,R_2+W_2>x)\leq\\ &\pr(\min(R_1,R_2)+\max(W_1,W_2)>x).
    \end{split}
\end{equation}
The second relates to the tail behaviour on the convolution with a $\text{CE}_\alpha$ distribution.
\begin{lemma}[Theorem 1 of \citet{cline1986convolution} and Lemma 5.1 of \citet{pakes2004convolution}]\label{lem:CE_convolution} \upshape
Let $Y_1\sim F_1$ and $Y_2\sim F_2$ be random variables. If distribution function $F_1\in CE_\alpha$ with $\alpha\geq 0$ and $E(e^{\alpha Y_2})<\infty$ while $\pr(Y_2>x)/\pr(Y_1>x)\rightarrow c\geq 0$ as $x\rightarrow \infty$, then
\begin{equation*}
    \pr(Y_1+Y_2>x)/\pr(Y_1>x)\rightarrow E\left(e^{\alpha Y_2}\right)+cE\left(e^{\alpha Y_1}\right),\quad x\rightarrow\infty.
\end{equation*}
\end{lemma}

\begin{lemma}[Proposition 5 in \citet{engelke2019extremal}]\label{lem:RV_convolution} \upshape
Suppose $\overline{F}_W\in RV_{-\alpha_W}$ with $\alpha_W\geq 0$, and $\overline{F}_R\in RV_{-\alpha_R}$ with $\alpha_R>\alpha_W$. Let $R\sim F_R$ and $(W_1,W_2)\stackrel{d}{=}F_W$ marginally while $R\indep (W_1,W_2)$. Denote $(X_1,X_2)=R(W_1,W_2)$. If the extremal dependence of $(W_1,W_2)$ is summarized by $(\chi_W,\eta_W)$, then $\chi_X=\chi_W$ and 
\begin{equation*}
\eta_X=
\begin{cases}
    \alpha_W/\alpha_R,&\text{ if }\alpha_R<\alpha_W/\eta_W,\\
    \eta_W,&\text{ if }\alpha_R>\alpha_W/\eta_W.
    \end{cases}
\end{equation*}
\end{lemma}

\begin{proof}[Proof of Theorem~\ref{thm:dependence_properties}  ]
\begin{enumerate}[(a)]
    \item Since $\phi_i>\alpha$ and $\phi_j>\alpha$, \eqref{eqn:marginal} ensures
\begin{equation*}
    F_{X_i}^{-1}(1-q)\sim \{(1-\alpha/\phi_i)\theta_i\}^{-\phi_i/\alpha}q^{-\phi_i/\alpha} \text{ with }\theta_i=\frac{1}{2C_\alpha\bar{\gamma}^\alpha_{i}}>0.
\end{equation*}
Similar result holds for $F_{X_j}^{-1}(1-q)$. Therefore
\begin{equation}\label{joint_transform1}
    \begin{split}
        P&(X_i\geq F_{X_i}^{-1}(1-q), X_j\geq F_{X_j}^{-1}(1-q))\\
        &=P\left(R_i^{\phi_i}W_i>\{(1-\alpha/\phi_i)\theta_i\}^{-\phi_i/\alpha}q^{-\phi_i/\alpha}, R_j^{\phi_j}W_j>\{(1-\alpha/\phi_j)\theta_j\}^{-\phi_j/\alpha}q^{-\phi_j/\alpha}\right)\\
        &=P\left(\theta_i R_i^{\alpha}\frac{W_i^{\alpha/\phi_i}}{E(W_i^{\alpha/\phi_i})}>q^{-1}, \theta_j R_j^{\alpha}\frac{W_j^{\alpha/\phi_j}}{E(W_j^{\alpha/\phi_j})}>q^{-1}\right).
    \end{split}
\end{equation}

% \LZadd{
Using Expression~\eqref{eqn:sandwich} and Lemma~\ref{lem:breiman} part~\ref{easy_breiman} again, we deduce that that the right-hand side of \eqref{joint_transform1} is bounded with the range
\begin{footnotesize}
\begin{equation*}
    E\left\{\min\left(\frac{W_i^{\alpha/\phi_i}}{E(W_i^{\alpha/\phi_i})},\frac{W_j^{\alpha/\phi_j}}{E(W_j^{\alpha/\phi_j})}\right)\right\}\bigg[\pr(\min(\theta_i R_i^{\alpha},\theta_j R_j^{\alpha})>q^{-1}), \pr(\max(\theta_i R_i^{\alpha},\theta_j R_j^{\alpha})>q^{-1})\bigg].
\end{equation*}
\end{footnotesize}
% }
Meanwhile, we know from Proposition~\ref{prop:R_joint_nonzero}\ref{lem:R_joint} that 
\begin{equation}\label{eqn:min_approx}
    \begin{split}
        \pr(\min(\theta_i R_i^{\alpha},\theta_j R_j^{\alpha})>q^{-1})=\pr(\theta_i^{1/\alpha} &R_i>q^{-1/\alpha},\theta_j^{1/\alpha} R_j>q^{-1/\alpha})\\
        &= 2C_\alpha C_K(\theta_i^{1/\alpha}\bw_i,\theta_j^{1/\alpha}\bw_j,\boldsymbol{\gamma}) q,
    \end{split}
\end{equation}
in which 
% \LZadd{
\begin{equation*}
    \begin{split}
         C_K(\theta_i^{1/\alpha}\bw_i,\theta_j^{1/\alpha}\bw_j,&\boldsymbol{\gamma})= \sum_{k=1}^K \min\{\theta_i\gamma_k^\alpha w^\alpha_{ki}, \theta_j\gamma_k^\alpha w^\alpha_{kj}\}=\frac{1}{2C_\alpha}\sum_{k=1}^K \min\left\{\frac{\gamma_k^\alpha w^\alpha_{ki}}{\bar{\gamma}^\alpha_{i}}, \frac{\gamma_k^\alpha w^\alpha_{kj}}{\bar{\gamma}^\alpha_{j}}\right\}\\
         &=\frac{1}{2C_\alpha}\sum_{k=1}^K \min\left\{\frac{(w_{ki}\gamma_k)^\alpha}{\sum_{k\in \mathcal{C}_i}(w_{ki}\gamma_k)^\alpha}, \frac{(w_{kj}\gamma_k)^\alpha}{\sum_{k\in \mathcal{C}_j}(w_{kj}\gamma_k)^\alpha}\right\}=\frac{1}{2C_\alpha}\sum_{\mathcal{C}_i\cap\mathcal{C}_j}v_{k,\wedge}.
    \end{split}
\end{equation*}
% }

% \LZadd{
On the other hand,
\begin{equation*}
    \begin{split}
        \pr(\max(\theta_i R_i^{\alpha},\theta_j R_j^{\alpha})>q^{-1})&=1-\pr(\max(\theta_i R_i^{\alpha},\theta_j R_j^{\alpha})\leq q^{-1})\\
        &= \pr(\theta_i R_i^{\alpha}>q^{-1}, \theta_j R_j^{\alpha}<q^{-1})+\pr(\theta_i R_i^{\alpha}<q^{-1}, \theta_j R_j^{\alpha}>q^{-1})+\\
        &\quad \pr(\theta_i R_i^{\alpha}>q^{-1}, \theta_j R_j^{\alpha}>q^{-1}),
    \end{split}
\end{equation*}
By Proposition~\ref{prop:R_joint_nonzero}\ref{lem:R_joint} and Expression~\eqref{eqn:min_approx}, the previous display becomes
\begin{equation*}
    \pr(\max(\theta_i R_i^{\alpha},\theta_j R_j^{\alpha})>q^{-1})=2C_\alpha C_K(\theta_i^{1/\alpha}\bw_i,\theta_j^{1/\alpha}\bw_j,\boldsymbol{\gamma}) q+o(q).
\end{equation*}
Therefore, both $\pr(\min(\theta_i R_i^{\alpha},\theta_j R_j^{\alpha})>q^{-1})$ and $\pr(\max(\theta_i R_i^{\alpha},\theta_j R_j^{\alpha})>q^{-1})$ are dominated by $2C_\alpha C_K(\theta_i^{1/\alpha}\bw_i,\theta_j^{1/\alpha}\bw_j,\boldsymbol{\gamma}) q$. Consequently, we get $\eta_X=1$ and 
$$\chi_{ij}=E\left\{\min\left(\frac{W_i^{\alpha/\phi_i}}{E(W_i^{\alpha/\phi_i})},\frac{W_j^{\alpha/\phi_j}}{E(W_j^{\alpha/\phi_j})}\right)\right\}\sum_{k=1}^K v_{k,\wedge}.$$ 
% }

\item When $0<\phi_i<\phi_j<\alpha$, write $\lambda_i=E(R_i^{\phi_i})$ and $\lambda_j=E(R_j^{\phi_j})$. By \eqref{eqn:marginal},
\begin{equation*}
    F_{X_i}^{-1}(1-q)\sim \lambda_i q^{-1}, \quad F_{X_j}^{-1}(1-q)\sim \lambda_j q^{-1},
\end{equation*}
and
\begin{equation*}
    \begin{split}
        \pr(&X_i\geq F_{X_i}^{-1}(1-q), X_j\geq F_{X_j}^{-1}(1-q))=P\left(\frac{R_i^{\phi_i}}{\lambda_i}W_i>q^{-1}, \frac{R_j^{\phi_j}}{\lambda_j}W_j> q^{-1}\right)%\\        &=P\left(\log (R_i^{\phi_i}/\lambda_i)+\log W_i>\log q^{-1}, \log (R_j^{\phi_j}/\lambda_j)+\log W_j>\log q^{-1}\right)
        .
    \end{split}
\end{equation*}

If $\phi_i/\alpha<\phi_j/\alpha<\eta_W$, we first regard $R^*=\max(R_i^{\phi_i}/\lambda_i,R_j^{\phi_j}/\lambda_j)$ as a radial variable. From Proposition~\ref{prop:R_joint_nonzero}\ref{lem:R_max_min}, we have $\pr(R^*>x)\in RV_{-\alpha/\phi_j}$ and $\alpha_{R^*}=\alpha/\phi_j>\alpha_W=1$. Since $\alpha/\phi_j>1/\eta_W$, we know from Lemma \ref{lem:RV_convolution} that 
\begin{equation*}
    \lim_{q\rightarrow 0}\frac{\log \pr(R^*W_i>q^{-1},R^*W_j>q^{-1})}{\log \pr(R^*W_i>q^{-1})}=\eta_W.
\end{equation*}

Then we regard $R_*=\min(R_i^{\phi_i}/\lambda_i,R_j^{\phi_j}/\lambda_j)$ as a radial variable. Since $\alpha_{R_*}=\alpha/\phi_i>1/\eta_W$, Lemma \ref{lem:RV_convolution} again gives
\begin{equation}\label{eqn:max_joint_eta1}
    \lim_{q\rightarrow 0}\frac{\log \pr(R_*W_i>q^{-1},R_*W_j>q^{-1})}{\log \pr(R_*W_i>q^{-1})}=\eta_W.
\end{equation}

Moreover, $E(R^*)<\infty$ and $E(R_*)<\infty$ due to $\phi_i<\phi_j<\alpha$. We can apply Lemma \ref{lem:breiman} part \ref{easy_breiman} again to show $\log \pr(R^*W_i>q^{-1})\sim \log \pr(R_*W_i>q^{-1})\sim \log \pr(R_i^{\phi_i}W_i>\lambda_i q^{-1})$. By sandwich limit theorem and \eqref{eqn:sandwich},
\begin{equation*}
    \eta_X=\lim_{q\rightarrow 0}\frac{\log \pr(R_i^{\phi_i}W_i>\lambda_i q^{-1},R_i^{\phi_i}W_j>\lambda_i q^{-1})}{\log \pr(R_i^{\phi_i}W_i>\lambda_i q^{-1})}=\eta_W,
\end{equation*}
and $\chi_X=\chi_W=0$.

If $\eta_W<\phi_i/\alpha<\phi_j/\alpha$, we have $\alpha/\phi_j<\alpha/\phi_i<1/\eta_W$. Lemma \ref{lem:RV_convolution} ensures 
\begin{equation}\label{eqn:max_joint_eta}
    \lim_{q\rightarrow 0}\frac{\log \pr(R^*W_i>q^{-1},R^*W_j>q^{-1})}{\log \pr(R^*W_i>q^{-1})}=\alpha_W/\alpha_{R^*}=\phi_j/\alpha,
\end{equation}
and 
\begin{equation*}
    \lim_{q\rightarrow 0}\frac{\log \pr(R_*W_i>q^{-1},R_*W_j>q^{-1})}{\log \pr(R_*W_i>q^{-1})}=\alpha_W/\alpha_{R_*}=\phi_i/\alpha.
\end{equation*}

Therefore, $\eta_X\in [\phi_i/\alpha,\phi_j/\alpha]$ and $\chi_X=\chi_W=0$.

If $\phi_i/\alpha<\eta_W<\phi_j/\alpha$, we have $\alpha/\phi_j<1/\eta_W<\alpha/\phi_i$. Therefore, \eqref{eqn:max_joint_eta} holds for $R^*$ and \eqref{eqn:max_joint_eta1} holds for $R_*$, which proves $\eta_X\in [\eta_W,\phi_j/\alpha]$ and $\chi_X=0$.

\item When $\phi_i<\alpha<\phi_j$, write $\lambda_i=E(R_i^{\phi_i})$ and $\psi_j=(1-\alpha/\phi_j)/\{2C_\alpha\bar{\gamma}^\alpha_{i}\}$. By \eqref{eqn:marginal},
\begin{equation*}
    F_{X_i}^{-1}(1-q)\sim \lambda_i q^{-1}, \quad F_{X_j}^{-1}(1-q)\sim (\psi_j q)^{-\phi_j/\alpha}.
\end{equation*}

First, we assume $\boldsymbol{w}_i=\boldsymbol{w}_j$, i.e.,  $R_i=R_j$. Denote the distribution and density function of $R_i$ and $R_j$ as $F_R$ and $f_R$, and then
\begin{equation*}
    F_R= \text{Stable}\left(\alpha,1,\bar{\gamma},0\right),\quad \bar{\gamma}=\left\{\sum_{k\in \mathcal{C}_i}(w_{ki}\gamma_k)^\alpha\right\}^{\frac{1}{\alpha}}.
\end{equation*}
Thus,
\begin{small}
\begin{equation}\label{eqn:integral_case3}
\begin{split}
    \pr(X_i\geq F_{X_i}^{-1}(1-q), X_j\geq &F_{X_j}^{-1}(1-q))=\pr(R_i^{\phi_i}W_i>\lambda_i q^{-1}, R_j^{\phi_j}W_j>(\psi_j q)^{-\phi_j/\alpha})\\
    &=\int_0^{\infty} P\left(r^{\phi_i}W_i>\lambda_i q^{-1}, r^{\phi_j}W_j>(\psi_j q)^{-\phi_j/\alpha}\right)f_R(r)dr.
    % &=\pr(A_q)+P\left(\left. \frac{R_i^{\phi_i}}{\lambda_i}W_i>q^{-1}\right\rvert B_q\right)\pr(B_q)+P\left(\left.\frac{R_i^{\phi_i}}{\lambda_i}W_i>q^{-1}, \theta_j R_j^{\alpha}\frac{W_j^{\alpha/\phi_j}}{E(W_j^{\alpha/\phi_j})}>q^{-1}\right\rvert C_q\right)\pr(C_q).
\end{split}
\end{equation}
\end{small}

Since $q^{-1/\phi_j}>q^{-1/\alpha}$ for sufficiently small $q$, we can split the limits of the integral in \eqref{eqn:integral_case3} into $(0,(\theta_j^* q)^{-1/\alpha})$, $((\theta_j^* q)^{-1/\alpha},(\lambda_i q^{-1})^{1/\phi_i})$, and $((\lambda_i q^{-1})^{1/\phi_i},\infty)$.
\begin{enumerate}[label=(\roman*)]
    \item When $r\in ((\lambda_i q^{-1})^{1/\phi_i},\infty)$, we have $\lambda_i q^{-1}/r^{\phi_i}<1$ and $(\psi_j q)^{-\phi_j/\alpha}/r^{\phi_j}<1$. Thus,
    \begin{equation}\label{joint_case3_part1}
    \begin{split}
        \int_{(\lambda_i q^{-1})^{1/\phi_i}}^\infty& P\left(r^{\phi_i}W_i>\lambda_i q^{-1}, r^{\phi_j}W_j>(\psi_j q)^{-\phi_j/\alpha}\right)f_R(r)dr\\
     =&\pr(R_i>(\lambda_i q^{-1})^{1/\phi_i})\sim 2C_\alpha\bar{\gamma}^{\alpha}\lambda_i^{-\alpha/\phi_i}q^{\alpha/\phi_i}.
    \end{split}
    \end{equation}

    \item When $r\in ((\theta_j^* q)^{-1/\alpha},(\lambda_i q^{-1})^{1/\phi_i})$, $\lambda_i q^{-1}/r^{\phi_i}>1$ and $(\psi_j q)^{-\phi_j/\alpha}/r^{\phi_j}<1$. Thus,
   \begin{align*}
       &\int_{(\theta_j^* q)^{-1/\alpha}}^{(\lambda_i q^{-1})^{1/\phi_i}} P\left(r^{\phi_i}W_i>\lambda_i q^{-1}, r^{\phi_j}W_j>(\psi_j q)^{-\phi_j/\alpha}\right)f_R(r)dr\\
     &=\int_{(\theta_j^* q)^{-1/\alpha}}^{(\lambda_i q^{-1})^{1/\phi_i}} \pr(r^{\phi_i}W_i>\lambda_i q^{-1})f_R(r)dr=\lambda_i^{-1}q \int_{(\theta_j^* q)^{-1/\alpha}}^{(\lambda_i q^{-1})^{1/\phi_i}} r^{\phi_i}f_R(r)dr\numberthis\label{joint_case3_part2}\\
     &=\lambda_i^{-1}q \int_{(\theta_j^* q)^{-1/\alpha}}^{(\lambda_i q^{-1})^{1/\phi_i}} \frac{1}{\pi\bar{\gamma}}\sum_{m=1}^\infty \frac{\Gamma(m\alpha+1)\sin(m\pi\alpha)(-1)^{m+1}}{m!}\left(\cos\frac{\pi\alpha}{2}\right)^{-m}\frac{r^{\phi_i-m\alpha-1}}{\bar{\gamma}^{-m\alpha-1}}dr\\
    %  &=\frac{\lambda_i^{-1}q}{\pi}\sum_{m=1}^\infty \frac{\Gamma(m\alpha+1)\sin(m\pi\alpha)(-1)^{m+1}}{m!}\left(\cos\frac{\pi\alpha}{2}\right)^{-m}\left\{\frac{(\theta_j^* q)^{m-\phi_i/\alpha}}{m\alpha-\phi_i}-\frac{(\lambda_i^{-1}q)^{m\alpha/\phi_i-1}}{m\alpha-\phi_i}\right\}\\
     &\sim 2\Gamma(\alpha)\sin(\pi\alpha/2)\bar{\gamma}^{\alpha}\frac{\lambda_i^{-1}\psi_j^{1-\phi_i/\alpha}}{\pi(1-\phi_i/\alpha)}q^{2-\phi_i/\alpha}=\frac{\lambda_i^{-1}\psi_j^{-\phi_i/\alpha}(1-\alpha/\phi_j)}{1-\phi_i/\alpha}q^{2-\phi_i/\alpha}.
   \end{align*}
   The penultimate line uses the series expansion for the stable density when $\alpha\neq 1$; see \citet[Chapter 2]{zolotarev1986one}.
    
    \item When $r\in (0,(\theta_j^* q)^{-1/\alpha})$, we have
    \begin{equation}\label{joint_case3_part3}
       \begin{split}
     &\int_0^{(\theta_j^* q)^{-1/\alpha}} P\left(r^{\phi_i}W_i>\lambda_i q^{-1}, r^{\phi_j}W_j>(\psi_j q)^{-\phi_j/\alpha}\right)f_R(r)dr\\
     &\leq \int_0^{(\theta_j^* q)^{-1/\alpha}}\exp\left[-\frac{\log(\lambda_i q^{-1}/r^{\phi_i})+\log\{(\psi_j q)^{-\phi_j/\alpha}/r^{\phi_j}\}}{1+\rho}\right] f_R(r)dr\\
     &=(\psi_j^{-\phi_j/\alpha}\lambda_i)^{-1/(1+\rho)}q^{(1+\phi_j/\alpha)/(1+\rho)}\int_0^{(\theta_j^* q)^{-1/\alpha}}r^{(\phi_i+\phi_j)/(1+\rho)}f_R(r)dr.
     \end{split}
   \end{equation}
    Note that $r^{(\phi_i+\phi_j)/(1+\rho)}f_R(r)\rightarrow 0$ as $r\rightarrow 0$, and $r^{(\phi_i+\phi_j)/(1+\rho)}f_R(r)\sim r^{(\phi_i+\phi_j)/(1+\rho)-\alpha-1}$ as $r\rightarrow \infty$. By Karamata's Theorem \citep[p.17]{resnick2013extreme}, $\int_0^x r^{(\phi_i+\phi_j)/(1+\rho)}f_R(r)dr\in RV_{(\phi_i+\phi_j)/(1+\rho)-\alpha}$ when $(\phi_i+\phi_j)/(1+\rho)\geq \alpha$. When $(\phi_i+\phi_j)/(1+\rho)<\alpha$, $\int_0^x r^{(\phi_i+\phi_j)/(1+\rho)}f_R(r)dr< E\{R_i^{(\phi_i+\phi_j)/(1+\rho)}\}<\infty$. Apply this result to the right-hand side of \eqref{joint_case3_part3} to get
    \begin{equation*}
       \begin{split}
     &\int_0^{(\theta_j^* q)^{-1/\alpha}} P\left(r^{\phi_i}W_i>\lambda_i q^{-1}, r^{\phi_j}W_j>(\psi_j q)^{-\phi_j/\alpha}\right)f_R(r)dr\\
     &\leq \begin{cases}
     L_R(q^{-1})q^{(1-\phi_i/\alpha)/(1+\rho)+1},&\text{if }(\phi_i+\phi_j)/(2\eta_W)\geq \alpha,\\
     C_R q^{(1+\phi_j/\alpha)/(1+\rho)},&\text{if }(\phi_i+\phi_j)/(2\eta_W)< \alpha,
     \end{cases}
     \end{split}
   \end{equation*}
   where $L_R\in RV_0$ and $C_R=(\psi_j^{-\phi_j/\alpha}\lambda_i)^{-1/(1+\rho)}E\{R_i^{(\phi_i+\phi_j)/(1+\rho)}\}$.
    
    On the other hand, 
 \begin{equation}\label{joint_case3_part4}
       \begin{split}
     \int_0^{(\theta_j^* q)^{-1/\alpha}} &P\left(r^{\phi_i}W_i>\lambda_i q^{-1}, r^{\phi_j}W_j>(\psi_j q)^{-\phi_j/\alpha}\right)f_R(r)dr\\
     \geq &\int_0^{(\theta_j^* q)^{-1/\alpha}} P\left(r^{\phi_i}W_i>\lambda_i q^{-1}\right)P\left(r^{\phi_j}W_j>(\psi_j q)^{-\phi_j/\alpha}\right)f_R(r)dr\\
     =&\psi_j^{\phi_j/\alpha}\lambda_i^{-1}q^{1+\phi_j/\alpha}\int_0^{(\theta_j^* q)^{-1/\alpha}}r^{\phi_i+\phi_j}f_R(r)dr=\Tilde{L}_R(q^{-1})q^{2-\phi_i/\alpha},
     \end{split}
   \end{equation}
    in which $\Tilde{L}_R\in RV_0$.
\end{enumerate}
Combine the results from \eqref{joint_case3_part1} - \eqref{joint_case3_part4} under the assumption that $\boldsymbol{w}_i=\boldsymbol{w}_j$, and we can obtain lower and upper bounds of the regularly varying index for $\pr(X_i\geq F_{X_i}^{-1}(1-q), X_j\geq F_{X_j}^{-1}(1-q))$, which is also known as $\eta_X$.

When $\boldsymbol{w}_i\neq \boldsymbol{w}_j$ and $\mathcal{C}_i\cap \mathcal{C}_j\neq \emptyset$,
\begin{footnotesize}
\begin{equation*}
    \begin{split}
        P\left[\left(\sum_k w_{k,\wedge}S_k\right)^{\phi_i}W_i>\right.&\left.\lambda_i q^{-1},\left(\sum_k w_{k,\wedge}S_k\right)^{\phi_j}W_j>(\psi_j q)^{-\phi_j/\alpha}\right]\\
        &\leq \pr(R_i^{\phi_i}W_i>\lambda_i q^{-1}, R_j^{\phi_j}W_j>(\psi_j q)^{-\phi_j/\alpha})\\
        &\leq P\left[\left(\sum_k w_{k,\vee}S_k\right)^{\phi_i}W_i>\lambda_i q^{-1},\left(\sum_k w_{k,\vee}S_k\right)^{\phi_j}W_j>(\psi_j q)^{-\phi_j/\alpha}\right],
    \end{split}
\end{equation*}
\end{footnotesize}
whose bounds can be dealt with the results we just obtained while assuming $\boldsymbol{w}_i=\boldsymbol{w}_j$. Since the bounds for $\eta_X$ did not depend on the weights, they stay the same for when $\boldsymbol{w}_i\neq \boldsymbol{w}_j$.
\end{enumerate}
\end{proof}

\section{MCMC details}\label{sec:Appendix MCMC}

Conditioning on the scaling variable $\boldsymbol{S}_t$ at the knots, we define the hierarchical model as
\begin{align*}\label{eq:likelihood}
    \mathcal{L}(\boldsymbol{Y}_t|\boldsymbol{\theta}_{\boldsymbol{Y}_t}, \boldsymbol{S}_t, \boldsymbol{\gamma}, \boldsymbol{\phi}, \boldsymbol{\rho}) &= \varphi_D(\boldsymbol{Z}_t)\left\lvert\frac{\partial\boldsymbol{Z}_t}{\partial \boldsymbol{Y}_t}\right\rvert \\
    \boldsymbol{S}_{t} \given \boldsymbol{\gamma} &\sim \text{Stable}(\alpha = 0.5, 1, \boldsymbol{\gamma}, \delta = 0) \\
    \phi_k & \iid \text{Beta}(5,5), \\
    \rho_k & \iid \text{halfNormal}(0,2), \\
    k &= 1, ..., K
\end{align*}
in which $\varphi_D$ is the $D$-variate Gaussian density function with covariance matrix $\boldsymbol{\Sigma_{\rho}}$ and $\partial\boldsymbol{Z}_t/\partial \boldsymbol{Y}_t$ is the Jacobian. 
Using the inverse function theorem that the derivative of $F^{-1}(t)$ is equal to $1/F'(F^{-1}(t))$ as long as $F'(F^{-1}(t))\neq 0$, the Jacobian is a diagonal matrix with elements
\begin{align*}
    \frac{\partial Z_{tj}}{\partial Y_{tj}} 
    &= \frac{1}{\varphi\left\{\Phi^{-1}\left(1-\frac{1}{X_{tj}/R_{tj}^{\phi_j}+1}\right)\right\}}\cdot \frac{1}{(X_{tj}/R_{tj}^{\phi_j}+1)^2R_{tj}^{\phi_j}}\cdot \frac{\partial X_{tj}}{\partial Y_{tj}}\\
    &=\frac{1}{\varphi(Z_{tj})}\cdot \frac{1}{(X_{tj}/R_{tj}^{\phi_j}+1)^2R_{tj}^{\phi_j}}\cdot \frac{f_Y(Y_{tj})}{f_X(X_{tj})},
\end{align*}

\noindent in which $\varphi$ and $\Phi$ are, respectively, the density and the distribution function of a univariate standard Gaussian, $X_{tj} = F_X^{-1}\circ F_Y(Y_{tj})$, $Z_{tj} = g^{-1}(X_{tj}/R_{tj}^{\phi_j})=\Phi^{-1}(1- 1/(X_{tj}/R_{tj}^{\phi_j} + 1))$, $f_Y$ is the marginal density of the observed data distribution, and $f_X$ is the univariate density of the dependence model derived in \eqref{eqn:pdf_X(s)}.
Then, as we have assumed temporal independence by introducing marginal temporal parameter, likelihoods across the independent time replicates are multiplied together for the joint likelihood.

% \BAS{Let's try to move this back to the main text if we can.}

\section{Additional Results on Simulation Study}\label{sec:Appendix Simulation}

\begin{figure}[H]
    \centering
    \begin{minipage}[b]{0.197\textwidth}
        \centering
        \includegraphics[width=\linewidth]{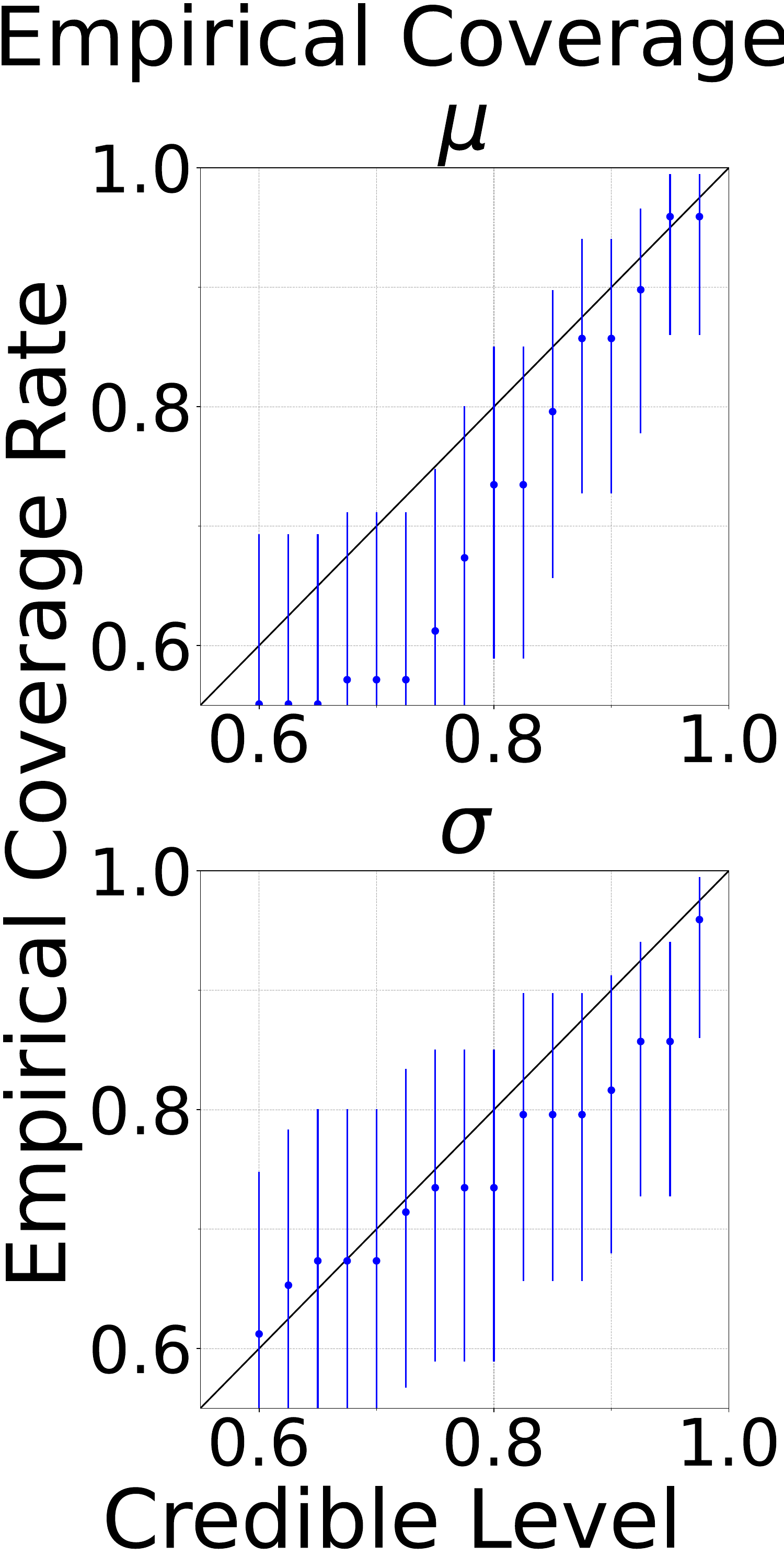}
    \end{minipage}
    % \hfill
    \begin{minipage}[b]{0.393\textwidth}
        \centering
        \includegraphics[width=\linewidth, clip=true, trim=0 0 8 0]{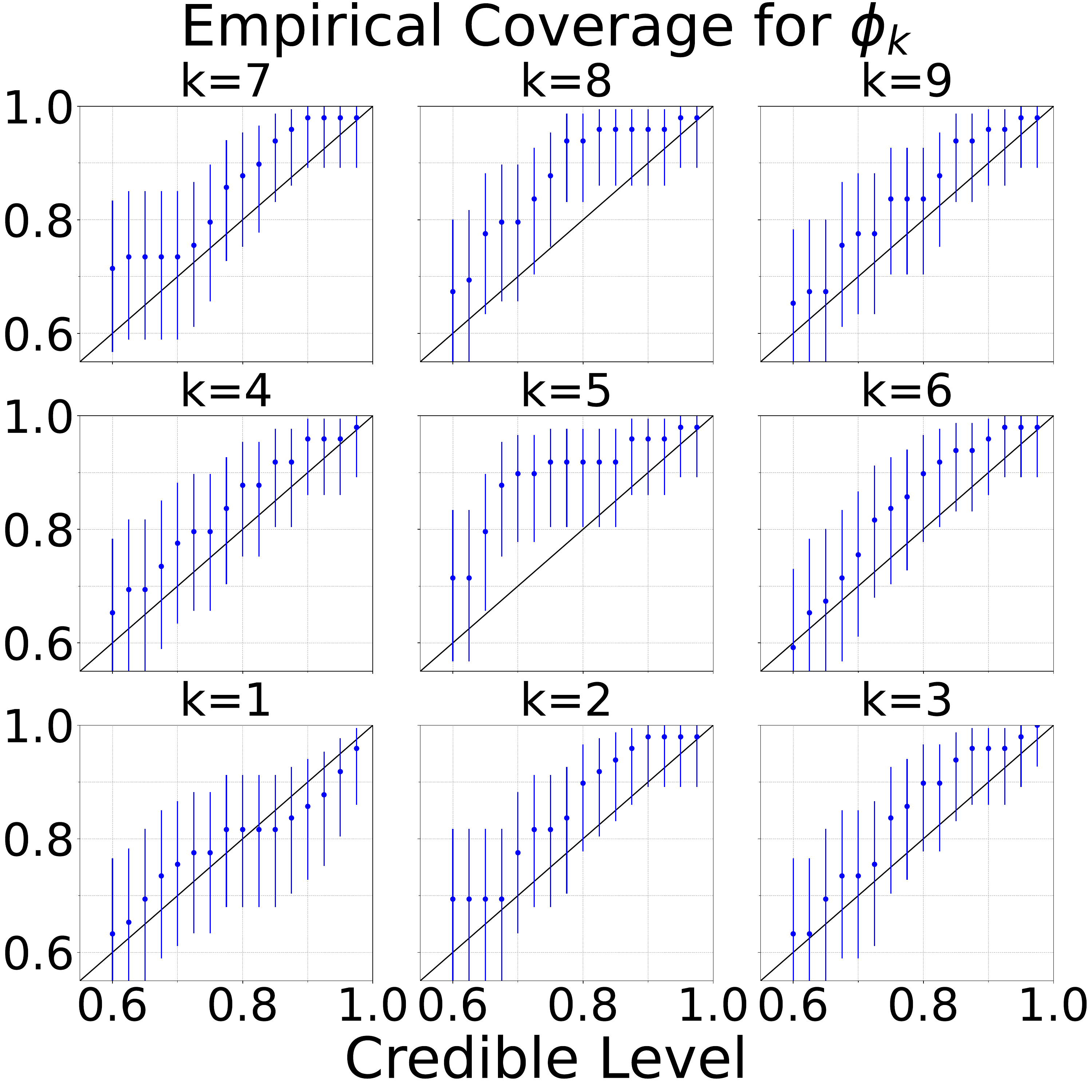}
    \end{minipage}
    % \hfill
    \begin{minipage}[b]{0.393\textwidth}
        \centering
        \includegraphics[width=\linewidth, clip=true, trim=8 0 0 0]{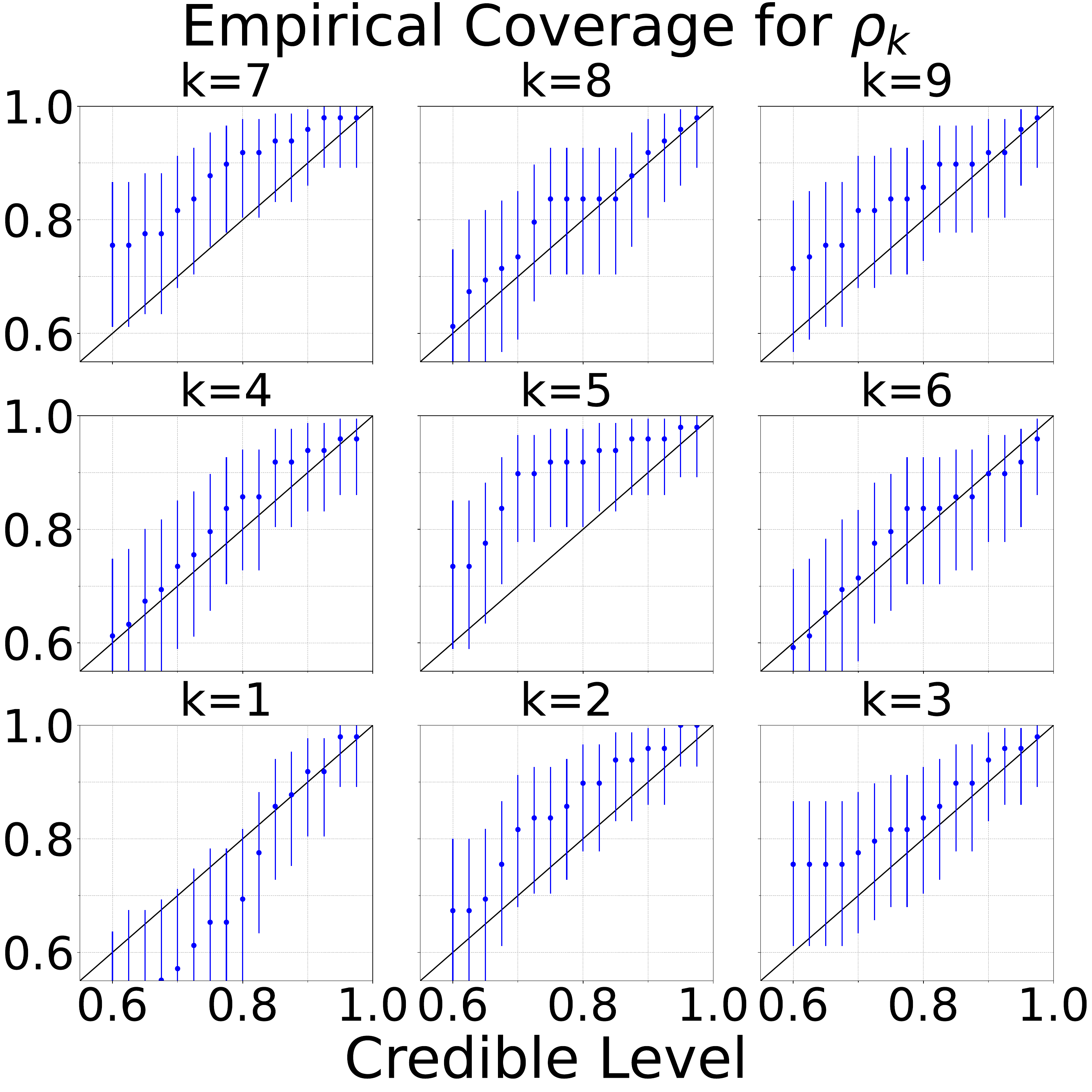}
    \end{minipage}
    \caption{Empirical coverage rates of credible intervals of the marginal parameters $\mu$ and $\sigma$ (left), the dependence parameters $\phi_k$, $k=1,\ldots, 9$ (middle), and $\rho_k$, $k = 1, \dots, 9$ (right), in simulation scenario 2.}
    \label{fig:scenario2_coverage}
\end{figure}

\begin{figure}[H]
    \centering
    \begin{minipage}[b]{0.197\textwidth}
        \centering
        \includegraphics[width=\linewidth]{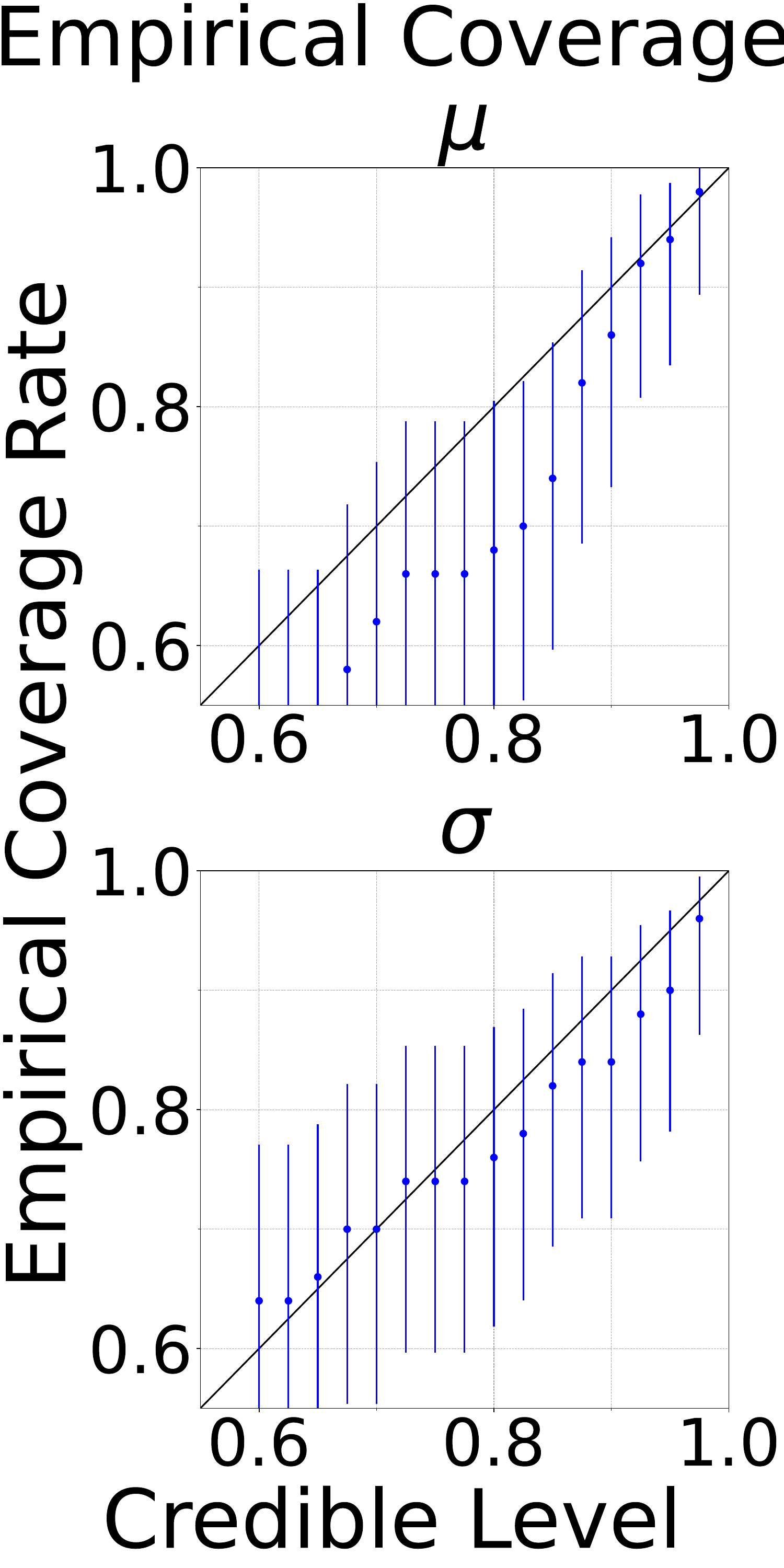}
    \end{minipage}
    % \hfill
    \begin{minipage}[b]{0.393\textwidth}
        \centering
        \includegraphics[width=\linewidth, clip=true, trim=0 0 8 0]{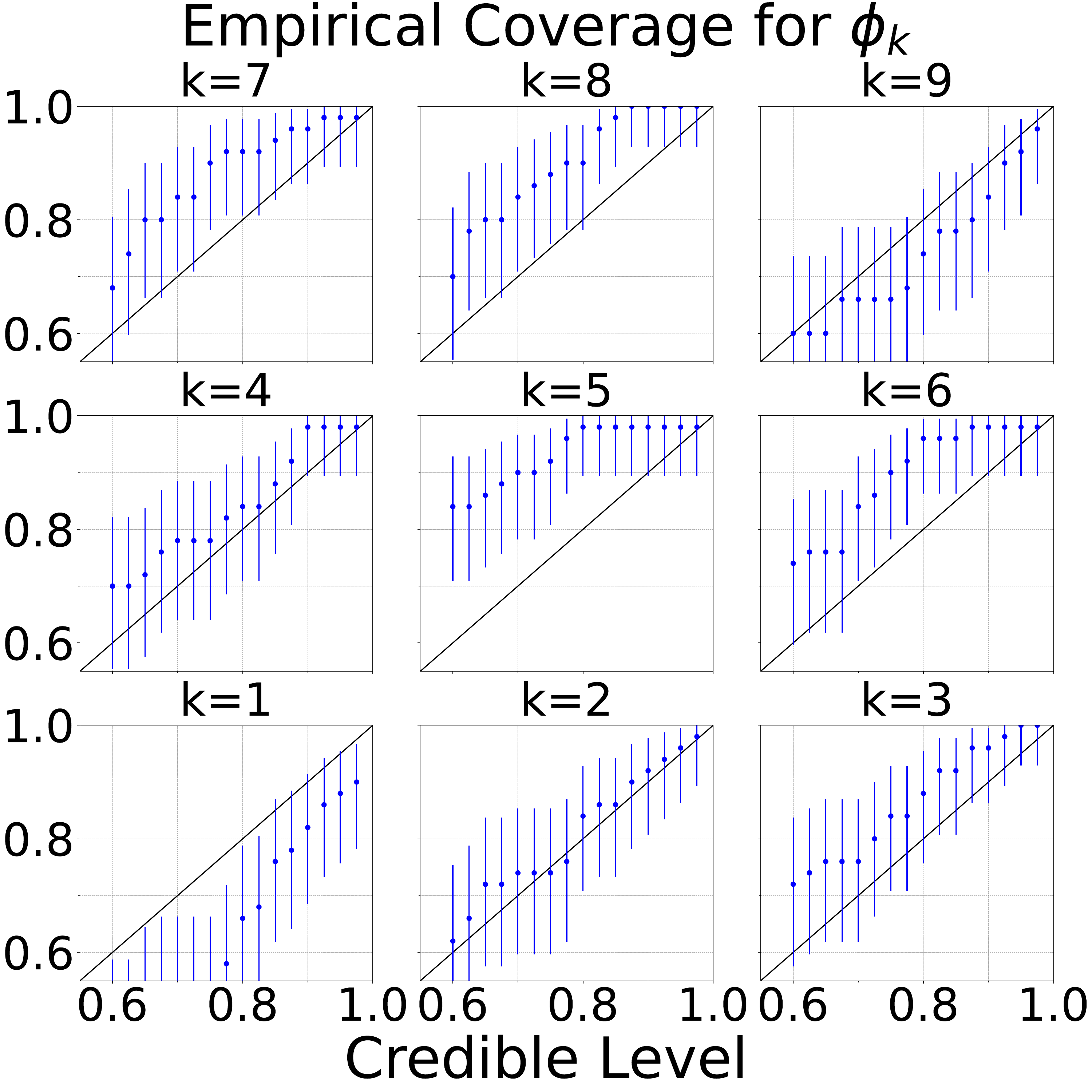}
    \end{minipage}
    % \hfill
    \begin{minipage}[b]{0.393\textwidth}
        \centering
        \includegraphics[width=\linewidth, clip=true, trim=8 0 0 0]{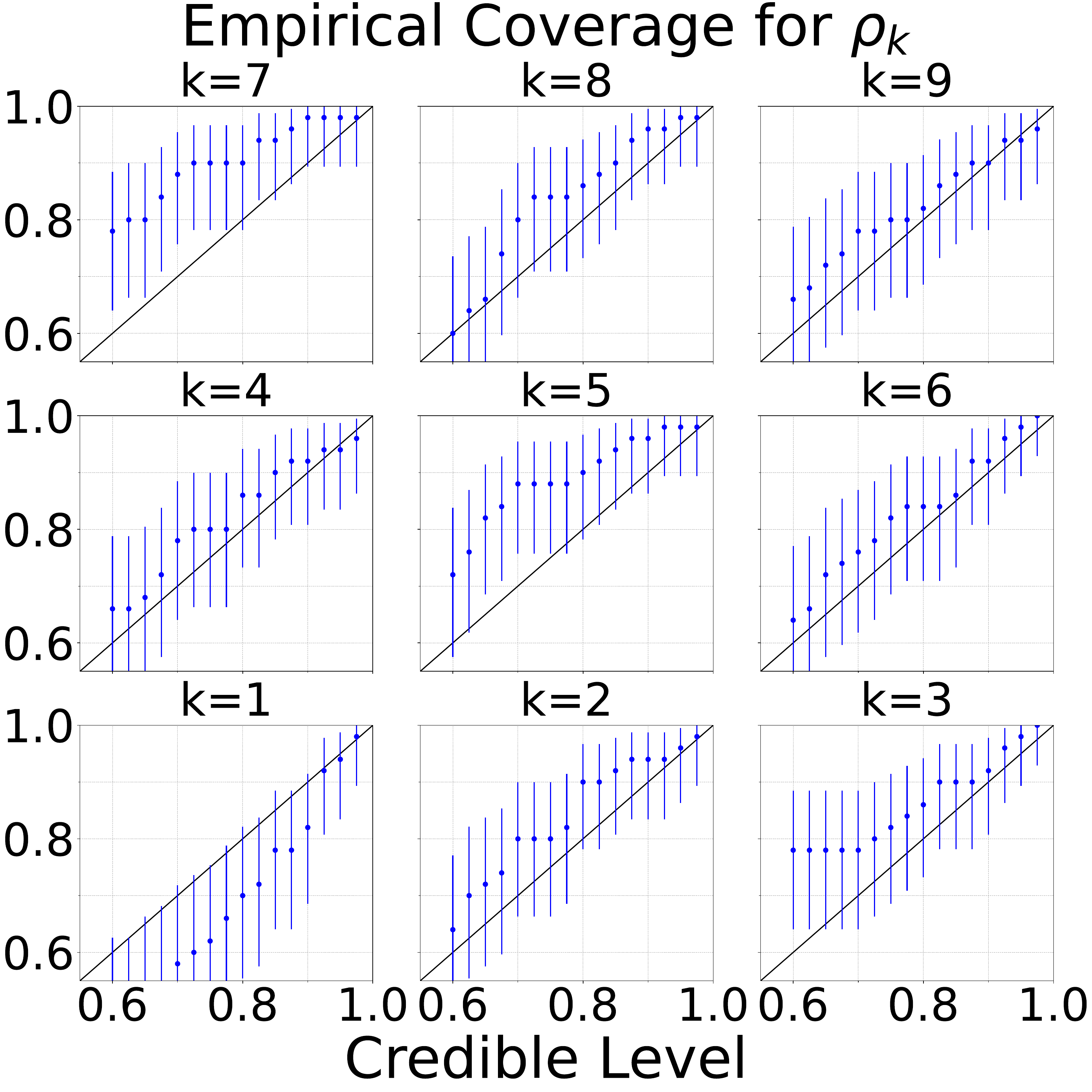}
    \end{minipage}
    \caption{Empirical coverage rates of credible intervals of the marginal parameters $\mu$ and $\sigma$ (left), the dependence parameters $\phi_k$, $k=1,\ldots, 9$ (middle), and $\rho_k$, $k = 1, \dots, 9$ (right), in simulation scenario 3.}
    \label{fig:scenario3_coverage}
\end{figure}

\section{Justification for Marginal Model Choices}\label{sec:Appendix_GEV}

\revise{To build our model for the marginal GEV parameters $\mu$, $\sigma$, and $\xi$, we perform a suite of exploratory analyses.  First, we fit GEV distributions by maximum likelihood to the annual maxima at each station, independently.  Because the resulting fits are very noisy, to reveal any underlying trends we next smooth the estimated parameters using a bivariate cubic spline as implemented in \texttt{Scipy}'s \texttt{SmoothBivariateSpline} \citep{Scipy}.}

\revise{In Figure  \ref{fig:GEV-elevation}, we compare the smoothed surfaces of the MLEs of $\log(\sigma)$ and $\xi$ with the negative of elevation at each observation location.  We observe similar spatial structure between smoothed estimates of $\log(\sigma)$ and elevation (there is not much spatial structure in the smoothed estimates of $\xi$ at all), so we conclude that elevation is sufficient to use as the sole covariate for $\log(\sigma)$ and $\xi$. Furthermore, consistent with previous work on precipitation in this region showing no gain from allowing $\sigma$ and $\xi$ to vary temporally \citep[e.g.][]{Westra2013Global,risser2019probabilistic}, we specify a time trend in only the location parameter $\mu$. }

\revise{To model the temporal trend in $\mu$, previous studies have used either the natural logarithm of CO$_2$ concentrations \citep{risser2017attributable}, the radiative forcing associated with CO$_2$ concentrations \citep[which is logarithmic in concentrations; see][]{etminan2016radiative}, or the sum-total radiative forcing from the five primary greenhouse gases \citep{risser2024anthropogenic}, making any one of these a natural choice. However, each of these variables is nearly linear over our study period of 1949--2023 (correlations in excess of 0.987; see the figure below). Therefore, our results would be the same regardless of the specific covariate used. }

\revise{Additionally, we could have used other covariates in the time trend model that describe the effect of large-scale climate drivers on extreme precipitation in the central US, e.g., the El Nino/Southern Oscillation. However, previous work has shown that these drivers have a minimal effect on seasonal extreme precipitation in the central US \citep[see Figure 1 of][]{risser2021quantifying}}

\begin{figure}[h]
\centering
\includegraphics[width=0.9\textwidth]{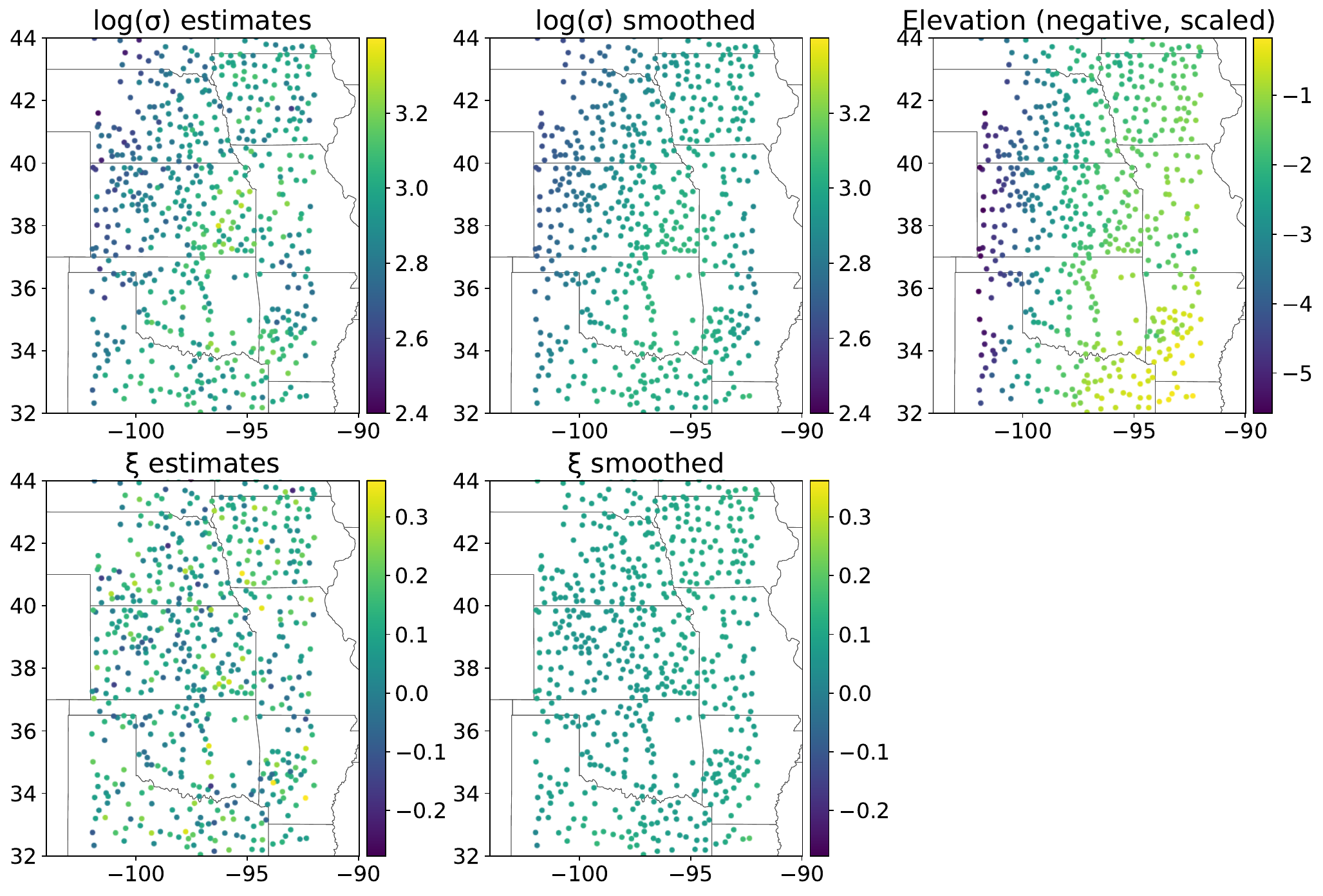}
\caption{\revise{Estimates and smoothed surfaces of maximum likelihood estimates of $\log(\sigma)$ and $\xi$, fit to annual maxima at each station independently.  The smoothed estimates of $\log(\sigma)$ show similar spatial patterns as the negative of scaled elevation at each station (rightmost panel).}}
\label{fig:GEV-elevation}
\end{figure}

\begin{figure}[H]
\centering
\includegraphics[width=0.9\textwidth]{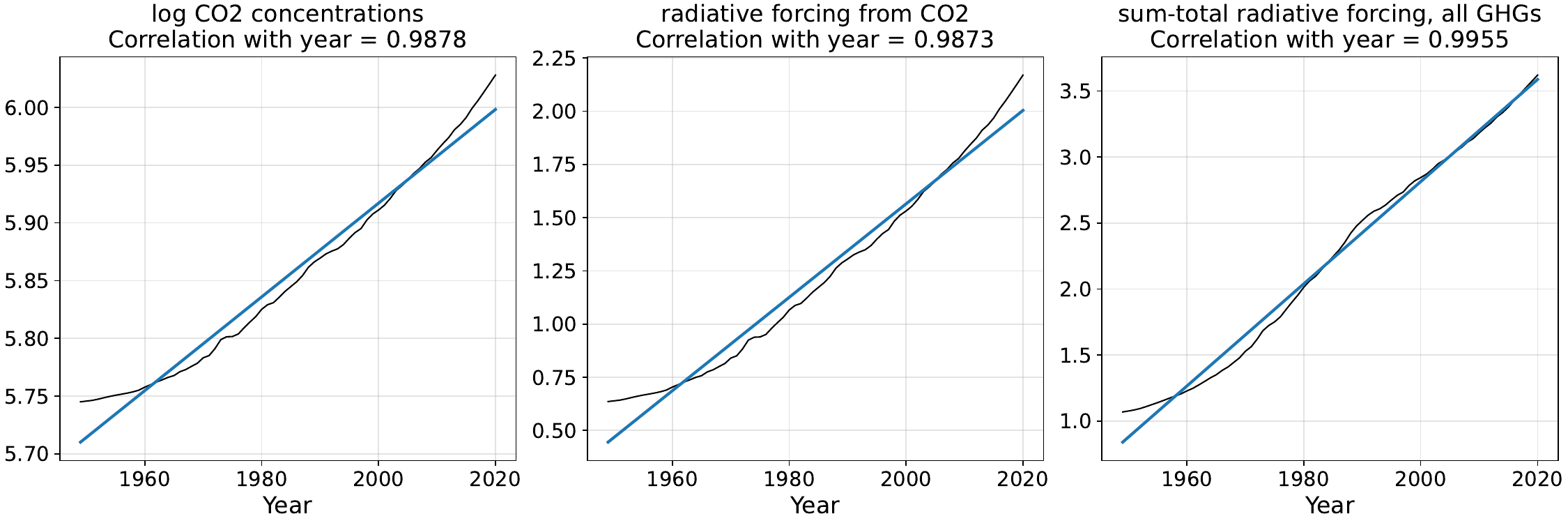}
\caption{\revise{CO$_2$ concentrations (and its functions) against time in black curve. Least square fit in Blue.}}
\label{fig:co2}
\end{figure}

\section{\revise{Confidence Intervals for the Empirical and Model-based  $\hat{\chi}$}}\label{sec:Appendix_chi}

\revise{The figures below show approximate confidence intervals for the empirical and model realized $\chi$ surfaces. Across spatial lags, the model’s $\chi$ surfaces generally fall within the empirical confidence bands at higher quantiles (e.g., $u$=0.99), suggesting good agreement in the upper tail. At lower quantiles (e.g., $u$=0.9), however, the model does seem to estimate slightly higher dependence. The model's 95\% confidence interval for $\chi$ does seems to overlap with the intervals from the empirical estimates.  Empirical $\hat{\chi}$ plots are in the \textit{left} column, and those from model realizations are in the \textit{right} column.}
\begin{adjustwidth}{-1.5cm}{-1.5cm}
\begin{center}
    \includegraphics[width=0.56\textwidth]{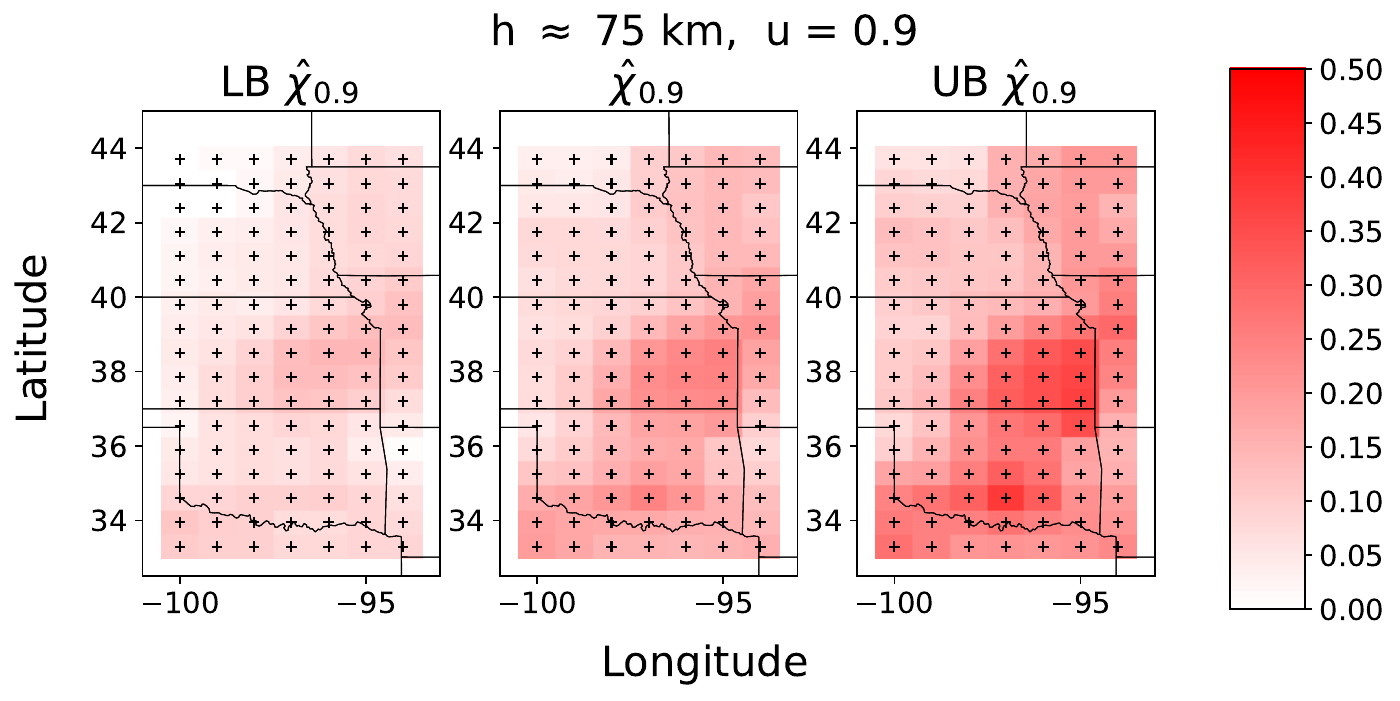}
    \includegraphics[width=0.56\textwidth]{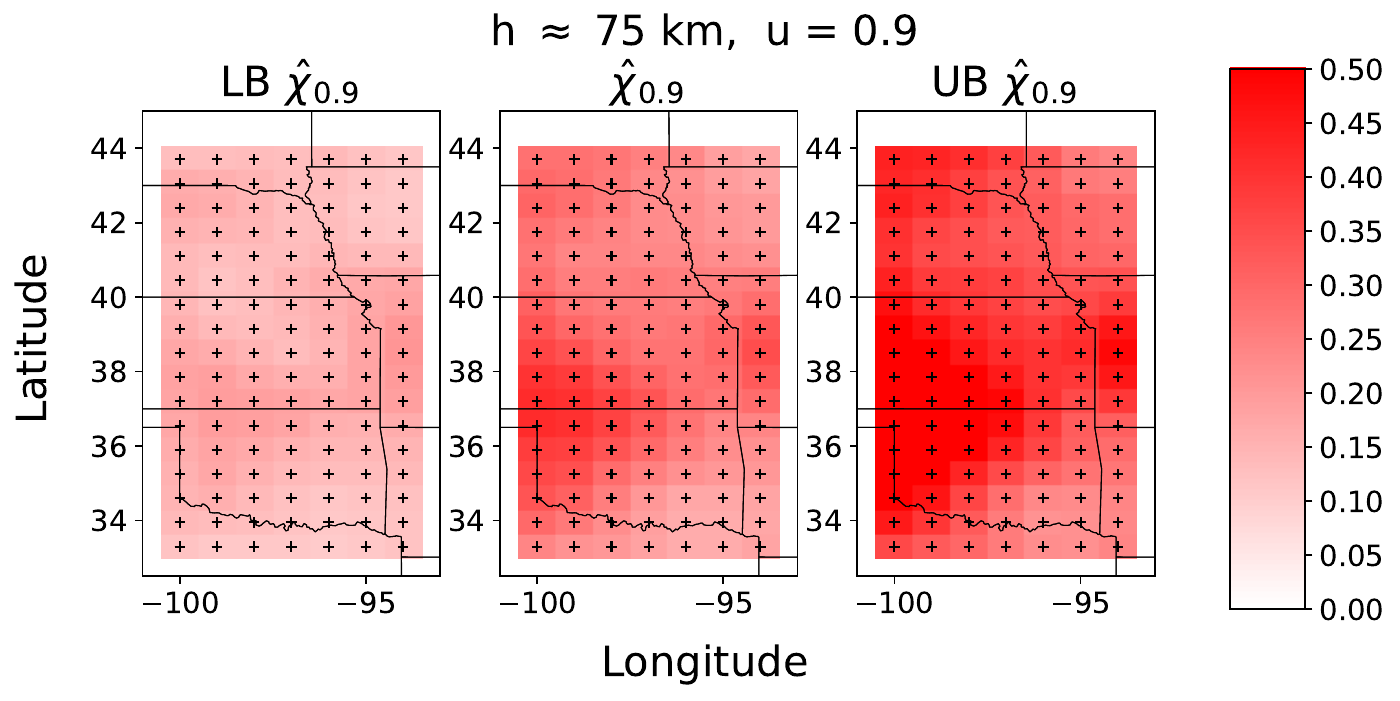}
\end{center}
\begin{center}
    \includegraphics[width=0.56\textwidth]{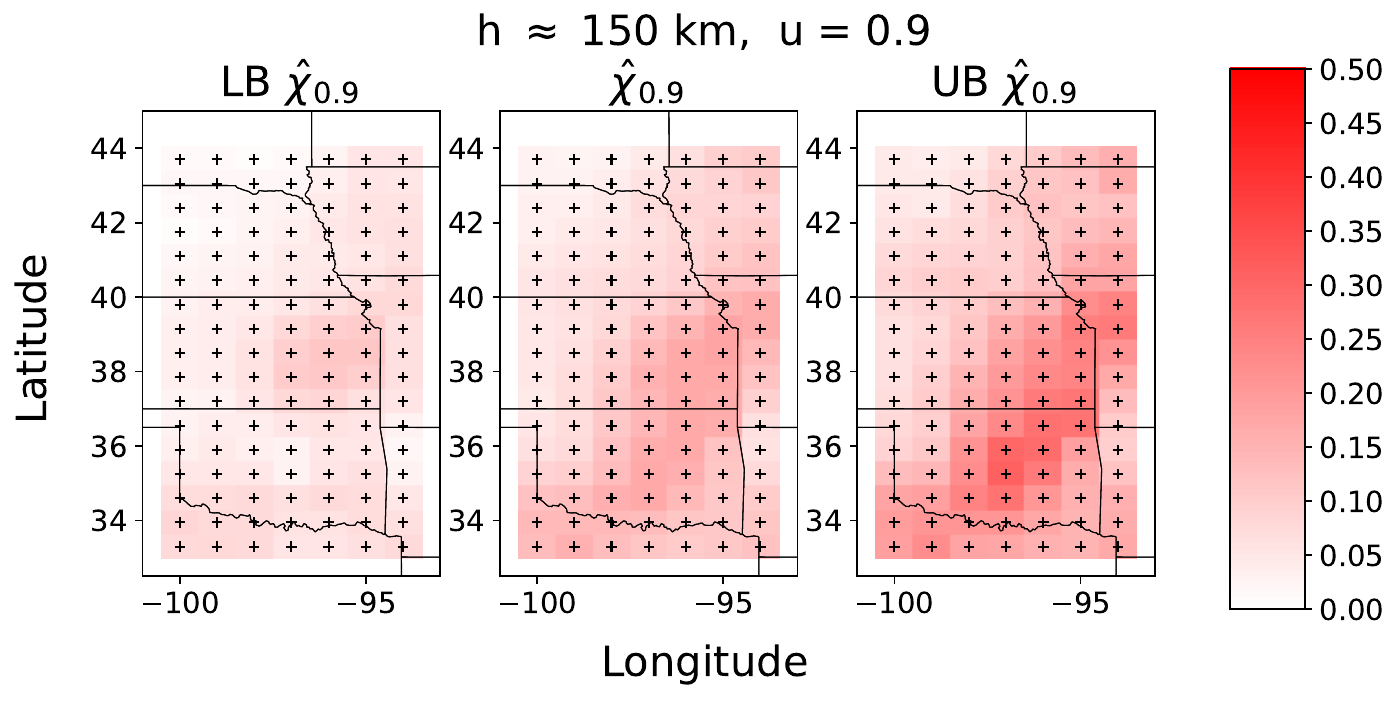}
    \includegraphics[width=0.56\textwidth]{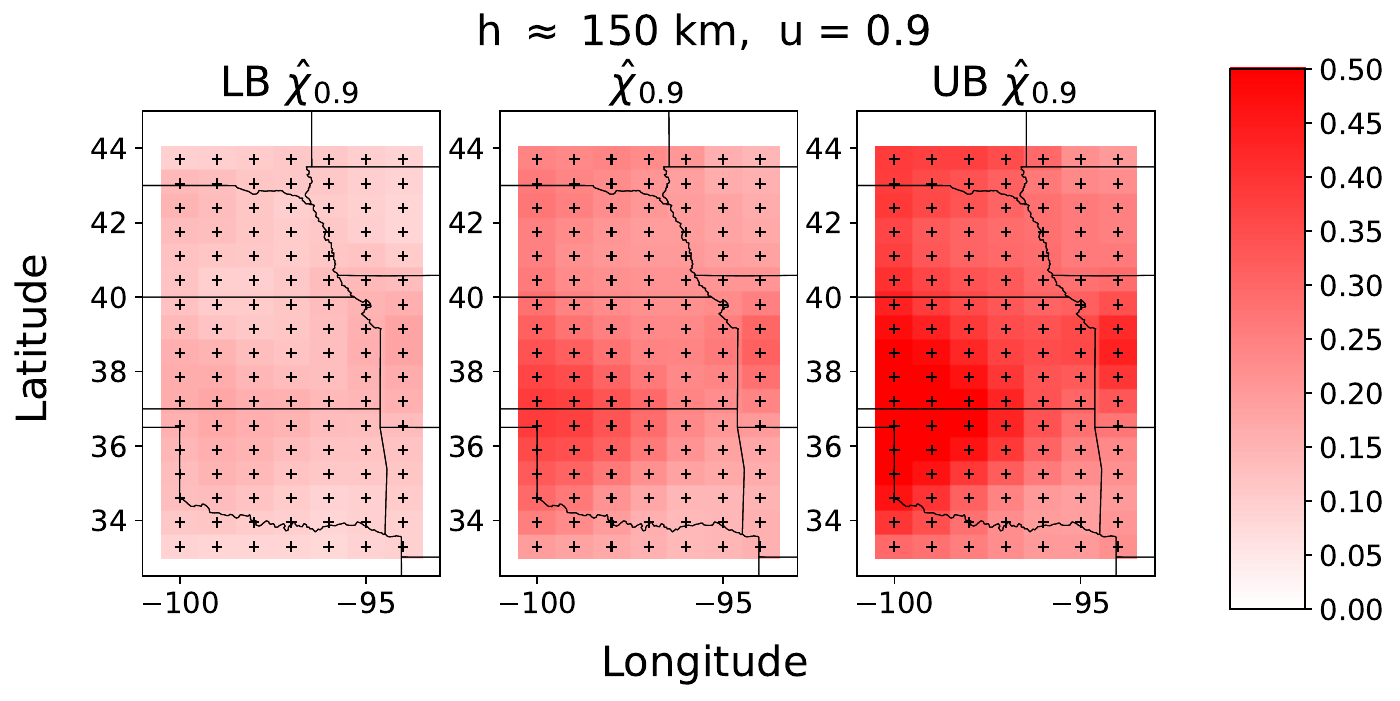}
\end{center}
\begin{center}
    \includegraphics[width=0.56\textwidth]{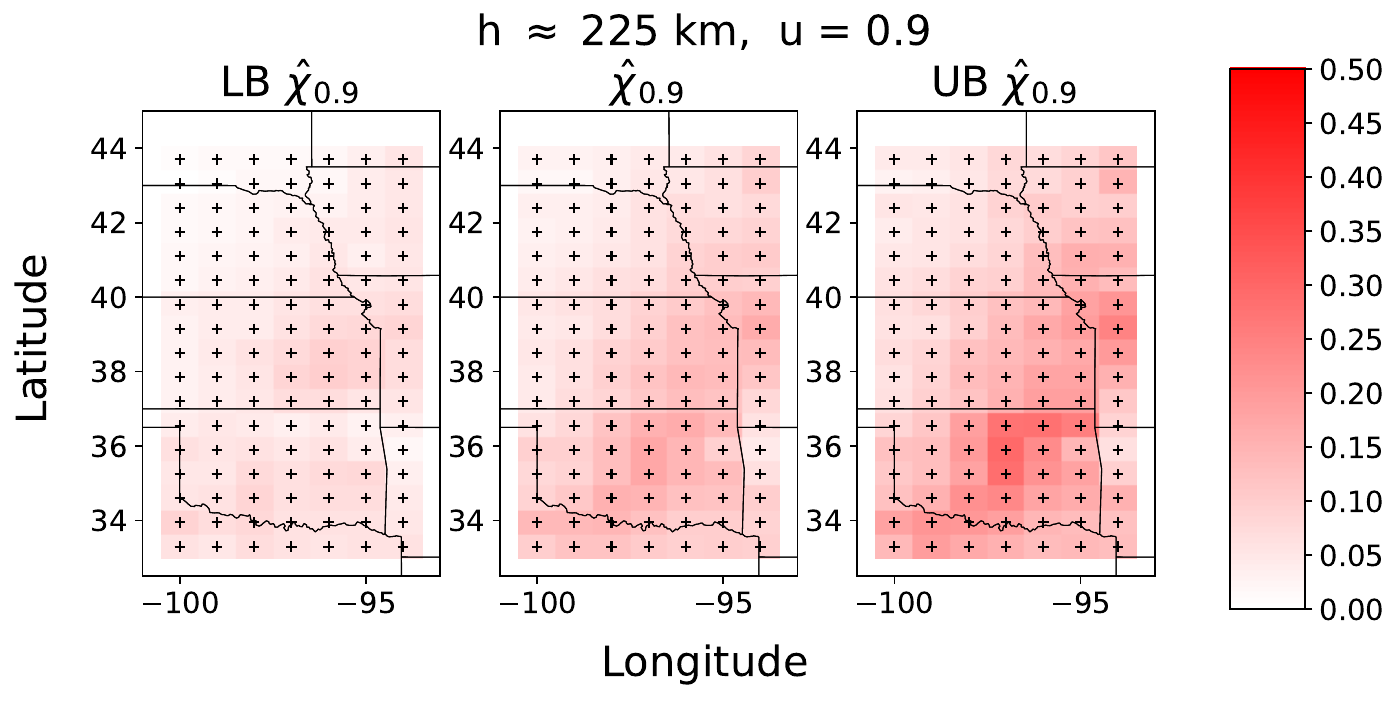}
    \includegraphics[width=0.56\textwidth]{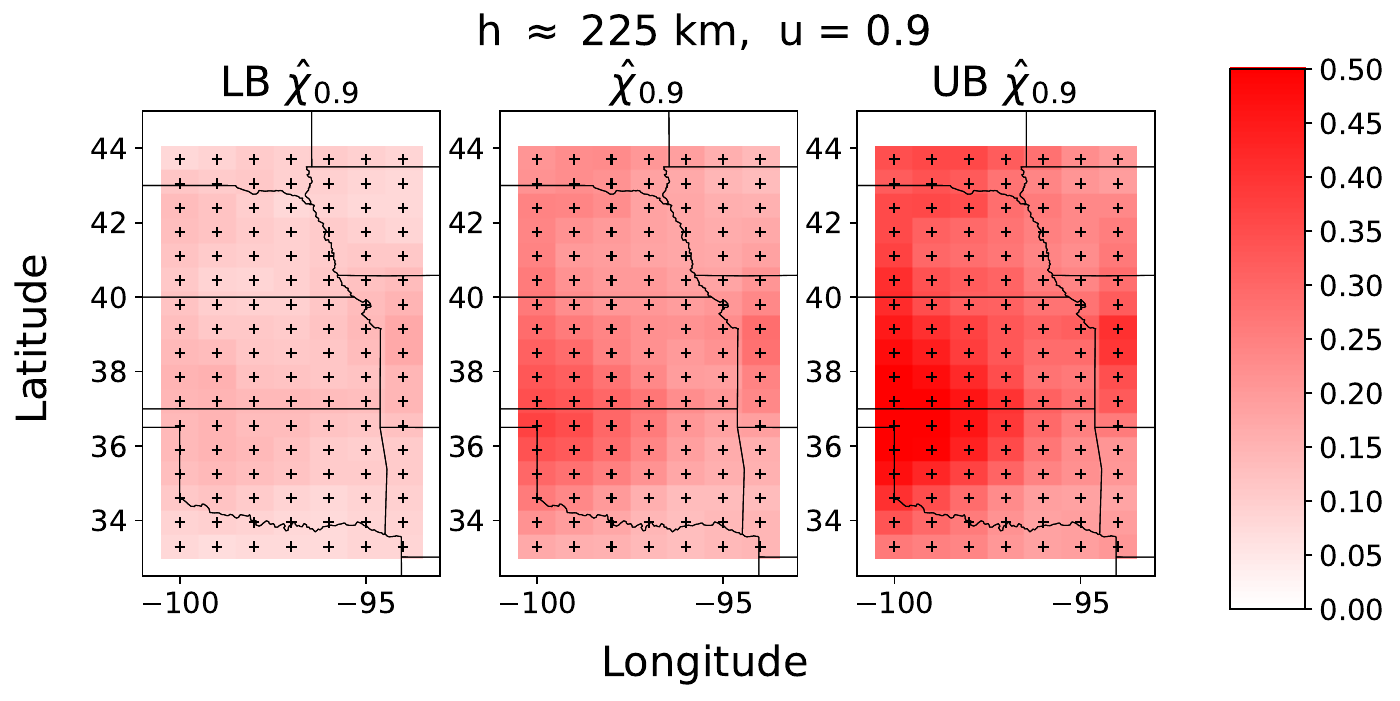}
\end{center}
\begin{center}
    \includegraphics[width=0.56\textwidth]{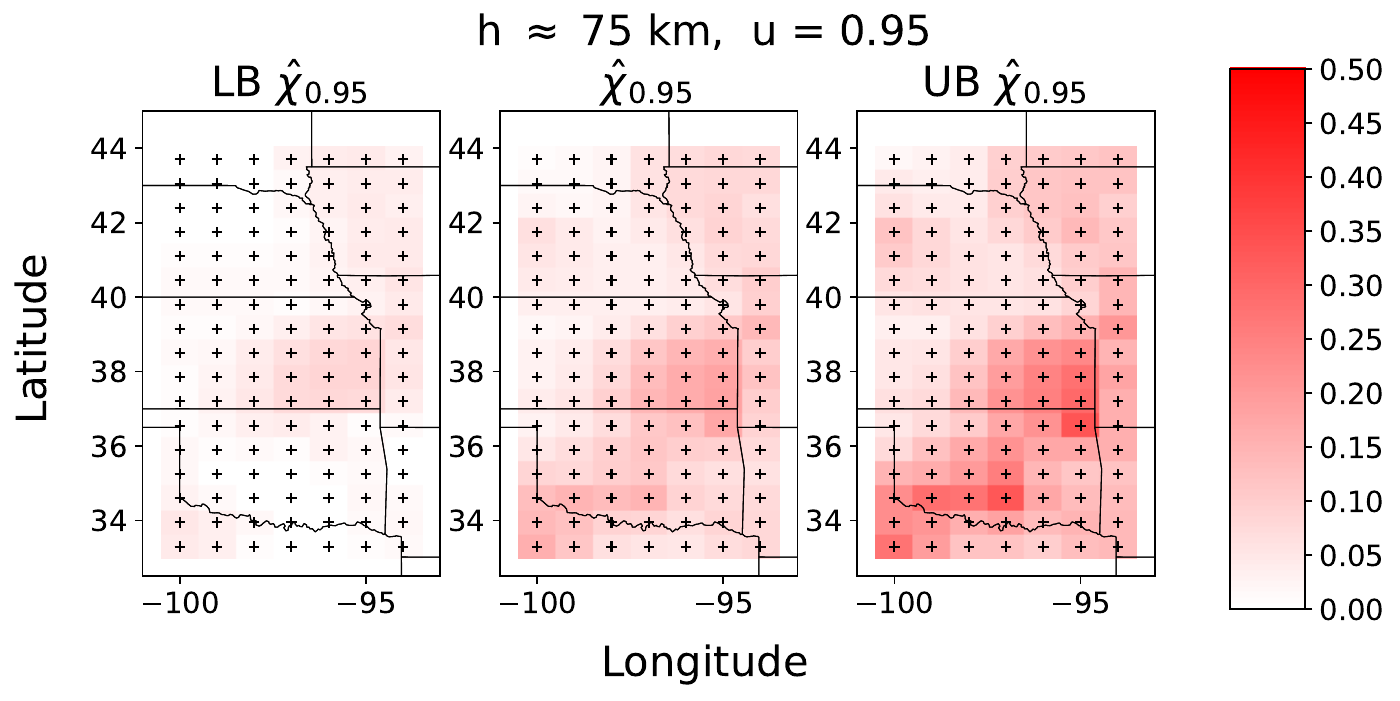}
    \includegraphics[width=0.56\textwidth]{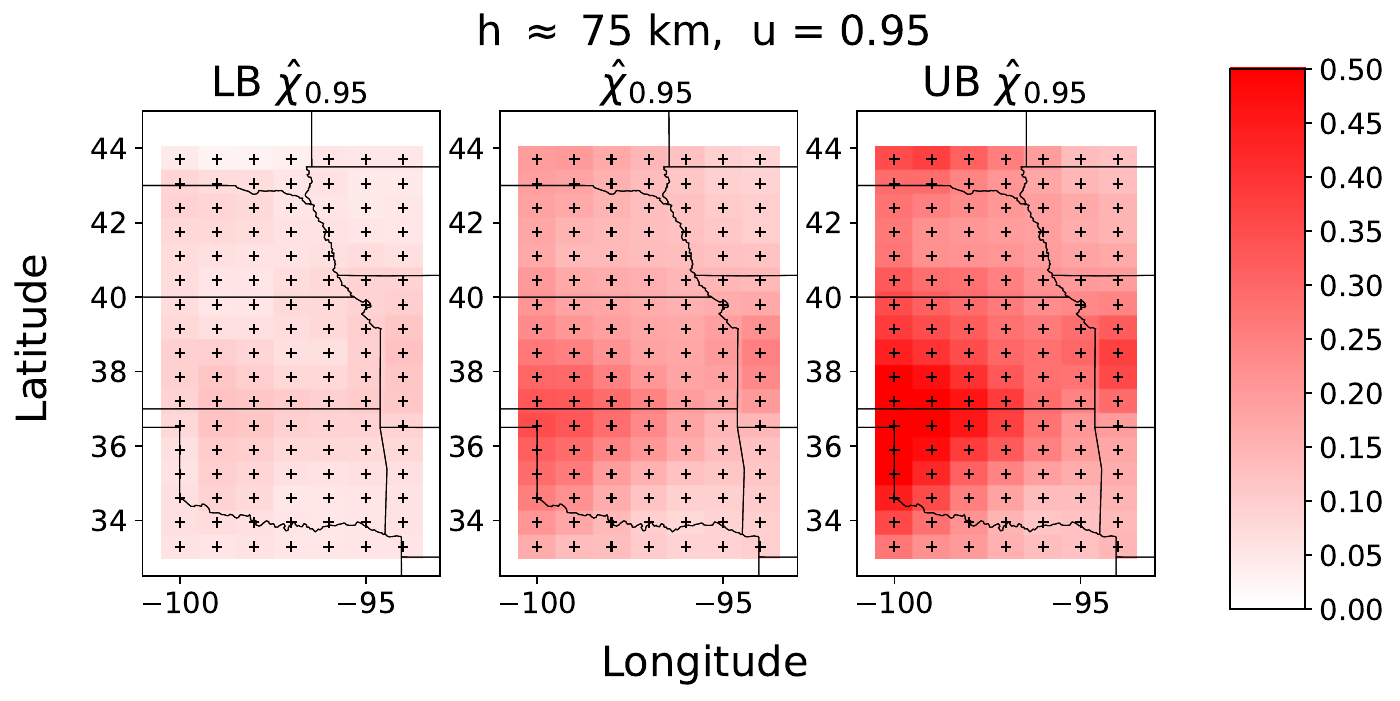}
\end{center}
\begin{center}
    \includegraphics[width=0.56\textwidth]{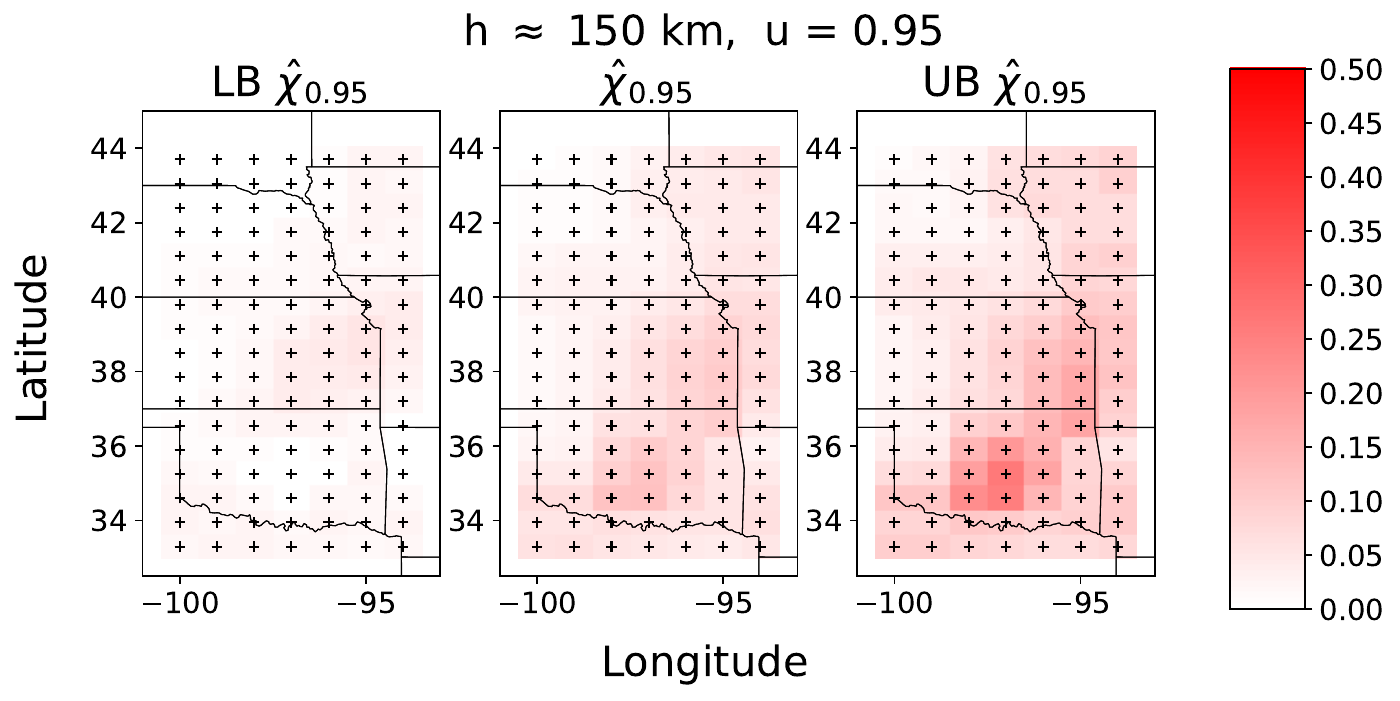}
    \includegraphics[width=0.56\textwidth]{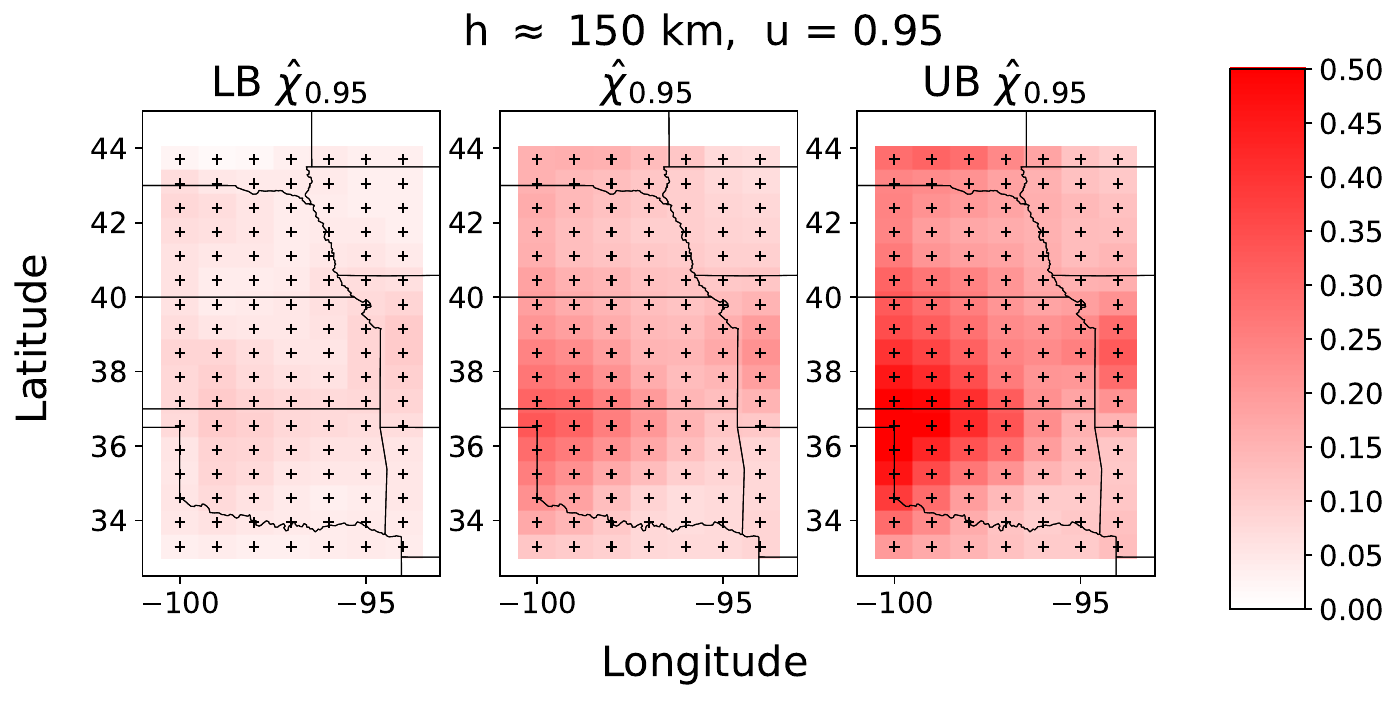}
\end{center}
\begin{center}
    \includegraphics[width=0.56\textwidth]{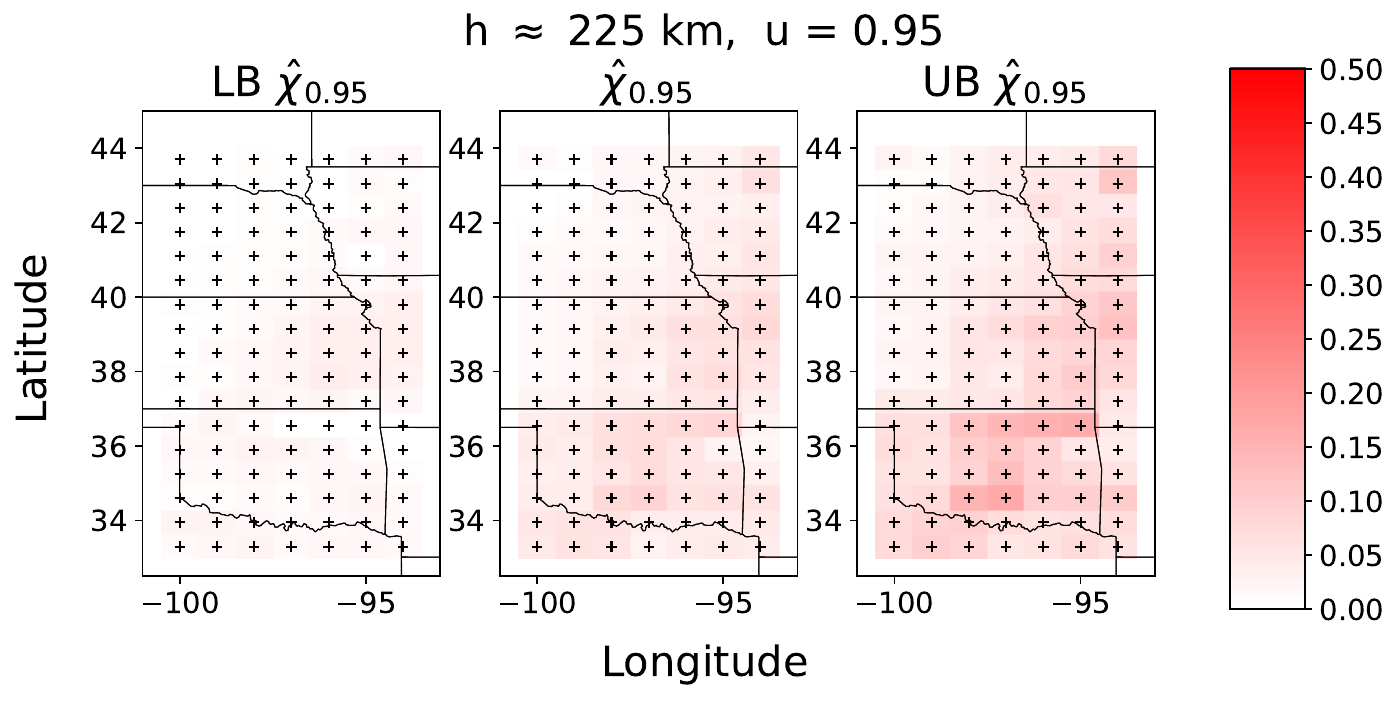}
    \includegraphics[width=0.56\textwidth]{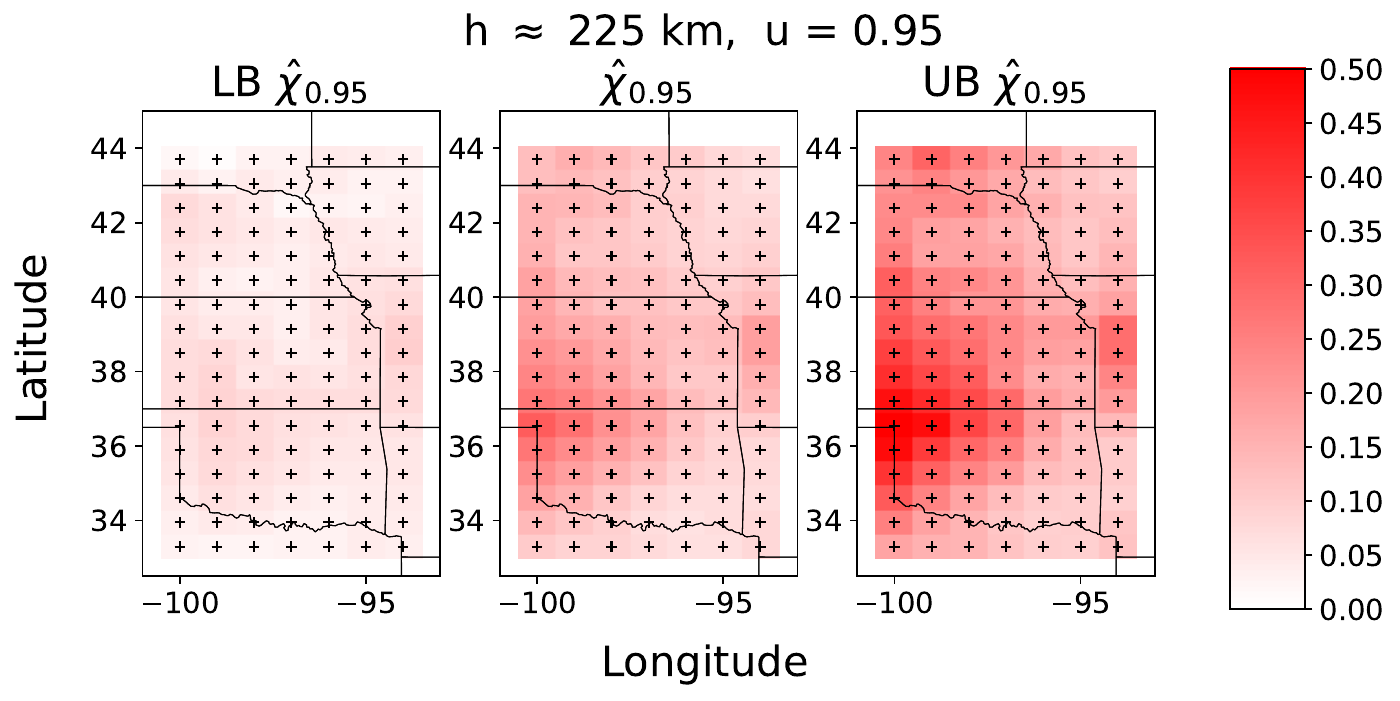}
\end{center}
\begin{center}
    \includegraphics[width=0.56\textwidth]{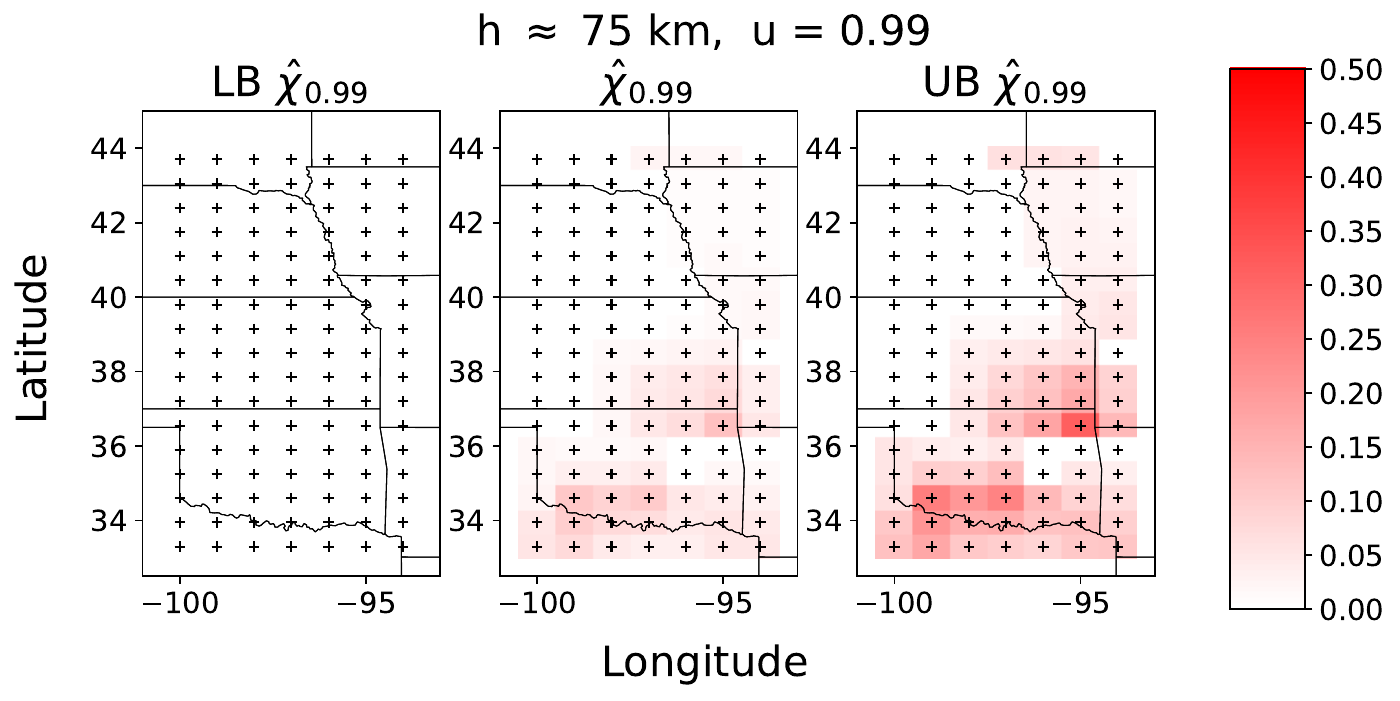}
    \includegraphics[width=0.56\textwidth]{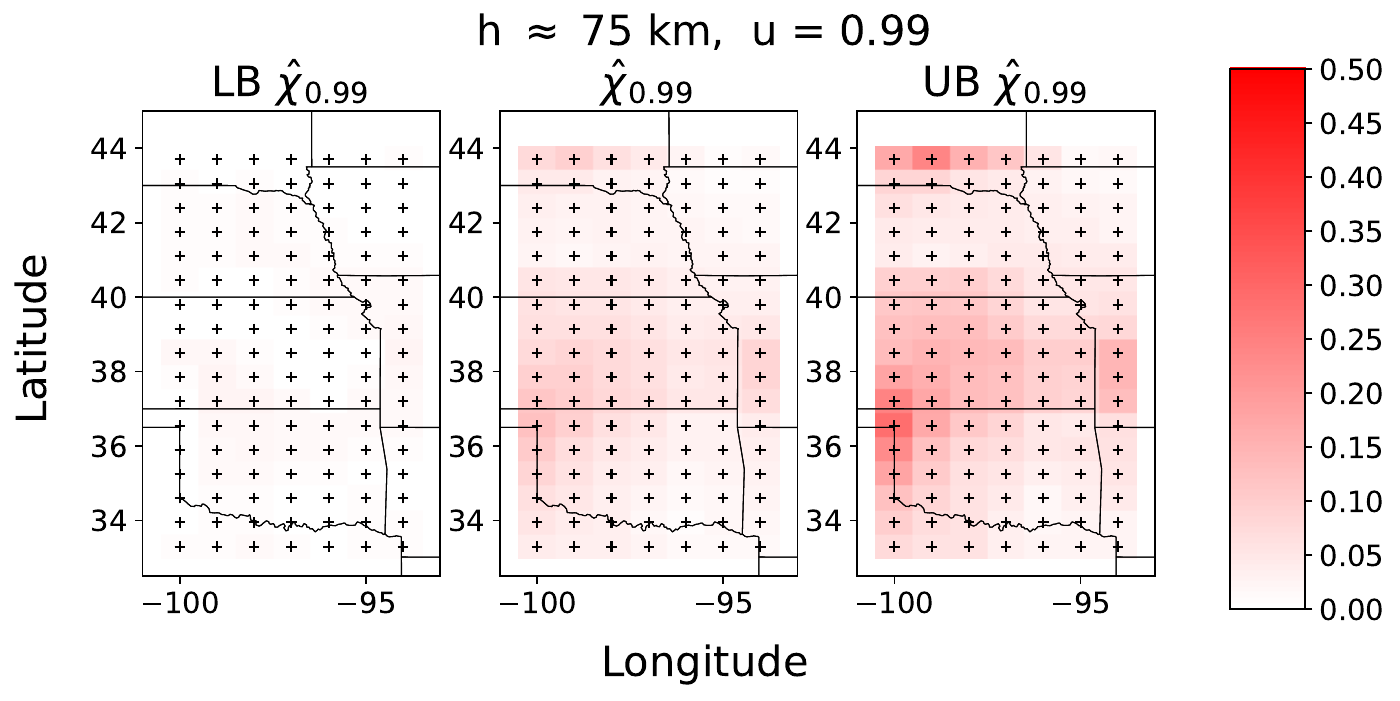}
\end{center}
\begin{center}
    \includegraphics[width=0.56\textwidth]{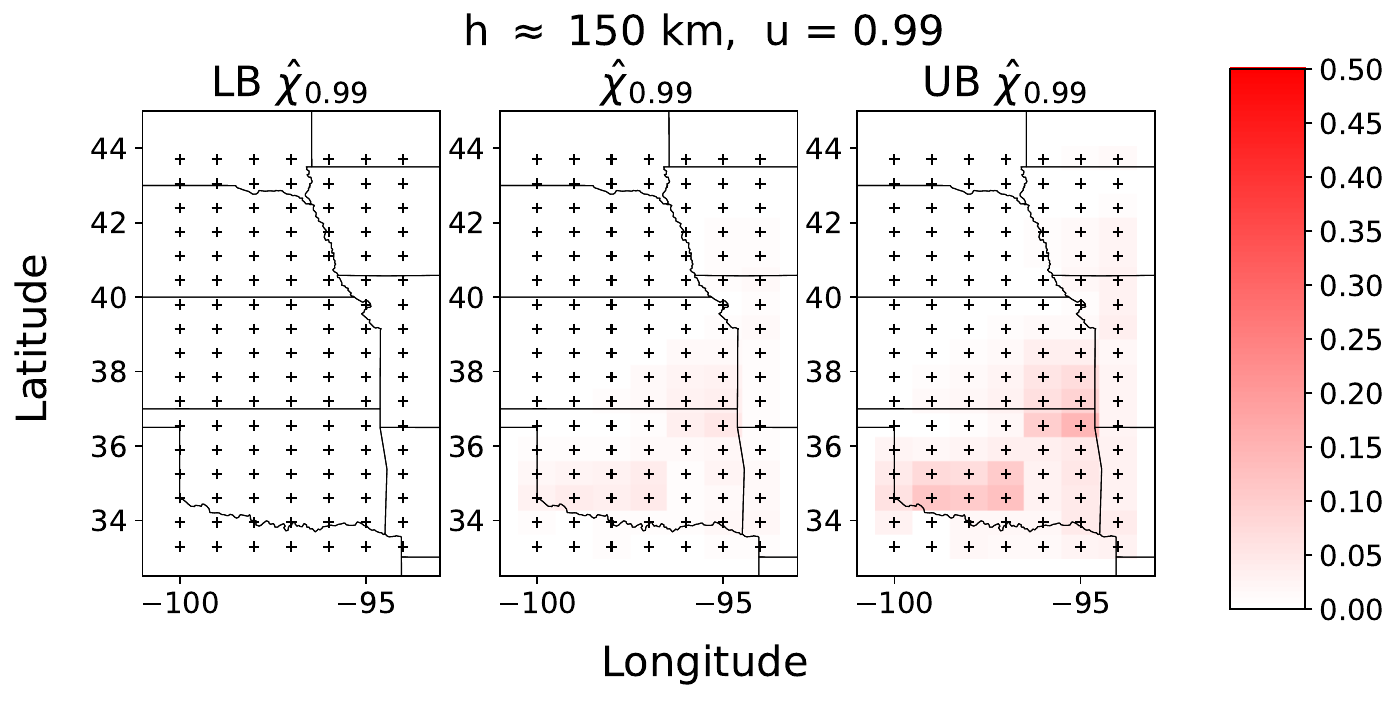}
    \includegraphics[width=0.56\textwidth]{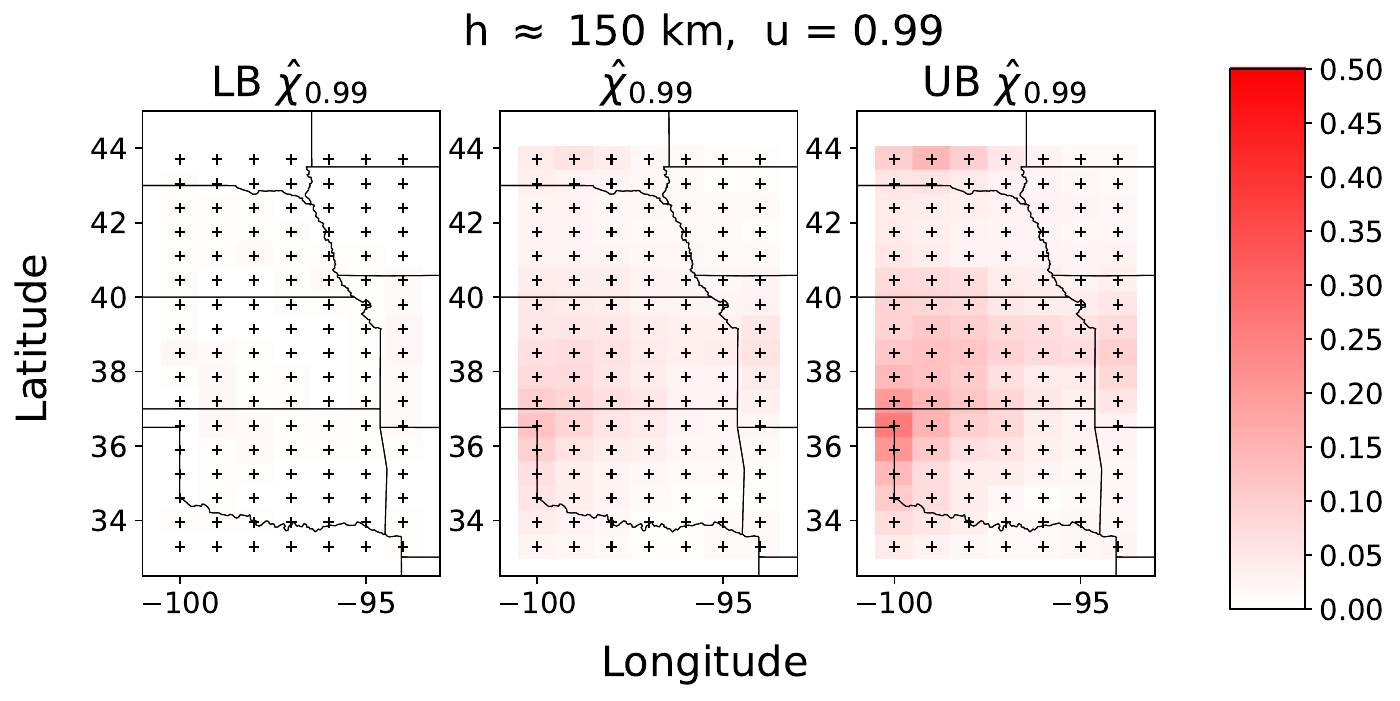}
\end{center}
\begin{center}
    \includegraphics[width=0.56\textwidth]{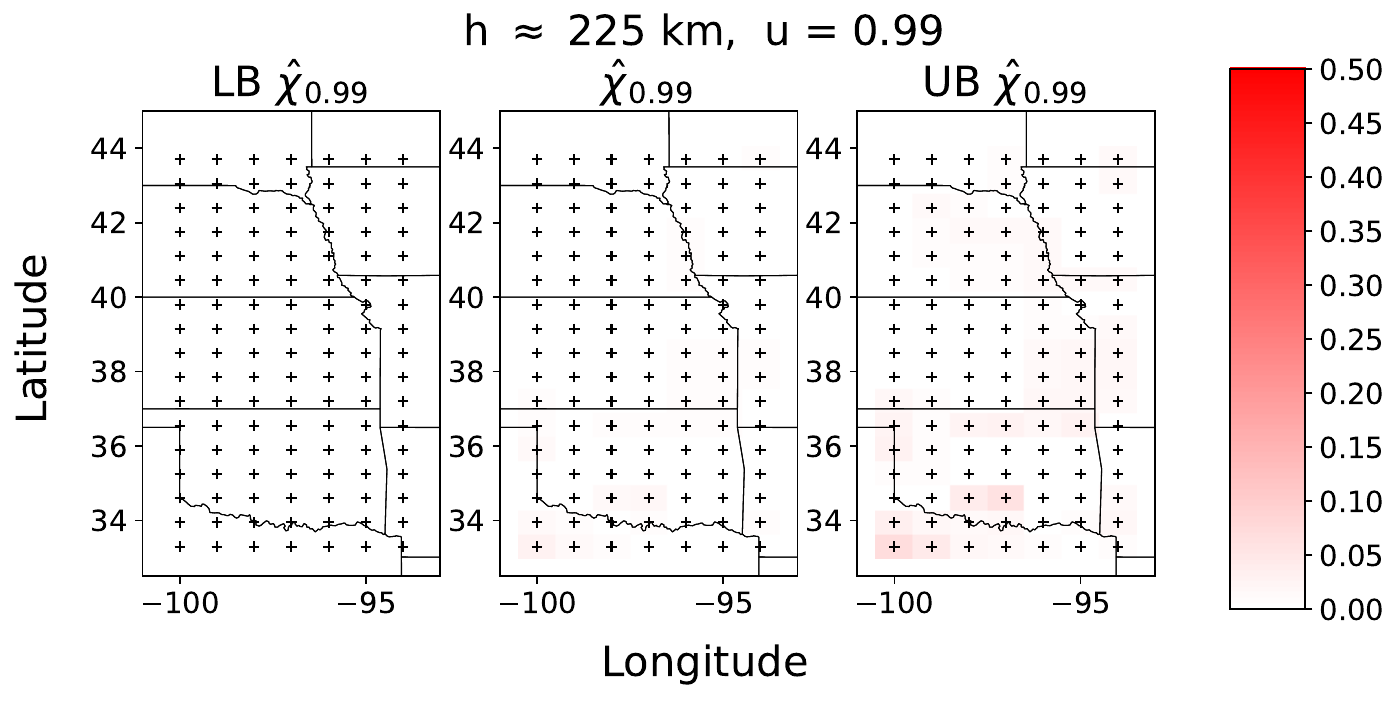}
    \includegraphics[width=0.56\textwidth]{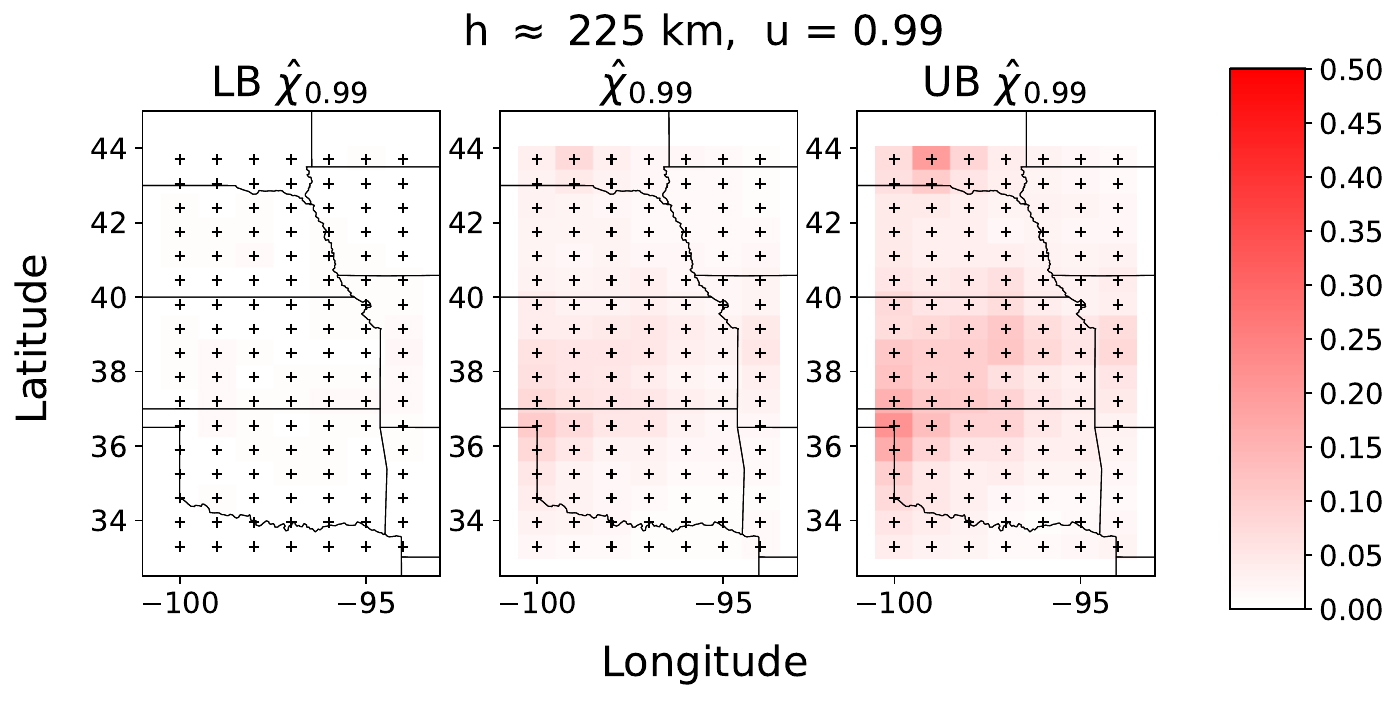}
\end{center}
\end{adjustwidth}

\end{document}